\documentclass[10pt,a4paper,onecolumn]{IEEEtran}
\IEEEoverridecommandlockouts
\usepackage{epsfig}
\input{epsf.sty}
\usepackage{epsfig}
\usepackage{latexsym}
\usepackage{bbm,dsfont}
\usepackage{amsmath}
\usepackage{amssymb}
\usepackage{array}
\usepackage{graphicx}
\usepackage{times}
\usepackage{setspace}
\usepackage{mathrsfs}
\usepackage{amsgen,amsfonts,amsbsy,amsthm}
\usepackage{txfonts}
\usepackage{textcomp}
\newcommand{\remark}[1]{
                        {\emph{Remark:}} #1}

\newcommand{\BEQA}{\begin{eqnarray}}
\newcommand{\EEQA}{\end{eqnarray}}
\newcommand{\define}{\stackrel{\triangle}{=}}

\newcommand{\step}{\emph{Step \emph}}
\newcommand{\gap}{\vspace{2mm}}
\newcommand{\comment}{\textbf{comment: }}
\newcommand{\ind}{\mathbbm{1}}
\newcommand{\F}{\mathcal{F}}
\newcommand{\G}{\mathcal{G}}
\renewcommand{\P}{\mathcal{P}}
\newcommand{\U}{\mathcal{U}}
\newcommand{\X}{\mathcal{X}}
\newcommand{\Xh}{\X_{h_{\max}}}
\newcommand{\Xep}{\X_{\epsilon,\delta}}
\newcommand{\Xepk}{\X_{\epsilon,\delta,k}}

\newcommand{\e}{\mathbf{e}}
\newcommand{\bo}{\mathbf{b}}

\newcommand{\beq}{\begin{equation}}
\newcommand{\eeq}{\end{equation}}

\newcommand{\bea}{\begin{eqnarray}}
\newcommand{\eea}{\end{eqnarray}}
\newcommand{\bean}{\begin{eqnarray*}}
\newcommand{\eean}{\end{eqnarray*}}

\newcommand{\bit}{\begin{itemize}}
\newcommand{\eit}{\end{itemize}}
\newcommand{\ben}{\begin{enumerate}}
\newcommand{\een}{\end{enumerate}}

\newtheorem{theorem}{Theorem}
\newtheorem{proposition}{Proposition}
\newtheorem{lemma}{Lemma}

\newtheorem{definition}{Definition}

\newtheorem{corollary}{Corollary}

\begin{document}
\title{QoS Aware and Survivable Network Design for Planned Wireless Sensor Networks} 
\author{Abhijit Bhattacharya and Anurag Kumar\\Dept.\ of Electrical Communication Engineering\\Indian Institute of Science, Bangalore, 560012, India\\ email: abhijit@ece.iisc.ernet.in, anurag@ece.iisc.ernet.in}

\maketitle

\begin{abstract} We study the problem of wireless sensor network design by deploying a minimum number of \emph{additional} relay nodes at a subset of \emph{given potential} relay locations, in order to convey the data from existing sensor nodes (hereafter called source nodes) to a Base Station(BS), while meeting a quality of service (QoS) objective specified as a hop count bound. The hop count bound suffices to ensure a certain probability of the data being delivered to the BS within a given maximum delay under the so-called ``lone packet'' traffic model. We study two variations of the problem. 

First, we study the problem of guaranteed QoS, connected network design, where the objective is to have at least one path from each source to the BS with the specified hop count bound. We observe that the problem is NP-Hard. For a problem in which the number of existing sensor nodes and potential relay locations is $n$, we propose an $O(n)$ approximation algorithm of polynomial time complexity. Results show that the algorithm performs efficiently in various randomly generated network scenarios; in over 90\% of the tested scenarios, it gave solutions that were either optimal or were worse than optimal by just one relay. Under a certain stochastic setting, we then obtain an upper bound on the average case approximation ratio of a class of algorithms (including the proposed algorithm) for this problem as a function of the number of source nodes, and the hop count bound. Experimental results show that the actual performance of the proposed algorithm is much better than the analytical upper bound. In carrying out this study of the algorithm, for small problems the optimal solutions are obtained by an exhaustive search, whereas for large problems we obtain a lower bound to the optimal value via an ILP formulation (involving so called ``node cut'' based inequalities) whose LP relaxation has a polynomial number of constraints (unlike usual path based formulation which has exponential number of constraints). 

Next, we study the problem of survivable network design with guaranteed QoS, i.e., the requirement is to have at least $k > 1$ node disjoint, hop constrained paths from each source to the BS. We observe that the problem is NP-Hard, and that the problem of finding a feasible solution to this optimization problem is NP-Complete. We propose a polynomial time heuristic for this problem. Finally, we study its performance on several randomly generated network scenarios, and provide an extensive analysis of these results. Similar in spirit to the one connectivity problem, we obtain, under a certain stochastic setting, an upper bound on the average case approximation ratio of a class of algorithms (including the proposed algorithm) for the hop constrained, survivable network design problem.
\end{abstract}

\section{Introduction}
\label{sec:intro}
 
Large industrial establishments such as refineries, power plants, and electric power distribution stations, typically have a large number of sensors distributed over distances of 100s of meters from the control center. Individual wires carry the sensor readings to the control center. Recently there has been increasing interest in replacing these wireline networks with wireless packet networks (\cite{honey,isa}). 
A similar problem arises in an intrusion detection application using a fence of passive infrared (PIR) sensors \cite{pir}, where the event sensed by several sensors has to be conveyed to a Base Station (BS) quickly and reliably. 

 The communication range of the sensing nodes is typically a few tens of meters (depending on the RF propagation characteristics of the deployment region). Therefore, usually multi-hop communication is needed to transmit the sensed data to the BS. The problem then is \emph{to design a multi-hop wireless mesh network with minimum deployment cost, i.e., minimum number of additional relays,} so as to communicate from each sensing (source) node to a central node, which we will call the BS (we shall use the terms BS and sink interchangebly), while meeting certain performance objectives such as a \emph{delay bound, and packet delivery probability.}

The relay placement problem can be broadly classified into two classes of problems. One is the \emph{unconstrained} relay placement problem, where the relay locations can be anywhere in the 2-dimensional region. In most practical applications, however, due to the presence of obstacles to radio propagation (e.g., a firewall, a large machine, or a building), or due to taboo regions (e.g., a pond or a ditch), we cannot place relay nodes anywhere in the region, but only at certain designated locations. This leads to the problem of \emph{constrained relay placement} in which the relays are constrained to be placed at certain \emph{potential relay locations} (see Figure~\ref{fig:constrained} for a depiction of the problem). In either of these problems, relays would have to be placed so that as few of them as possible are used while meeting performance objectives such as an upper bound on packet delivery delay, or a lower bound on packet delivery probability, or topological objectives such as the number of redundant paths.

\begin{figure}[t]
\begin{center}
\includegraphics[scale=0.3]{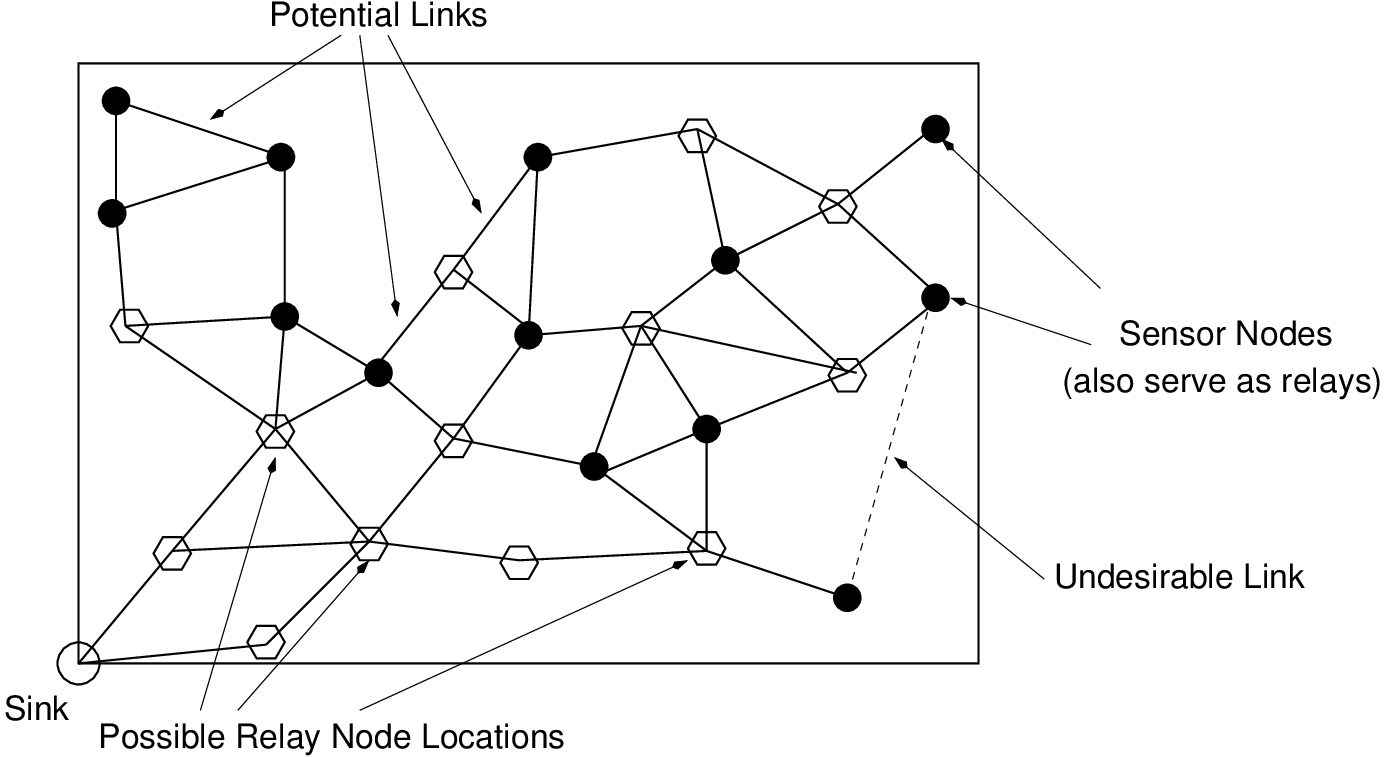}
\end{center}
\caption{The constrained relay placement problem; circles indicate sources, and the hexagons indicate potential relay locations. The edges denote the useful links between the nodes.}
\label{fig:constrained}
\end{figure}

As depicted in Figure~\ref{fig:constrained}, the source locations and the potential relay locations are specified. Only certain links are permitted; this could be because some links could be too long, leading to high bit error rate and hence large packet delay, or due to an obstacle, e.g., a firewall, or a building. The problem is to obtain a subnetwork that connects the source nodes to the base station with the requirement that
\begin{enumerate}
\item A minimum number of relay nodes is used.
\item There are at least $k$ node disjoint paths from each source node to the BS.
\item The maximum delay on any path is bounded by a given value $d_{\max}$, and the packet delivery probability (the probability of delivering a packet within the delay bound) on any path is $\geq p_{\mathrm{del}}$.
\end{enumerate}

To the best of our knowledge, this problem of QoS constrained, cost optimal network design has not yet been solved; we shall discuss the relevant literature in more detail in Section~\ref{subsec:lit-survey}. In this paper, we address this problem for the case in which (a) the nodes use the CSMA/CA Medium Access Control (as standardized in IEEE~802.15.4 \cite{IEEE}), and (b) the traffic from the source nodes is such that at any point of time only one measurement packet flows from a source in the network to the base station. We call this the ``lone packet traffic model'', which is realistic for many applications where the time between successive measurements being taken is sufficiently long so that the measurements can be staggered so as not to occupy the medium at the same time. For example, see Figure~\ref{fig:delaydist}, which depicts the cumulative distribution function (CDF) of end-to-end delay along a 5-hop path in a beaconless IEEE~802.15.4 network. The CDF has been obtained as a convolution of per hop delay distributions which are obtained using the backoff parameters given in the standard \cite{IEEE}. 

\begin{figure}[t]
\begin{center}
\includegraphics[scale=0.3]{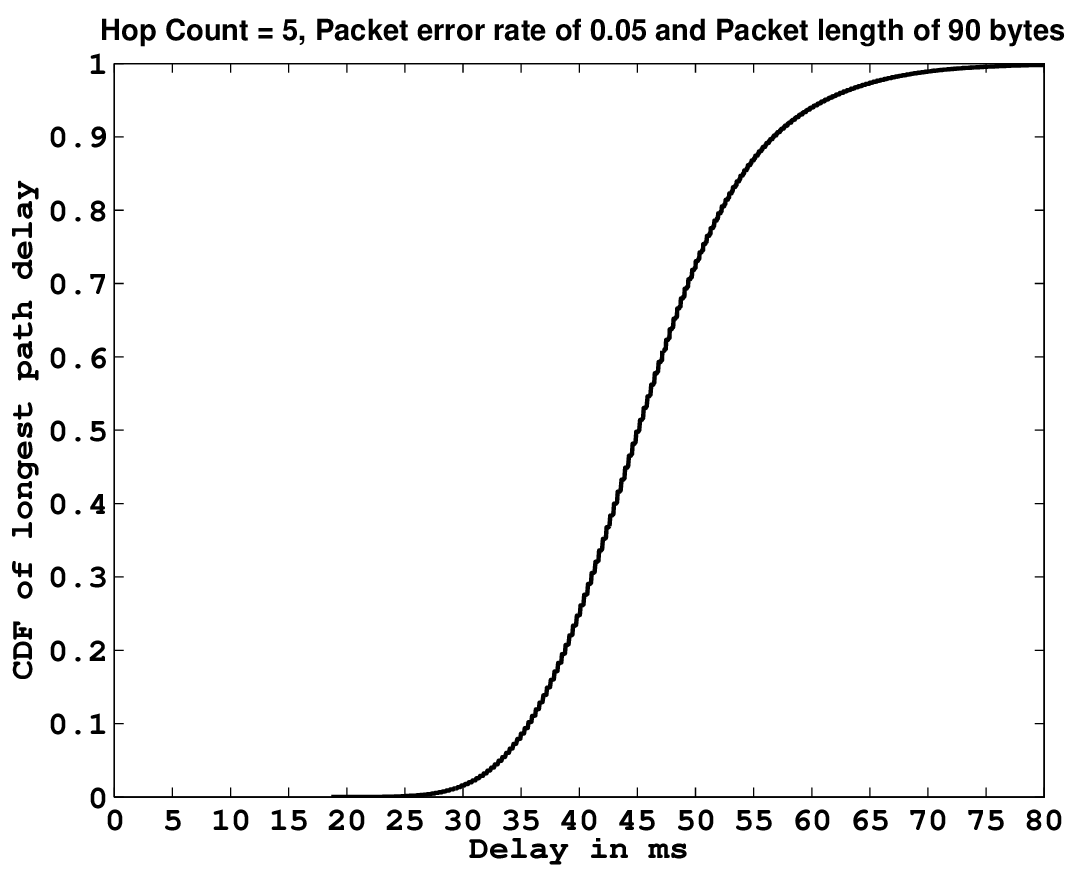}
\end{center}
\caption{CDF of end-to-end delay along a 5-hop path in a beaconless IEEE~802.15.4 network, assuming packet error rate of 0.05, and packet length of 90 bytes. The end-to-end delay does not include the fixed processing delay at each node.}
\label{fig:delaydist}
\end{figure}

Also, from Figure~\ref{fig:delaydist}, we see that the end-to-end delay is $\leq 69$ msec (without considering the node processing delay) with probability 0.99. The per hop processing delay was measured to be 15.48 msec \cite[p. 31, Section 2.2.5]{abhijitmeth}. Thus, the total end-to-end delay over the 5-hop path turns out to be $\leq 146.4$ msec with probability 0.99. This suggests that, even if such a network is designed based on the lone packet model, it will support a \emph{positive} aggregate packet arrival rate, while meeting a target delivery delay objective with a high probability. Indeed, our analytical modeling of such networks (see \cite{rachitpaper}) has shown that the arrival rate of a packet every few seconds (e.g., 5 to 10 seconds) from each source can be sustained. Such slow measurement rates are typical of so-called \emph{condition monitoring/industrial telemetry} applications \cite{telemetry1,telemetry2}.

Moreover, note that for a design (network) to satisfy the QoS objectives for a given positive arrival rate (continuous traffic), it is \emph{necessary} that the network satisfies the QoS objectives under zero/light traffic load, i.e., the ``lone packet'' model (for a more formal proof of this fact, see Section~\ref{sec:lone-packet-necessity}). As we shall see in subsequent sections, even under this simplified model of light traffic load, the problem of QoS constrained network design is computationally hard, and it does not seem to have been addressed before (we provide a literature survey in Section~\ref{subsec:lit-survey}). We cannot hope to solve the general problem of QoS aware network design for continuous traffic unless we have a reasonably good solution to the more basic problem of ``lone packet'' based network design, as well as a good analytical tool to model accurately, the stochastic interaction between the nodes in the network under traffic in which multiple links contend for the medium using CSMA/CA. A fast and accurate approximate performance analysis of multihop beaconless CSMA networks has been developed in \cite{rachitpaper}. In our current paper, we address the basic problem of QoS aware network design under ``lone packet'' model, as a step towards our future work of exploiting an analysis, such as the one in \cite{rachitpaper}, to augment the "lone packet" based design to one that meets the QoS for traffic in which several packets can occupy the network at the same time. Design of multihop beaconless CSMA networks under given positive packet arrival rates from the sources, by \emph{combining the ``lone packet'' based design with the analytical tool developed in \cite{rachitpaper}}, is a topic of our current research.
 
Note that even if the traffic is infrequent, the end user may still like to constrain the delay between when a measurement packet is generated and when the packet is received. Moreover, the links in a wireless sensor network are typically low power, \emph{lossy} links, where packet loss probabilities of 1\% to 5\% can be expected even on good links. Thus, the designs need to be constrained even for a lone packet model to achieve a satisfactory packet delivery probability. In applications, the measurements are currently conveyed to the BS via a wireline network. While replacing the wireline network (which is expensive to install and maintain) with a wireless mesh network, we aim to constrain the end-to-end performance achieved by the wireless network by imposing a hop count bound of $h_{\max}$ between each source and the BS. One possible approach of deriving such a hop count bound is presented in Section~\ref{subsec:motivation}.


Given a graph with feasible links between the potential locations and the source nodes, and a hop constraint, $h_{\max}$, the problem we address is to eliminate as many relays as possible from this graph so as to leave a graph with at least $k$ paths, each of at most $h_{\max}$ hops, between each source and the BS. We consider the case $k=1$ first, and then the case $k>1$. We provide a survey of related literature after a formal statement of the problem in Section~\ref{sec:formulation}. 

The rest of the paper is organized as follows: in Section~\ref{sec:formulation}, we describe the problem formulations for one connected, and $k$-connected hop constrained network design, show that the problems are NP-Hard, and present a brief survey of closely related literature. In Section~\ref{sec:sptalgo}, we propose a polynomial time algorithm (SPTiRP) for one connected network design, and provide a complete analysis of the algorithm. In particular, we provide a worst case approximation guarantee of the algorithm for arbitrary potential relay locations and source locations. We also derive a sufficient condition on the number and distribution of potential relay locations to ensure feasibility of the problem with high probability. Under such a stochastic setting, we provide an upper bound on the average case performance of the proposed algorithm. In Section~\ref{sec:ilp-oneconnect}, we propose a node-cut based ILP formulation for the one-connected hop constrained network design problem, whose LP relaxation has a polynomial number of constraints (unlike usual path based formulation which has exponential number of constraints). The LP relaxation can be useful in obtaining a lower bound on the optimal solution for problems of prohibitively large size, where an exhaustive enumeration of all possible solutions is impractical. In Section~\ref{sec:results}, we provide extensive numerical results for the SPTiRP algorithm applied to a set of random scenarios. Section~\ref{sec:simulation} provides packet level simulation results (using Qualnet, and assuming IEEE~802.15.4 CSMA/CA Medium Access Control) for the designs obtained using our algorithm to quantify the performance limits of the designs under ``positive'' traffic arrival rates. In Section~\ref{sec:formulation_kconnect}, we study the complexity involved in obtaining a feasible solution for the hop-constrained $k$-connectivity problem. In Section~\ref{sec:algorithms}, we propose a polynomial time algorithm (E-SPTiRP) for solving the $k$-connectivity problem, and provide analysis for the time complexity and approximation guarantee of the algorithm. In particular, for a subclass of problems with arbitrary potential locations and source locations, we provide a worst case approximation guarantee of the proposed algorithm. Similar in spirit to the one-connectivity problem, we then derive a sufficient condition on the number and distribution of potential locations to ensure feasibility of the hop constrained $k$-connectivity problem with high probability. Under such a stochastic setting, we obtain an upper bound on the average case performance of the E-SPTiRP algorithm. In Section~\ref{sec:results_kconnect}, we provide detailed numerical results for the algorithm applied to a set of random network scenarios. Finally, we conclude the paper in Section~\ref{sec:conclude}.

\section{Comparison of Lone-Packet Model and Positive-Flow Model}
\label{sec:lone-packet-necessity}
In developing algorithms for QoS constrained network design, it is fairly intuitive to assume that the performance of the network under the lone-packet model, i.e., where packets enter and leave the system one at a time, would be better than that under a positive-flow model, where there is a positive arrival rate at each source node, and packets can co-exist in the network, leading to contention. It is well known, however, that in CSMA/CA networks, in general, the performance is not monotone with the arrival rates (see, e.g., \cite{Singh}); hence, the preceding statement about the lone-packet traffic model needs to made with care. We provide here a simple proof based on a sample path argument. In doing so, we also make the notions of lone-packet and positive-flow models more formal.

\subsection{Lone Packet vs. Positive Arrival Rate: A Sample Path Argument}

Consider an arbitrary tree network with a single sink, where each source $q$, $1\leq q\leq m$, has a route, $c_q\:=\{q,v^q_1,\ldots,v^q_{(h_q-1)},0\}$, with hop count $h_q$ to the sink (the sink node is denoted by $0$). Let $p_{err}$ be the packet error rate on any link in the network. Note that even if we consider a different packet error rate for every link, the following argument will carry through with little modification. But to convey the basic concept, we are dealing with a simpler version here. Further, we assume that the nodes use the CSMA/CA MAC, as standardised by IEEE~802.15.4\cite{IEEE} (in fact, the argument holds for any MAC, with appropriate changes in the construction developed in the proof). 

We consider two different stochastic processes, namely, a lone packet process (corresponding to the lone packet traffic model), and a positive flow process (corresponding to a positive traffic arrival rate vector $\underline{\lambda}\in \Re^{m}$). Let $\Omega_1$ and $\Omega_2$ denote the sample spaces associated with the lone packet process, and the positive flow process respectively.  

We define, for all $\omega \in \Omega_1$,

\begin{equation*}
D_{0,k}^{(q)}(\omega)=\left\{
\begin{array}{rl}
1&\text{if the $k^{th}$ packet on route $c_q$ in the lone-packet model is delivered}\\
0 & \text{if the $k^{th}$ packet on route $c_q$ in the lone-packet model is not delivered}
\end{array}\right.
\end{equation*}

Similarly, define, for all $\omega\in \Omega_2$,

\begin{equation*}
D_{+,k}^{(q)}(\omega)=\left\{
\begin{array}{rl}
1&\text{if the $k^{th}$ packet on route $c_q$ in the positive-flow model is delivered}\\
0 & \text{if the $k^{th}$ packet on route $c_q$ in the positive-flow model is not delivered}
\end{array}\right.
\end{equation*}

Then, we can define the following quantities:

\begin{equation}
p_{del}^{q}(\underline{0}) = \lim_{k\rightarrow\infty}\frac{1}{K}\sum_{k=1}^K D_{0,k}^{(q)} 
\end{equation}

which is the long term fraction of packets delivered on route $c_q$, $1\leq q\leq m$, under lone-packet model.

And,
\begin{equation}
p_{del}^{q}(\underline{\lambda}) = \lim_{k\rightarrow\infty}\frac{1}{K}\sum_{k=1}^K D_{+,k}^{(q)}
\end{equation}

which is the long term fraction of packets delivered on route $c_q$, $1\leq q\leq m$, for a \emph{positive} arrival rate vector $\underline{\lambda}\in \Re^{m}$.

\begin{proposition}
\begin{equation}
p_{del}^{q}(\underline{0})\geq p_{del}^{q}(\underline{\lambda})\quad\forall \underline{\lambda}>\underline{0},\:\forall q\in\{1,\ldots,m\}
\end{equation}
\end{proposition}

\begin{proof}
We shall prove via a sample path argument. We aim to couple the two processes, namely, the lone packet process, and the positive flow process onto a common probability space as follows. 

Let $n$ be the maximum number of attempts of a packet by a transmitter before the packet is discarded ($n$ is a parameter of the underlying CSMA protocol). For each route $j$, $1\leq j\leq m$, let $\{\e_k^{(j)},\:k\geq 1\}$ be a realisation of $n\times h_j$ Bernoulli random sequences, each with success probability $1-p_{err}$; as usual, a $1$ in these sequences denotes that the corresponding packet attempt is a success. Using this sequence of i.i.d coin tosses, we can couple the two different stochastic processes, namely, the lone-packet process, and the positive flow process (corresponding to $\underline{\lambda}>\underline{0}$) onto a common sample space as follows. 

For each process, consider the $k^{th}$ packet arriving on route $c_j$; for this packet, use the element (which is a matrix) $\e_k^{(j)}$ of the above coin tossing sequence to decide whether the packet encounters a link error, or not, on any transmission attempt on any hop. Once the packet is delivered successfully, or discarded, we discard the rest of the matrix $\e_k^{(j)}$, and use the next element in the sequence for the next packet arriving on $c_j$. For example, if $\e_k^{(j)}(t,h)=1$, the $t^{th}$ transmission attempt (if it occurs) on the $h^{th}$ link of the $k^{th}$ packet on route $c_j$ is successful, and if $\e_k^{(j)}(t,h)=0$, the attempt encounters a link error. If the packet is delivered, or discarded before the $t^{th}$ attempt, or the $h^{th}$ hop, then $\e_k^{(j)}(t,h)$ is abandoned without use. 

Thus, the link errors seen by the $k^{th}$ packet, $k\geq 1$, on route $c_j$, $1\leq j\leq m$, are the same in both the processes. 

We define the two processes (the one with lone-packet traffic and the one with positive arrival rates) on the same sample space $\Omega$ with sample points $\omega$ constructed as follows:

\begin{align}
\omega&=\{(\{a_k^{(1)},k\geq 1\},\ldots,\{a_k^{(m)},k\geq 1\}),\nonumber\\
&(\{\mathbf{b}_k^{(1)},k\geq 1\},\ldots,\{\mathbf{b}_k^{(N)},k\geq 1\}),\nonumber\\
&(\{\mathbf{e}_k^{(1)},k\geq 1\},\ldots,\{\mathbf{e}_k^{(m)},k\geq 1\}),\{r_l,l\geq 1\}\}\nonumber
\end{align}

where,
\begin{description}
\item $a_k^{(i)}\in\Re_+$: $k^{th}$ interarrival time at source $i$
\item $\bo_k^{(j)}\in \Re^B_+$: vector of backoff durations of the $k^{th}$ packet arriving at the $j^{th}$ node, $1\leq j\leq N$; $B$ is the maximum number of CCA failures before a packet is discarded at a node
\item $\e_k^{(j)}\in \{0,1\}^{n\times h_j}$: link error indication matrix for $k^{th}$ packet arriving on route $c_j$
\item $r_l\in\{c_1,\ldots,c_m\}$: The source node (and hence the route) of the $l^{th}$ packet arriving in the network. Note that this is for the lone-packet model only. 
The interarrival time sequence determines the sequence of routes taken for the positive traffic model. 
\end{description}

Note that given such an instance $\omega\in\Omega$, we can construct the sample paths of both the lone-packet process and the positive-flow process. To construct the sample path of the lone-packet process, we need the components $(\{\mathbf{b}_k^{(1)},k\geq 1\},\ldots,\{\mathbf{b}_k^{(N)},k\geq 1\})$, 
$(\{\mathbf{e}_k^{(1)},k\geq 1\},\ldots,\{\mathbf{e}_k^{(m)},k\geq 1\})$, and $\{r_l,l\geq 1\}$ of $\omega$. In the lone-packet process, the departure of the $l^{th}$ packet triggers the arrival of the $(l+1)^{th}$ packet, which takes the route $r_{l + 1}$.


To construct the sample path of the positive-flow process, we need to use the components $(\{a_k^{(1)},k\geq 1\},\ldots,\{a_k^{(m)},k\geq 1\})$, $(\{\mathbf{b}_k^{(1)},k\geq 1\},\ldots,\{\mathbf{b}_k^{(N)},k\geq 1\})$, and
 $(\{\mathbf{e}_k^{(1)},k\geq 1\},\ldots,\{\mathbf{e}_k^{(m)},k\geq 1\})$ of $\omega$. In this process, packets enter the system at a source $j$, $1\leq j\leq m$, according to the sequence of interarrival times $\{a_k^{(j)},k\geq 1\}$. 

We consider the probability space $(\Omega,\F,P)$, where $\F$ is an appropriate $\sigma$-algebra, and the probability measure $P$ satisfies the additional property

\begin{equation*}
P\left\{\omega:\:\lim_{K\rightarrow\infty}\frac{1}{K}\sum_{l=1}^K\ind_{\{r_l=c_q\}}\:=\Psi_q>0\right\}=1\quad\forall q\in\{1,\ldots,m\}
\end{equation*}

which ensures that even in the lone-packet model, an infinite number of packets arrive to each source, so that each route is traversed infinitely often, with probability 1.

Note that the random variables $D_{0,k}^{(q)}$ and $D_{+,k}^{(q)}$ are defined on this common probability space in exactly the same way as they were defined earlier for the respective probability spaces of the two processes. 

Observe that $D_{0,k}^{(q)}(\omega)=0$ if and only if the link error indication matrix $\e_k^{(q)}$ has an all zero column (since in the lone-packet system, packets can be lost only due to link errors, and all the $n$ transmission attempts over a link on the packet's route must fail before the packet is discarded). 
Since the same link error indication matrix, $\e_k^{(q)}$, is used for the $k^{th}$ packet on route $c_q$ in the positive-flow process, it follows that if the $k^{th}$ packet on route $c_q$ in the lone-packet model is discarded, the $k^{th}$ packet on route $c_q$ in the positive-flow model will also be discarded. Moreover, there will be additional discards in the positive-flow model due to CCA failures, and collisions, which are absent in the lone-packet model. It follows, therefore, that for all $K$, $\omega\in \Omega$,

\begin{align}
\sum_{k=1}^K D_{+,k}^{(q)}(\omega)\leq \sum_{k=1}^K D_{0,k}^{(q)}(\omega)\nonumber
\end{align}

Hence,

\begin{equation*}
P\left\{\lim_{k\rightarrow\infty}\frac{1}{K}\sum_{k=1}^K D_{+,k}^{(q)}\leq \lim_{k\rightarrow\infty}\frac{1}{K}\sum_{k=1}^K D_{0,k}^{(q)}\right\}=1
\end{equation*}

Recall that by definition, 

\begin{align}
p_{del}^{q}(\underline{0}) &= \lim_{k\rightarrow\infty}\frac{1}{K}\sum_{k=1}^K D_{0,k}^{(q)}\:a.s.\nonumber\\
p_{del}^{q}(\underline{\lambda}) = \lim_{k\rightarrow\infty}\frac{1}{K}\sum_{k=1}^K D_{+,k}^{(q)}\:a.s.\nonumber
\end{align}

Thus, we have shown that $p_{del}^{q}(\underline{0})\geq p_{del}^{q}(\underline{\lambda})$, for any $\underline{\lambda}>0$. 
\end{proof}

Moreover, note that

\begin{equation*}
P\left\{\lim_{k\rightarrow\infty}\frac{1}{K}\sum_{k=1}^K D_{0,k}^{(q)}=(1-p_{err}^n)^{h_q}\right\}=1, \quad \text{by strong law of large numbers}
\end{equation*}

Hence, $p_{del}^{q}(\underline{\lambda})\leq (1-p_{err}^n)^{h_q}$.

\section{Design constraints to ensure end-to-end performance objectives}
\label{subsec:motivation}

\subsection{Assumptions}
In several industrial telemetry applications, the rate at which measurements are obtained from the sensors is low, for example, as little as one reading per hour from each sensor. We also assume that the alarm traffic is so infrequent that it does not interfere with any regular data transmission. Then, if the data transmission from the sensors is staggered over the hour, it can be assumed that each measurement packet flows over the network with no interference from any other packet flow. Our work in this paper is concerned with this ``lone packet'' traffic model. We also assume that IEEE~802.15.4 standard \cite[p. 30-179, p. 640-643]{IEEE} is used for PHY and MAC layers. 

We can obtain the bit error rate, $\epsilon$, on a link as a function (which depends on the modulation scheme) of received Signal to Noise Ratio (SNR), by using a formula given in the standard. Then, for a Physical layer (PHY) packet data unit length of $L$ bytes, the packet error rate (PER) on a link can be obtained as $1-(1-\epsilon)^L$. Given the PER $q$ on a link as a function of received SNR, we can obtain an expression for $D_q(\cdot)$, the c.d.f.\ of packet delay on the link (given that the packet is not dropped), using the backoff behavior and parameters of IEEE~802.15.4 CSMA/CA MAC. We can also obtain the packet drop probability as a function of the link PER, and we denote this function by $\delta(\cdot)$.

We also permit \emph{slow fading} of links; so the link PER can vary slowly over time, thus leading to the concept of ``link outage''. We say, a link is in outage if the PER of the link exceeds a target maximum link PER, designated by $q_{\max}$ (obtained as a function of the target SNR, $\gamma_{\min}$). Let us denote by $p_{\mathrm{out}}$, the maximum probability of a link being in outage.  

Before proceeding further, we summarize for our convenience, the notations used in the development of the model.

\vspace{2mm}
\textbf{User requirements:}
\begin{description}
\item[$L$] The longest distance from a source to the base-station (in meters)
\item[$k$] The required number of node disjoint paths between each source and
  the base-station
\item[${d}_{\max}$] The maximum acceptable end-to-end delay of a
  packet sent by a source (packet length is assumed to be fixed and
  given)
\item[$p_{\mathrm{del}}$] Packet delivery probability: the probability
  that a packet is not dropped \emph{and} meets the delay bound
  (assuming that at least one path is available from each source to
  the base station).
\end{description}
\vspace{2mm}
\textbf{Parameters obtained from the standard:}
\begin{description}
\item[$D_q(\cdot)$] The cumulative distribution function of packet
  delay on a link with PER $q$, given that the packet is not dropped;
  $D^{(h)}_q(\cdot)$ denotes the $h$-fold convolution of
  $D_q(\cdot)$. Under the lone packet model, $D_q(\cdot)$ is obtained
  by a simple analysis of the backoff and attempt process at a node,
  as defined in the IEEE~802.15.4 standard for beaconless mesh
  networks. 
\item[$b(\cdot)$] The mapping from SNR to link BER for the modulation scheme
\item[$\delta(\cdot)$] The mapping from PER to packet drop probability
  over a link. Note that even when there is no contention,
  packets could be lost due to random channel errors on links (i.e.,
  non-zero link PER). A failed packet transmission is reattempted at
  most three times before being dropped.
\end{description}
\vspace{2mm}
\textbf{Design parameters:}
\begin{description}
\item[$P_{\mathrm{xmt}}$] The transmit power over a link (assumed here
  to be the same for all nodes)
\item[$\gamma_{\min}$] The target SNR on a link
\item[$q_{\max}$] The target maximum PER on a link
\end{description}
\vspace{2mm}
\textbf{Parameters obtained by making field measurements:}
\begin{description}
\item[$r_{\max}$] The maximum allowed length of a link on the field to meet the target SNR, and outage probability requirements
\item[$p_{\mathrm{out}}$] The maximum probability of a link SNR falling below
  $\gamma_{\min}$ due to temporal variations. A link is ``bad'' if its
  outage probability is worse than $p_{\mathrm{out}}$, and ``good'', otherwise
\end{description}

\vspace{2mm}
\textbf{To be derived:}
\begin{description}
\item [$h_{\max}$] The hop count bound on each path, required to meet the packet delivery objectives
\end{description} 

\noindent
\textbf{Remark:} In practice, the value $k$ can be chosen so that a network monitoring and repair process ensures that a path is available from each source to the BS at all times. The choice of $k$ is not in the scope of our formulation, and would depend on how quickly the network monitoring process can detect node failures, and how rapidly the network can be repaired. We, thus, assume that, whenever a packet needs to be delivered from  a source to the BS, there is a path available, and, by appropriate choice of the path parameters (the length of each link, and the number of hops), we ensure the delivery probability, $p_{\mathrm{del}}$.

\subsection{Design Constraints from Packet Delivery Objectives}

Consider, in the final design, a path between a source $i$ and the
base-station, which is $L_i$ meters away.  Suppose that this path has
$h_i$ hops, and the length of the $j$th hop on this path is $r_{i,j},
1 \leq j \leq h_i$. Then we can write
\begin{eqnarray}
\label{eqn:hopcount-bound-from-rmax}
  L_i \leq \sum_{j=1}^{h_i} r_{i,j} \leq h_i r_{\max}
\end{eqnarray}
where the first inequality derives from the triangle inequality, and
the second inequality is obvious. Since $L$ is the farthest that any
source is from the base station, we can conclude that the number of
hops on any path from a source to a sink is bounded below by
$\frac{L}{r_{\max}}$. 

Following a conservative approach, we take the
PER on every link to be $q_{\max}$ (we are taking the worst case PER
on each link, and are not accounting for a lower PER on a shorter
link) . 

Suppose that we have obtained a network in which there are $k$ node
independent paths from each source to the base-station, and all the
links on these paths are good (``good'' in the sense explained earlier in the definition of $p_{\mathrm{out}}$).
Consider a packet arriving at Source~$i$, for which, by design, there are $k$
paths, with hop counts $h_\ell, 1 \leq \ell \leq k$, and suppose that
at least one of these paths is available (i.e., all the nodes along
that path are functioning). The availability of such a path will be
determined by a separate route management algorithm, which is out of the scope of this paper. We select one of these good paths to route the packet. The path selection algorithm would incorporate a load and energy balancing strategy. If the chosen path has $h$ hops in it, then the probability that none of the edges
along the chosen path is in outage is given by
\begin{eqnarray*}
 (1 - p_{\mathrm{out}})^{h_\ell}
\end{eqnarray*}
Increasing $h_{\ell}$ makes this probability smaller. With this in
mind, let us seek an $h_{\max}$, by the following conservative
approach. First, we lower bound the probability of the chosen path not
being in outage by
\begin{eqnarray*}
   (1 - p_{\mathrm{out}})^{h_{\max}}
\end{eqnarray*}
Now we can ensure that the packet delivery constraint is
met by requiring
\begin{eqnarray}
  \label{eqn:using-the-pdel-constraint}
  (1 - p_{\mathrm{out}})^{h_{\max}}(1 - \delta(q_{\max}))^{h_{\max}} D^{(h_{\max})}_{q_{\max}} (d_{\max})\geq p_{\mathrm{del}}
\end{eqnarray}
where the additional terms lower bound the probability that the packet
is not dropped along the chosen path ($(1 -
\delta(q_{\max}))^{h_{\max}} $) and that the end-to-end delay is less
than or equal to $d_{\max}$ ($ D^{(h_{\max})}_{q_{\max}}
(d_{\max})$). Recall that we take the PER on each ``good'' link to be $q_{\max}$. The left hand expression
in \eqref{eqn:using-the-pdel-constraint} is decreasing as $h_{\max}$
increases; \emph{let $h_{\max}$ be the largest value so that the inequality
is met}. Thus, we can meet the end-to-end performance objectives by imposing a hop count constraint $h_{\max}$ from each source to the BS.

Also, combining \eqref{eqn:hopcount-bound-from-rmax} and the
$h_{\max}$ just obtained, we get, for every source $i$
\begin{eqnarray}
\label{eqn:rmax_lowerbnd}
  r_{\max} \geq \frac{L}{h_{\max}} 
\end{eqnarray}
 
Hence, under a given physical setting, we can convert the problem of network design with end-to-end delay bound, and guaranteed packet delivery probability to a problem of network design with end-to-end hop constraint on each path. 
 
\section{The Network Design Problems}
\label{sec:formulation}
 
\subsection{The Network Design Setting}
\label{subsec:network-setting}
Given a set of source nodes or required vertices $Q$ (including the BS) and a set of potential relay locations $R$ (also called Steiner vertices), we consider a graph $G= (V, E)$ on $V= Q\cup R$ with $E$ consisting of all \emph{feasible} edges. We assume that CSMA/CA as defined in IEEE~802.15.4 \cite{IEEE} is used for multiple access. 

Note that there are several ways in which we can obtain the graph $G$ above (i.e., the set of \emph{feasible} edges $E$), keeping in mind the end-to-end QoS objective. For example, we can impose a bound on the packet error rate (PER) of each link, or alternately, we can constrain the maximum allowed link length (which, in turn, affects the link PER). As shown in Section~\ref{subsec:motivation}, having characterized the link quality of each feasible link in the graph $G$, the QoS objectives ($d_{\max}$ and $p_{\mathrm{del}}$) can be met by imposing a hop count bound of $h_{\max}$ between each source node and the sink. 

We would like to point out that the graph design algorithms presented in this paper are in no way tied to the link modeling approach mentioned above for obtaining the graph $G$, and the hop constraint $h_{\max}$. The algorithms can be applied as long as a graph on $Q\cup R$, and a hop constraint is given, irrespective of how the graph and the hop constraint were obtained. The link modeling approach is just a convenient, and not necessarily unique, way of converting the QoS objective into a graph design objective. 

\subsection{Problem Formulation}
\label{subsec:formulation}

\subsubsection{One Connected Network Design Problem}
Given the graph $G= (V, E)$ on $V= Q\cup R$ with $E$ consisting of all \emph{feasible} edges (as explained in Section~\ref{subsec:network-setting}), and a hop constraint $h_{\max}$, \emph{the problem is to extract from this graph, a spanning tree on $Q$, rooted at the BS, using a minimum number of relays such that the hop count from each source to the BS is $\leq h_{max}$.} We call this the \emph{Rooted Steiner Tree-Minimum Relays-Hop Constraint} (RST-MR-HC) problem.
\noindent
\subsubsection{$k$-Connected Network Design Problem}

The requirement is to have at least $k$ node disjoint and hop constrained paths from each source to the sink. Then, we can formulate our relay placement problem as follows:

\emph{Given the graph $G= (V, E)$ on $V= Q\cup R$ with $E$ consisting of all \emph{feasible} edges, the problem is to extract from this graph, a subgraph spanning $Q$, rooted at the BS, using a minimum number of relays such that each source has at least $k$ node disjoint paths to the sink, and the hop count from each source to the BS on each path is $\leq h_{max}$.} We call this the \emph{Rooted Steiner Network-$k$ Connectivity-Minimum Relays-Hop Constraint} (RSN$k$-MR-HC) problem.

\subsection{Complexity of the Problems}

\begin{proposition}
\begin{enumerate}
\item The RST-MR-HC problem is NP-Hard.
\item The RSN$k$-MR-HC problem is NP-Hard.
\end{enumerate}

\end{proposition}
\begin{proof}

\begin{enumerate}
\item The subset of RST-MR-HC problems where the hop count bound is trivially satisfied is precisely the class of RST-MR \cite{Misra} problems (consider, for example, all RST-MR-HC problems where $|Q|+|R|= n$, $n$ being some positive integer, and the hop count bound is $h_{max}= n-1$. Clearly, the hop count bound is trivially satisfied in these problems). Thus, the RST-MR problem is a subclass of the RST-MR-HC problem. But, the RST-MR problem is NP-Hard (see \cite{Misra}). Hence, the RST-MR-HC problem, being a superclass of the RST-MR problem, is also NP-Hard \cite[p. 63, Section 3.2.1]{Garey}. 

\item We have just proved that the problem is NP-Hard even for $k=$1, since that is just the RST-MR-HC problem. Therefore the general problem is also NP-Hard, using the ``restriction'' argument \cite[p. 63, Section 3.2.1]{Garey}.
\end{enumerate}
\end{proof}

\subsection{Related Literature}
\label{subsec:lit-survey}
We see that the problem we have chosen to address belongs, broadly, to the class of Steiner Tree Problems (STP) on graphs (\cite{st2,st3}).

The classical STP is stated as: \emph{given an undirected graph $G= (V, E)$,
with a non-negative weight associated with each edge, and a set of
\emph{required vertices} $Q\subseteq V$, find a minimum total edge
cost subgraph of $G$ that spans $Q$, and may include vertices from the
set $S:= V-Q$, called the \emph{Steiner vertices}.}\label{stp}

The classical STP dates back to Gauss and it has been proven to be
NP-Hard. Lin and Xue \cite{Lin} proposed the Steiner Tree Problem with
Minimum Number of Steiner Points and Bounded Edge Length (STP-MSPBEL).
The STP-MSPBEL was stated as: given a set of $n$ terminal points $Q$
in 2-dimensional Euclidean plane, find a tree spanning $Q$, and some additional Steiner points such that each edge has length no more than $R$, and the number of Steiner points is minimized. This bound on edge length only constrains link quality, but not end to end QoS. The problem was shown to be NP-complete and a polynomial time 5-approximation algorithm was presented. This problem was the first well-studied problem on optimal relay placement (relay locations \emph{unconstrained}). However, no average case performance guarantee was provided for the proposed algorithm. 

Cheng et al. \cite{cheng} studied the same problem as Lin and Xue, and proposed a 3-approximation algorithm and a 2.5-approximation
algorithm. 

Lloyd and Xue \cite{Lloyd} studied a generalization of
STP-MSPBEL problem where each sensor node has range $r$ and each relay
node has range $R\geq r$.  They provided a 7-approximation polynomial
time algorithm. They also studied the problem of minimum number of
relay placement such that there exists a path consisting solely of
relay nodes between each pair of sensors. For this problem, they
provided a $(5 + \epsilon)$-approximation algorithm. The problems studied by Lloyd and Xue, as well as Cheng et al.\ fall in the category of \emph{unconstrained} relay placement problem. Neither work provide any average case performance guarantee of their proposed algorithms.

Voss \cite{Voss} studied the Steiner Tree Problem with
Hop Constraints (STPH). This problem is stated as:\emph{ given a directed
connected graph $G= (V, E)$, with non-negative weight associated with
each edge, consider a subset of $V$, namely, $Q= \{0, 1, 2,\ldots, n\}$ 
with 0 being the root vertex, and a positive integer $H$. The problem
is to find a minimum total edge cost subgraph $T$ of $G$ such that
there exists a path in $T$ from 0 to each vertex in $Q\backslash\{0\}$ not exceeding $H$ arcs (possibly including vertices from $S:= V-Q$)}. We can call this problem the \emph{Rooted Steiner Tree-Minimum Weight-Hop Constraint} problem (RST-MW-HC). This problem was shown to be NP-Hard, and a Minimal Spanning Tree  based heuristic algorithm was proposed to obtain a good quality feasible solution, followed by an improvement procedure using a variation of \emph{Local Search} method called the \emph{Tabu search} heuristic. \emph{No performance guarantee or complexity analysis of the heuristic was provided}. Also, the tabu search heuristic may not be polynomial time. 

Note that an instance of the RST-MR-HC problem can be converted to an instance of the RST-MW-HC problem in polynomial time as follows: replace each relay with a directed edge of weight 1, and replace each edge associated with the relay with two directed edges (each of weight 0), one incident into the tail of the edge substituting the relay, and one going out of the tip of the edge substituting the relay. Then, minimizing the number of relays in the original problem is equivalent to minimizing the total weight in the converted problem. 
Then, one could use Voss's algorithm on this instance of RST-MW-HC problem to solve the original problem. But, as we mentioned earlier, Voss's algorithm does not provide any performance guarantee, and because of the tabu search heuristic (which may not be polynomial time), it may take long to converge to a solution.  

Costa et al. \cite{Costa} studied the Steiner Tree Problem with revenue, budget, and hop constraints. Given a graph $G= (V, E)$, with a cost associated with each edge, and a \emph{non-negative revenue} associated with each vertex, the problem is to determine a revenue maximizing tree subject to a total edge cost constraint, and a hop constraint between the root vertex and every other vertex in the tree. They propose a greedy algorithm for initial solution followed by destroy-and-repair or tabu search to improve the initial solution. They have evaluated the performance of the proposed algorithms only through numerical experiments; \emph{no theoretical guarantee has been provided}. 

It is possible to cast our problem into the form of the one addressed by Costa et al.~\cite{Costa} as follows: assign a \emph{negative} revenue, say $-1$, to each relay node (Steiner vertex), and a large positive revenue, say $|R|+1$, where $|R|$ is the number of potential relay locations, to each source vertex. This cost assignment would ensure that a revenue maximizing tree has all the source vertices in it, since the gain in revenue by adding a source outweighs the loss in revenue due to the additional relays, if any, required to connect the source to the BS. Also, the negative revenue on relays ensures that the revenue maximizing tree contains in it, as few relays as possible. Now, choose the hop constraint to be the same as that in the original RST-MR-HC problem. Also, assign a cost of zero to each edge, and  choose a trivial total edge cost constraint (any positive real number). With these assignments/choices, the problem of minimizing total relay count while obtaining a hop constrained tree network (RST-MR-HC) is the same as the problem of obtaining a  revenue-maximizing Steiner tree subject to a hop constraint and a total edge cost constraint. This formulation, however, requires the node weights to be negative, whereas the algorithm proposed by Costa et al. \emph{requires} the nonnegativity of the node weights \footnote{The greedy algorithm that they proposed starts with the root node, and proceeds by adding a path connecting a non-selected \emph{profitable} vertex to the existing solution at each step; when the budget constraint can be trivially satisfied, this amounts to simply finding a hop constrained path from a profitable vertex to the root. This is not enough to ensure revenue maximization if the revenues associated with some of the nodes is negative, since the path selected from the profitable vertex to the root may contain vetices with negative revenue, thus reducing the profit along the way; thus, additional constraints must be imposed for selection of paths from the profitable vertices to the root node to ensure minimal usage of the negative-revenue vertices.}. Moreover, even if one could find a way to map the RST-MR-HC problem to the revenue-budget-hop constrained STP, the tabu search based heuristic proposed by Costa et al. to improve the initial solution to the revenue-budget-hop constraint problem is not guaranteed to be polynomial time in general, and may take a long time to converge. 

Kim et al. \cite{Kim} studied the Delay and Delay Variation Constrained multicastng Steiner Tree Problem. The problem is similar to the one studied by Voss, with a delay constraint instead of the hop constraint, and a constraint on delay variation between two sources. With the delay variation constraint relaxed, Kim's problem becomes the \emph{Rooted Steiner Tree-Minimum Weight-Delay Constraint} problem. They proposed a polynomial time heuristic algorithm to obtain feasible solutions, but they also did not provide any performance guarantee for their algorithm.

Bredin et al.\ \cite{bredin} studied the problem of optimal relay placement \emph{(unconstrained)} for $k-$connectivity. They proposed an $O(1)$ approximation algorithm for the problem with any \emph{fixed} $k\geq 1$. However, they did not provide any average case analysis for their algorithm.

 Misra et al.\ \cite{Misra} studied the \emph{constrained} relay placement problem for connectivity and survivability. They provided $O(1)$ approximation algorithms for both the problems. We can call their first problem the \emph{Rooted Steiner Tree-Minimum Relays} problem, and their second problem, the \emph{Rooted Steiner Tree-Minimum Relays-Survivability} problem. Although their formulation takes into account an edge length bound, namely edge length$\leq r_c$, which can model the link quality, the formulation does not involve a path constraint such as the hop count along the path; hence, there is \emph{no constraint on the end-to-end QoS.} 

Yang et al.\ \cite{yang} studied a variation of the problem in \cite{Misra}, namely the \emph{two-tiered} constrained relay placement problem for connectivity and survivability, where each source has to be \emph{covered} by one (two) relay nodes, and the relay nodes form a one (two)-connected network with the BS. They provided $O(\ln n)$ approximation algorithms for arbitrary settings, and $O(1)$ approximation for some special cases. Their formulation also \emph{does not involve any constraint on the end-to-end QoS}.

The numerical experiments in both \cite{Misra} and \cite{yang} actually evaluate the empirical average case performance of their proposed algorithms on random test scenarios, which they compare against the theoretically derived \emph{worst case} performance bounds. Neither work, however, attempt a formal analysis of the average case performance of the proposed algorithms. 

\begin{table}[ht]
  \centering
\caption{A Comparison with Closely Related Literature; the ``starred'' problems are the ones we address in this paper; an entry `$\times$' in a column means that the corresponding algorithm does not provide the attribute given in the top of that column, whereas  a `$\checkmark$' means that it does provide the attribute.}
\label{tbl:literature}
\scriptsize
  \begin{tabular}{|l|c|c|c|c|}\hline
     & End-to-End  & & Worst Case Approximation & Average Case Approximation\\
Problem & Performance & Complexity & Guarantee of & Guarantee of \\
& Objective & & Proposed & Proposed\\
& & & Algorithm & Algorithm\\ 
\hline
   RST-MR \cite{Misra}& $\times$ & NP-Hard & 6.2 & $\times$\\
   \hline
   RST-MW-HC \cite{Voss} & $\checkmark$ & NP-Hard & $\times$ & $\times$\\
   \hline
   RST-MW-DC \cite{Kim} & $\checkmark$ & NP-Hard & $\times$ & $\times$\\
   \hline
   RST-MR-HC$*$            & $\checkmark$ & NP-Hard & polynomial factor & polynomial factor\\
   \hline
   RSN$k$-MR-HC$*$         & $\checkmark$ & NP-Hard & polynomial factor & polynomial factor\\
   \hline
\end{tabular}
\normalsize
\end{table}
In Table~\ref{tbl:literature}, we present a brief comparison of the problem under study in this paper with some of the closely related problems studied in the literature.

\section{RST-MR-HC: A Heuristic and its Analysis}
\label{sec:sptalgo} 

\subsection{Shortest Path Tree (SPT) based Iterative Relay Pruning Algorithm (SPTiRP)}
\begin{enumerate}
\item \textbf{The Zero Relay Case:} Find the SPT on $Q$ alone, rooted at the sink. If the hop count $\leq h_{\max}$ for each path, we  are done; no relays are required in an optimal solution. Else, go to the next step.
\item Find the Shortest Path Tree $T$ on $G$, rooted at the sink. 
\item \textbf{Checking Feasibility:} If for any path in the SPT, the path weight exceeds $h_{\max}$, declare the problem infeasible. (Clearly, if the shortest path from a node to the sink does not meet the hop count bound, no other path from the node to the sink will meet the hop count bound). Else, go to the next step.

\vspace{1em}
\textbf{Pruning the SPT:}
\item Discard all nodes in $R\backslash T$. Note that this step may lead to suboptimality as some nodes in $R\backslash T$ could be part of an optimal solution. 
\item Now, for the remaining relay nodes in $R$, define the weight of a relay node as the number of paths in the SPT that use the node. 
\item Arrange the paths in SPT in increasing order of hop count.
\item Among the paths in the SPT that use relay nodes, choose one that has the least number of hops This path has the maximum ``slack'' in the hop constraint. Arrange the relay nodes on this path in increasing order of their weights as defined in (5).
\item Remove the least weight relay node and consider the restriction of $G$ to the remaining nodes in $T$. Find an SPT on this graph. If in this SPT, path cost exceeds $h_{\max}$ for any path, then discard this SPT, replace the removed relay node, and repeat this step with the next least weight relay node. If all the relays in the least cost path have been tried without success, move on to the next least cost path, and repeat steps 7 and 8 for the relays in this path that have not yet been tried.
\item If in the above step, the SPT obtained satisfies the delay constraint for all the paths, then delete the removed relay node permanently from $R$ and repeat Steps 4 through 9. 
\item Stop when no more relay pruning is possible without violating the hop constraint on one or more of the paths.
 
\end{enumerate}

\noindent
\textbf{Discussion:}

 Step 1 of the above algorithm ensures that if the optimal design does not use any relay node, then the same holds true for our algorithm. That way we can make sure that the algorithm does not do infinitely worse in the sense that $\frac{Relay_{algo}}{Relay_{opt}}$ is finite. 

The idea behind Steps 7, 8 and 9 is that choosing to remove a relay from the path with the most slack in cost (i.e., hop constraint), we stand a better chance of still meeting the delay requirement with the remaining relays. Also, removing a relay of less weight would mean affecting the cost of a small number of paths. So by pruning relays in  the manner as described in Steps 7, 8 and 9, we aim for a better exploration of the search space. 

\subsection{Analysis of SPTiRP}
\subsubsection{Complexity}
The complexity of determining the shortest path tree on $N$ nodes is $O(N\log N)$ \cite{cormen}. Let us denote this function  by $g_{SPT}(.)$. In Iteration 1 of the algorithm, the complexity is $g_{SPT}(|Q|)$ and in iteration 2, it is $g_{SPT}(|Q|+|R|)$. In subsequent iterations, we remove 1 relay node at a time and find the SPT on the resultant complete graph; if no improvement is found, we replace that node and continue. Thus, for the $k^{th}$ iteration, the worst case complexity will be $(|R|-k+3)g_{SPT}(|Q|+|R|-k+2)$, where in the worst case, $k= 3, 4, \ldots, |R|+1$. Let $g_{sptirp}(.)$ denote the overall complexity. Thus, the overall complexity will be
\begin{align*}
g_{sptirp}(|Q|+|R|) &= g_{SPT}(|Q|+|R|)+ \\&\sum_{j=1}^{|R|}(g_{SPT}(|Q|+|R|-j))(|R|-j+1)\\
            &\leq (1+|R|^2)(g_{SPT}(|Q|+|R|))
\end{align*}
which is polynomial time.

\subsubsection{Worst Case Approximation Factor}

\begin{theorem}
The worst case approximation guarantee for the SPTiRP algorithm is $\min\{m(h_{\max}-1),(|R|-1)\}$, where $m$ is the number of sources, $h_{\max}$ is the hop constraint, and $|R|$ is the number of potential relay locations.
\end{theorem}

\begin{proof}
The worst case occurs when the SPT obtained before we enter Step (4) does not contain any relay node(s) that correspond to some optimal design. If no relays are used in any optimal design, then the algorithm will yield an optimal design (Step (1)). If an optimal solution uses a positive number of relays but not all of them, then SPTiRP cannot stop by using all the relays. For suppose, SPTiRP stops and uses all the relays. Since there is a feasible tree containing a strict subset of the relays, the pruning steps in SPTiRP will succeed in pruning at least one relay. Hence, the worst possibility is that the optimal design uses just 1 relay node, whereas the SPT obtained in Step (2) consists of all the remaining $(|R|-1)$ relays, and moreover, pruning any of these $(|R|-1)$ relays will cause one or more paths in the resulting SPT to violate the hop constraint. Thus, in the worst case, the algorithm may lead to a design with $(|R|-1)$ relays instead of the optimal design with one relay. Also note that for a problem with $m$ sources, and a hop constraint $h_{\max}$, no feasible solution can use more than $m(h_{\max}-1)$ relays. Hence, we have a polynomial factor worst case approximation guarantee of $\min\{m(h_{\max}-1),(|R|-1)\}$.
\end{proof}

\subsubsection{Sharp Examples (for Worst Case Approximation and for Optimality)}
Let us now present \emph{a sequence of problems of increasing complexity for which the approximation guarantee is strict}, i.e., for these problems, the algorithm ends up using $|R|-1$ relays, while the optimum design uses one relay. Such examples are worthwhile to explore as they help to show that the approximation factor obtained above cannot be improved.
 \begin{figure}[t]
\begin{center}
\includegraphics[scale=0.4]{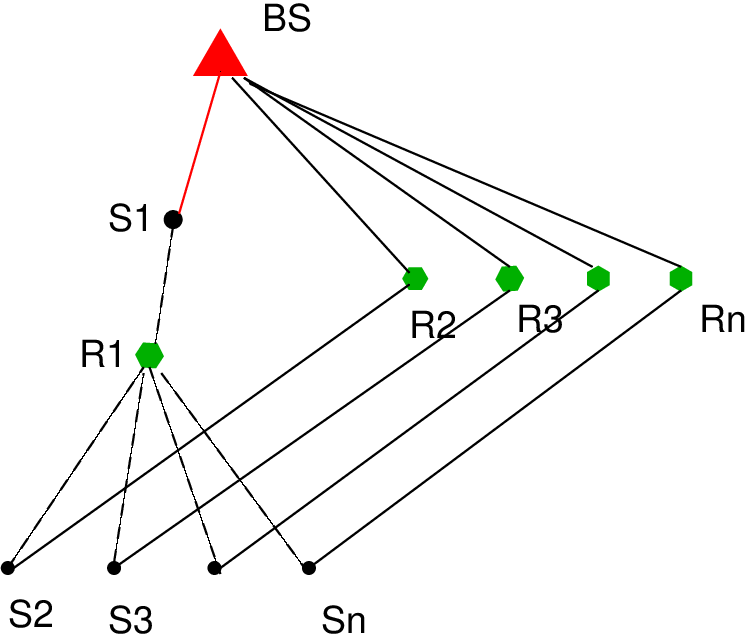}
\end{center}
\caption{ A Sequence of Problems where the Worst Case Approximation Guarantee is Strict}
\label{fig:worstcase}
\end{figure}
Consider the situation shown in Figure~\ref{fig:worstcase}. The green hexagons denote the relay node locations and the black circles represent the source node locations. Only the edges shown (coloured or black) are permitted. Consider the RST-MR-HC problem on this graph with $h_{max}= 3$. Clearly the optimal solution will use only one relay, $R1$, to reach from each source to the BS within the specified hop count bound. The black dotted links correspond to the optimal solution. The red link will belong to both the optimal solution and the outcome of our algorithm as it is a direct link between source $S1$ and the BS. Our SPT based algorithm will calculate the shortest paths and thus end up using relays $R2, R3, \ldots, Rn$, leaving out $R1$. The black solid links correspond to the solution given by our algorithm. Clearly, in such problems, we end up using $|R|-1$ relays instead of just one. 

Another sequence of problems of increasing complexity for which the algorithm gives the optimal design can be constructed as shown in Figure~\ref{fig:optimum}. Such examples help to show that the proposed algorithm does provide an optimal solution in some scenarios.

\begin{figure}[t]
\begin{center}
\includegraphics[scale=0.4]{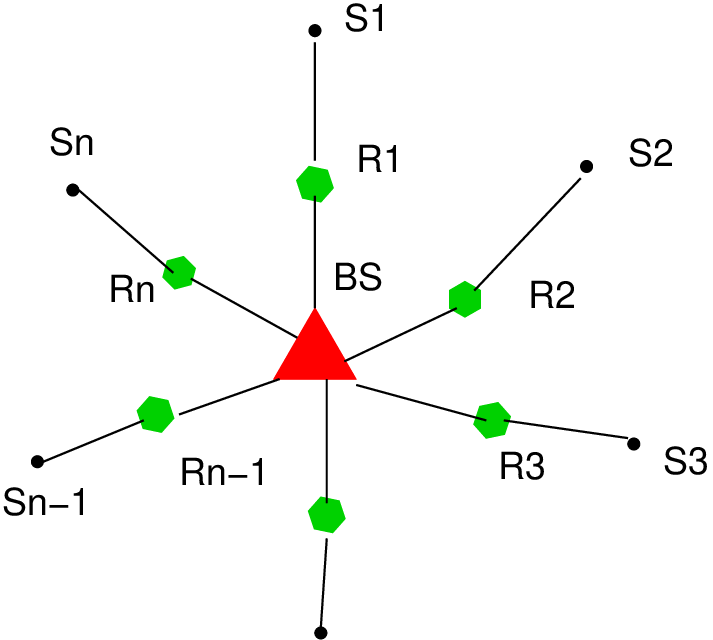}
\end{center}
\caption{ A Sequence of Problems where SPTiRP gives Optimal Solution}
\label{fig:optimum}
\end{figure}

As before, the green hexagons represent relay locations and the black dots represent source nodes. Suppose $h_{max}= 2$. Then clearly, the optimal solution is as shown in the figure. The algorithm, after calculating the SPT, will end up with the same solution. 

\subsubsection{Average Case Approximation Factor of SPTiRP}
\label{subsubsec:avg-approx-oneconnect}
We shall derive below, an upper bound on the average case approximation factor of SPTiRP in a certain stochastic setting, defined by a probability distribution on the potential relay locations and the source locations. The derivation, in fact, applies to any algorithm that starts with an SPT, and proceeds by pruning relays from the SPT in some manner. The probability distributions (and hence the setting) are chosen so as to ensure the existence of a feasible solution with high probability. 

We consider a square area $A(\subset \Re^2_+)$ of side $a$. The BS is located at (0,0). We deploy $n$ potential locations randomly over $A$, yielding the potential locations vector $\underline{x} \in A^n$. Then we place $m$ sources over $A$, yielding source location vector $\underline {y} \in A^m$. Let $\omega\:=\:(\underline{x}, \underline{y})$, i.e., $\omega$ denotes the joint potential locations vector and source locations vector. We assume a model where a link of length $\leq r$ metres has the desired PER so that $h_{\max}$ is the hop constraint. We then consider the geometric graph, $\mathcal{G}^r(\omega)$, over these $n+m$ points; i.e., in $\mathcal{G}^r(\omega)$ there is an undirected edge between a pair of nodes in $(\omega)$ if the Euclidean distance between these nodes is $\leq r$. If in this graph the shortest path from each source to the BS (at $(0,0)$) has a hop count $\leq h_{\max}$, then $\omega$ is feasible. Define

\begin{description}
\item $H_j(\omega)$: Hop distance (i.e., the number of hops in the shortest path) of source $j$ from the BS in $\mathcal{G}^r(\omega)$, $1\leq j \leq m$. 
($\infty$ if source $j$ is disconnected from BS in $\mathcal{G}^r(\omega)$)

\item $\X=\{(\underline{x},\underline{y}):\:\forall y_j,\:1\leq j\leq m,\:H_j \leq h_{\max}\}$: Set of all feasible instances
\end{description}

We would like $\mathcal{X}$  to be a high probability event. For this we need to limit the locations of the sources to be no more than $(1-\epsilon) r h_{\max}$ from the BS; Theorem~\ref{thm:average-case-analysis-setting}, later, will help characterize the relationship between $\epsilon$, the number of potential locations, and the probability of $\mathcal{X}$.

For a given $\epsilon\in(0,1)$, let $A_{\epsilon}(\subset A)$ denote the quarter circle of radius $(1-\epsilon)h_{\max}r$ centred at the BS, where $h_{\max}$ is the hop constraint, and $r$ is the maximum allowed communication range. 

Formally, we deploy $n$ potential locations independently and identically distributed (i.i.d) uniformly randomly over the area $A$; then deploy $m$ sources i.i.d uniformly randomly over the area $A_{\epsilon}$. The probability space of this random experiment is denoted by $(\Omega^{(n)}_{m,\epsilon},\mathcal{B}^{(n)}_{m,\epsilon},P^{(n)}_{m,\epsilon})$, where,

\begin{description}
\item $\Omega^{(n)}_{m,\epsilon}=(A^n\times A^m_{\epsilon})(\subset \Re^{2(n+m)}_+)$: Sample space; the set of all possible deployments
\item $\mathcal{B}^{(n)}_{m,\epsilon}$: The Borel $\sigma$-algebra in $\Omega^{(n)}_{m,\epsilon}$
\item $P^{(n)}_{m,\epsilon}$: Probability measure induced on $\mathcal{B}^{(n)}_{m,\epsilon}$ by the uniform i.i.d deployment of nodes
\end{description}

Consider the random geometric graph $\mathcal{G}^r(\omega)$ induced by considering all links of length $\leq r$ on an instance $\omega \in \Omega^{(n)}_{m,\epsilon}$. We introduce the following notation:

\begin{description}
\item $N_{SPTiRP}(\omega)$: number of relays in the outcome of the SPTiRP algorithm on $\mathcal{G}^r(\omega)$ ($\infty$ if $\omega\in\X^c$)
\item $R_{Opt}(\omega)$: number of relays in an optimal solution to the RST-MR-HC problem on $\mathcal{G}^r(\omega)$ ($\infty$ if $\omega\in\X^c$)
\end{description}

The \emph{average case approximation ratio} of the SPTiRP algorithm over feasible instances is defined as
\begin{equation}
\text{Average case approximation ratio, }\alpha \define \frac{E[N_{SPTiRP}|\X]}{E[R_{Opt}|\X]}
\end{equation}

\remark This would be a useful quantity if the user of the algorithm wishes to apply the algorithm to several instances of the problem, yielding the required number of relays $N_1, N_2, \cdots, N_k,$ as against the optimal number of relays $R_1, R_2, \cdots, R_k,$ and is interested in the ratio 
$\frac{N_1 + N_2 + \cdots + N_k}{R_1 + R_2 +\cdots + R_k}$.

In the derivation to follow, we will need $\mathcal{X}$ to be a high probability event, i.e., with probability greater than $1 - \delta$ for a given $\delta > 0$. The following result ensures that this holds for the construction provided earlier, provided the number of potential locations is large enough.

\begin{theorem}
\label{thm:average-case-analysis-setting}
For any given $\epsilon,\delta \in (0,1)$, $h_{\max}>0$ and $r>0$, there exists $n_0(\epsilon,\delta,h_{\max},r)\in \mathbb{N}$ such that, for any $n\geq n_0$, $P^{(n)}_{m,\epsilon}(\X)\:\geq\:1-\delta$ in the random experiment $(\Omega^{(n)}_{m,\epsilon},\mathcal{B}^{(n)}_{m,\epsilon},P^{(n)}_{m,\epsilon})$.  
\end{theorem} 

\begin{proof}
The proof follows along the lines of the proof of Theorem 3 in \cite{nath-venkatesan}. We make the construction as shown in Figure~\ref{fig:cutarc}. From the BS $b_l$, we draw a circle of radius $h_{\max}r$ centered at $b_l$, this is the maximum distance reachable in $h_{\max}$ hops, by triangle inequality, since each hop can be of maximum length $r$. We then construct blades as shown in Figure~\ref{fig:cutarc}. We start with one blade. It will cover some portion of the circumference of the circle of radius $h_{\max}r$; see Figure~\ref{fig:cutarc}. Construct the next blade so that it covers the adjacent portion of the circumference that has not been covered by the previous blade. We go on constructing these blades until the entire portion of the circle lying inside the area $A$ is covered (see Figure~\ref{fig:cutarc}). Let us
define,

\begin{itemize}
\item $J(r)$ : Number of blades required to cover the part of the
  circle within $A$.
\item $B_j^l$ : $j^{th}$ blade drawn from the point $b_l$ as
  shown in Figure~\ref{fig:cutarc}, $1 \leq j \leq J(r)$.
\end{itemize}
On each of these blades, we construct $h_{\max}$ strips\footnote{A construction with improved convergence rate based on lens-shaped areas rather than rectangular strips is presented in Appendix~B of \cite{nath-venkatesan}.} , shown shaded in Figure~\ref{fig:h hop},
$u(r)$ being the width of the blade and $t(r)$ the width of the strip.
We define the following events.
\begin{description}
\item $A_{i,j}^l$ = \{$\omega$: $\exists$ at least one node out of the $n$ potential loacations in the $i^{th}$ strip of $B_j^l$\}
\item $\X_{\epsilon,\delta}\:=\:\{\omega:\:\omega\:\in\cap_{j=1}^{J(r)} \cap_{i=1}^{h_{\max}-1} A_{i,j}^l\}$: Event that there exists at least one node out of the $n$ potential locations in each of the first $(h_{\max}-1)$ strips (see Figure~\ref{fig:h hop}) for all the blades $B_j^l$
\end{description}

\begin{figure}[t]
\centering
\includegraphics[width= 9cm, height= 7cm]{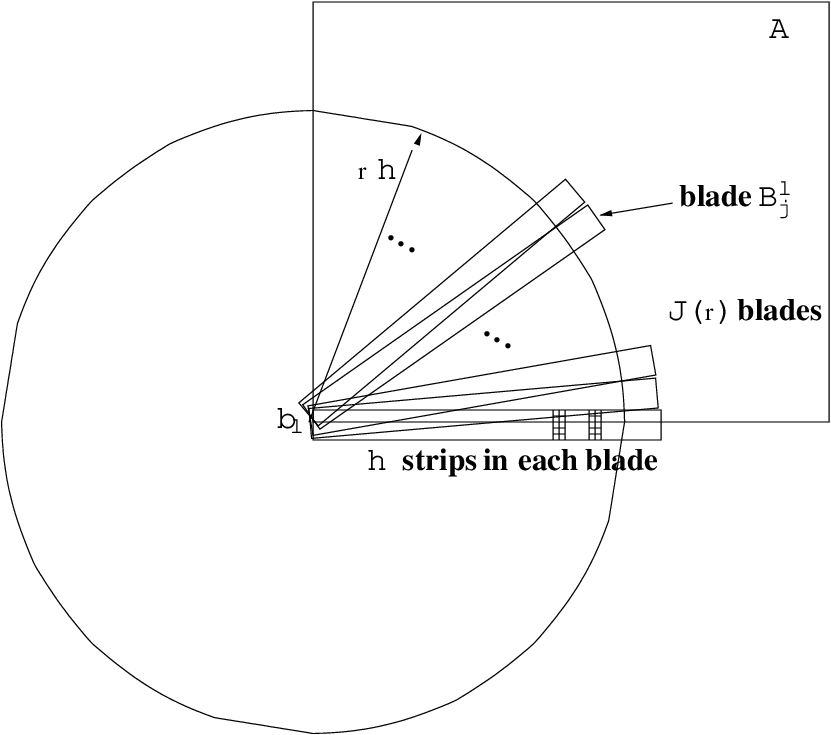}
\\\caption{Construction using the blades cutting the circumference of
  the circle of radius $hr$ (adapted from Nath et al.~\cite{nath-venkatesan}).}
\label{fig:cutarc}
\end{figure}

\begin{figure}[t]
\centering
\includegraphics[width= 9cm, height= 3cm]{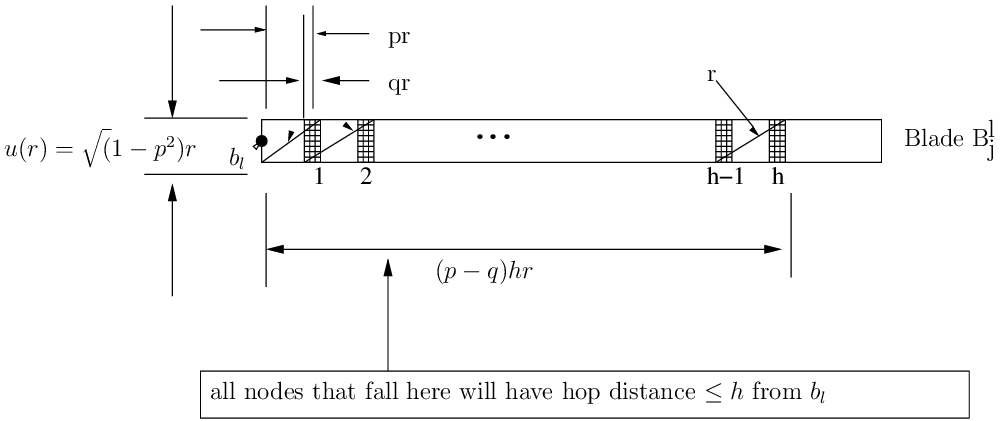}
\caption{The construction with $h$ hops (adapted from Nath et al.~\cite{nath-venkatesan}).}
\label{fig:h hop}
\end{figure}

Note that for an instance $\omega\in\X_{\epsilon,\delta}$, \emph{all nodes (and in particular, all sources) at a distance $<(p-q)hr$ from $b_l$, $1\leq h\leq h_{\max}$, are reachable in at most $h$ hops}. Since $1>p>q>0$, we can choose $p-q$ to be equal to $1 - \epsilon$, for the given $\epsilon > 0$. It follows that
\begin{equation}
\label{eq:intersection}
\X_{\epsilon,\delta}\subseteq \X
\end{equation}

and hence, $P^{(n)}_{m,\epsilon}(\X)\geq P^{(n)}_{m,\epsilon}(\X_{\epsilon,\delta})$.

Thus, to ensure $P^{(n)}_{m,\epsilon}(\X)\geq 1-\delta$, it is sufficient to ensure that $P^{(n)}_{m,\epsilon}(\X_{\epsilon,\delta})\geq 1-\delta$, which we aim to do next. 
  
To find the value of $J(r)$, we need to define the following.
\begin{description}
\item $a(r)$: Length of the arc of radius $h_{\max}r$ that lies
  within a blade, drawn taking $b_l$ as center, as shown in
  Figure~\ref{fig:arc}.
\item $\alpha(r)$ : Angle subtended by $a(r)$ at $b_l$ , see
  Figure~\ref{fig:arc}.
\end{description}

Now from Figure~\ref{fig:cutarc}, we have, $J(r) = \left \lceil
  \frac{\pi}{2\alpha(r)}\right \rceil$. We also have from
Figure~\ref{fig:arc}, $h_{\max}r\alpha(r)=a(r) \geq
u(r)=\sqrt{1-p^2}r$. Hence, $\alpha(r) \geq
\frac{\sqrt{1-p^2}}{h_{\max}}$. So, $J(r) \leq \left \lceil \frac{\pi h_{\max}}{2\sqrt{1-p^2}}\right \rceil$.

To simplify notations, we shall henceforth write $P(\cdot)$ to indicate $P^{(n)}_{m,\epsilon}(\cdot)$.

Now we compute,
\begin{eqnarray}
  \label{eq:stripnode}
  \lefteqn{P(\X_{\epsilon,\delta})} \nonumber \\
  &=& 1-P\left(\cup_{j=1}^{J(r)}\cup_{i=1}^{h_{\max}-1}
{A_{i,j}^{l}}^c\right) \nonumber \\
  &\geq& 1-\sum_{j=1}^{J(r)}\sum_{i=1}^{h_{\max}-1}
P\left({A_{i,j}^{l}}^c\right) \nonumber \\
  &\geq& 1 - \left \lceil \frac{\pi h_{\max}}{2\sqrt{1-p^2}}\right \rceil (h_{\max}-1)
\left(1-\frac{u(r)t(r)}{A}\right)^n \nonumber \\
  &\geq& 1- \left \lceil \frac{\pi h_{\max}}{2\sqrt{1-p^2}}\right \rceil (h_{\max}-1)
e^{-\frac{nu(r)t(r)}{A}} \nonumber \\
  &=& 1- \left \lceil \frac{\pi h_{\max}}{2\sqrt{1-p^2}}\right \rceil (h_{\max}-1)
e^{-n\frac{q\sqrt{1-p^2}r^2}{A}} 
\end{eqnarray}

\begin{figure}[t]
\centering
\includegraphics[width= 8.5cm, height= 4cm]{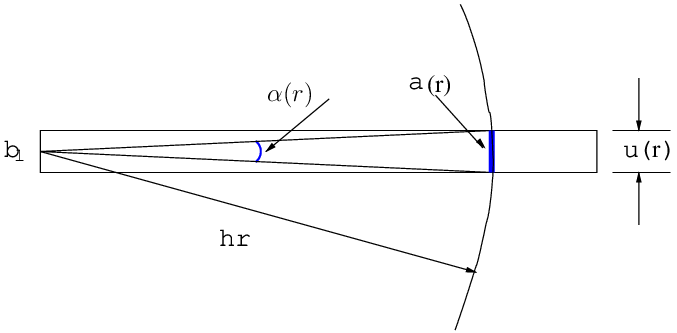}
\\\caption{Construction to find $J(r)$ (adapted from Nath et al.~\cite{nath-venkatesan}).}
\label{fig:arc}
\end{figure}

The first inequality comes from the union bound, the second
inequality, from the upper bound on $J(r)$. The third inequality uses
the result $1-x \leq e^{-x}$. 

Thus, in order to achieve $P(\X_{\epsilon,\delta})\geq 1-\delta$ (and hence, $P(\X)\geq 1-\delta$), it is sufficient that 

\begin{align}
1- \left\lceil \frac{\pi h_{\max}}{2\sqrt{1-p^2}}\right\rceil e^{-n\frac{q\sqrt{1-p^2}r^2}{A}}&\geq 1-\delta\nonumber\\
\Rightarrow n &\geq \frac{A}{q\sqrt{1-p^2}r^2}\ln\left\{\left\lceil \frac{\pi h_{\max}}{2\sqrt{1-p^2}}\right\rceil \frac{1}{\delta}\right\}\define n_0(\epsilon,\delta,h_{\max},r)\label{eqn:minimum_nodes_feasibility}
\end{align}

where, $p$ and $q$ can be obtained in terms of $\epsilon$ by maximizing $q\sqrt{1-p^2}$ (so as to somewhat tighten the bound in Equation~\eqref{eqn:minimum_nodes_feasibility}) under the constraint $p-q\:=\:1-\epsilon$.

Note that a tighter bound can be obtained by using the ``eyeball'' construction presented in Appendix B of \cite{nath-venkatesan} instead of the rectangular strip construction presented here.
\end{proof}

\gap
\noindent
\remark For fixed $h_{\max}$ and $r$, $n_0(\epsilon,\delta)$ increases with decreasing $\epsilon$ and $\delta$.

\gap
\noindent
\textbf{The experiment:} In the light of Theorem~\ref{thm:average-case-analysis-setting}, we employ the following node deployment strategy to ensure, w.h.p, feasibility of the RST-MR-HC problem in the area $A$. 
Choose arbitrary small values of $\epsilon,\delta \in (0,1)$. Given the hop count bound $h_{\max}$ and the maximum communication range $r$, obtain $n_0(\epsilon,\delta,h_{\max},r)$ as defined in Theorem~\ref{thm:average-case-analysis-setting}. Deploy $n\geq n_0$ potential locations i.i.d uniformly randomly over the area of interest, $A$. $m$ sources are deployed i.i.d uniformly randomly within a radius $(1-\epsilon)h_{\max}r$ from the BS, i.e., over the area $A_{\epsilon}$. By virtue of Theorem~\ref{thm:average-case-analysis-setting}, this ensures that any source deployed within a distance $(1-\epsilon)h_{\max}r$ is no more than $h_{\max}$ hops away from the BS w.h.p, thus ensuring feasibility of the RST-MR-HC problem w.h.p. We check whether the deployment is feasible by computing the SPT on the induced random geometric graph with hop count as cost. \emph{In this stochastic setting, we derive an upper bound on the average case approximation ratio, $\alpha$, of the SPTiRP algorithm} as follows.

\begin{lemma}
\label{lem:spt-nodecount-expectn-feasible}
\begin{align}
E[N_{SPTiRP}|\X]&\leq m[h_{\max}-\frac{1}{(1-\epsilon)^2h_{\max}^2}\nonumber\\
&-\sum_{j=2}^{h_{\max}-1}\frac{j^2}{h_{\max}^2}]-m\nonumber\\
&+\:m\delta(h_{\max}-1)
\label{eqn:spt-nodecount-expectn-feasible}
\end{align}
\end{lemma}

\begin{proof}
We define the following:

\begin{description}
\item $N_{SPT}(\omega)$: number of relays in the SPT on $\mathcal{G}^r(\omega)$ for the $m$ sources with BS as the root ($\infty$ if $\mathcal{G}^r(\omega)$ is disconnected)
\item $\mathcal{S}_j(\omega)$: set of sources whose Euclidean Distance ($D_s$) from the BS satisfy $(1-\epsilon)(j-1)r<D_s \leq (1-\epsilon)jr$, for $j=3,\ldots,h_{\max}$, and $r<D_s \leq (1-\epsilon)2r$ for $j=2$, $0<D_s \leq r$ for $j=1$
\item $M_j(\omega)$: Number of sources in the set $\mathcal{S}_j(\omega)$, i.e., $|\mathcal{S}_j(\omega)|$
\item $\overline{H}_j(\omega)$: maximum number of hops in the shortest path from a source in $\mathcal{S}_j(\omega)$ to the BS
\end{description} 

Note that $\sum_{j=1}^{h_{\max}}M_j(\omega)=m$.

Recall that $N_{SPTiRP}(\omega)$ denotes the number of relays in the solution provided by the SPTiRP algorithm on a feasible instance. Since the algorithm starts by finding an SPT, and then pruning relays from that SPT, we have
\begin{equation*}
N_{SPTiRP}(\omega)\leq N_{SPT}(\omega) 
\end{equation*}
 Hence, we can upper bound the expected number of relays in the SPTiRP solution on a feasible instance as 
\begin{equation}
\label{eqn:sptirp-spt-ineq}
E[N_{SPTiRP}|\X]\leq E[N_{SPT}|\X]
\end{equation}

Now, observe that 
\begin{equation}
N_{SPT} \leq \sum_{j=1}^{h_{\max}}M_j\overline{H}_j\:-\:m
\label{eqn:spt-nodecount}
\end{equation}

Also note that, given $\X_{\epsilon,\delta}$ (i.e., given a deployment in $\X_{\epsilon,\delta}$), $\overline{H}_j \leq j,$ $\forall j=1,\ldots, h_{\max}$.

Therefore, taking expectation on both sides of \eqref{eqn:spt-nodecount}, we have

\begin{align}
E[N_{SPT}|\X_{\epsilon,\delta}]&\leq E[\sum_{j=1}^{h_{\max}}M_j\overline{H}_j|\X_{\epsilon,\delta}]\:-\:m\nonumber\\
&\leq \sum_{j=1}^{h_{\max}}jE[M_j|\X_{\epsilon,\delta}]\:-\:m\nonumber\\
&=\sum_{j=1}^{h_{\max}}jE[M_j]\:-\:m\quad \text{since }M_j\perp \X_{\epsilon,\delta}\nonumber\\
&= m[h_{\max}-\frac{1}{(1-\epsilon)^2h_{\max}^2}\nonumber\\
&-\sum_{j=2}^{h_{\max}-1}\frac{j^2}{h_{\max}^2}]-m, \quad\text{after simplification}
\label{eqn:spt-nodecount-expectn1}
\end{align}

However, a deployment in $\X_{\epsilon,\delta}$ is sufficient, but not necessary for feasibility of the RST-MR-HC problem. When a deployment is not in $\X_{\epsilon,\delta}$, but still there exists a feasible solution satisfying the hop constraint, the number of nodes in the SPT can be trivially upper bounded as $m(h_{\max}-1)$. Hence, 
\begin{align}
E[N_{SPT}|\X]&= E[N_{SPT}|\X_{\epsilon,\delta},\X]P[\X_{\epsilon,\delta}|\X]\nonumber\\
&\:+\:E[N_{SPT}|\X_{\epsilon,\delta}^c,\X]P[\X_{\epsilon,\delta}^c|\X]\nonumber\\
&\leq E[N_{SPT}|\X_{\epsilon,\delta}] \:+\:E[N_{SPT}|\X_{\epsilon,\delta}^c,\X]\delta\nonumber\\
&\leq m[h_{\max}-\frac{1}{(1-\epsilon)^2h_{\max}^2}\nonumber\\
&-\sum_{j=2}^{h_{\max}-1}\frac{j^2}{h_{\max}^2}]-m\nonumber\\
&+\:m\delta(h_{\max}-1)
\label{eqn:spt-nodecount-expectn}
\end{align}

The lemma follows by combining equations~\eqref{eqn:sptirp-spt-ineq} and \eqref{eqn:spt-nodecount-expectn}.
\end{proof}

\begin{lemma}
\label{lem:ropt-lowerbound-final}
\begin{align}
E[R_{Opt}|\X]&\geq \left[1-\left(\frac{h_{\max}-1}{(1-\epsilon)h_{\max}}\right)^{2m}\right](1-\delta)\sum_{i=1}^{h_{\max}-1}\left(1-\frac{\frac{n_i^2}{3}}{(1-\epsilon)^2h_{\max}^2}\right)^{m-1}
\label{eqn:ropt-lowerbound-final}
\end{align}
where, $n_i=\min(i,h_{\max}-i)$.
\end{lemma}

\begin{proof}

We can write

\begin{align}
E[R_{Opt}|\X]&\geq E[R_{Opt}\ind_{\X_{\epsilon,\delta}}|\X]\nonumber\\
&= P[\X_{\epsilon,\delta}|\X]\times \:E[R_{Opt}|\X_{\epsilon,\delta},\X]\nonumber\\
&\geq P[\X_{\epsilon,\delta},\X]\:E[R_{Opt}|\X_{\epsilon,\delta},\X]\nonumber\\
&= P[\X_{\epsilon,\delta}]\:E[R_{Opt}|\X_{\epsilon,\delta}],\quad\text{since $\X_{\epsilon,\delta}$ implies feasibility}\nonumber\\
&\geq (1-\delta)\:E[R_{Opt}|\X_{\epsilon,\delta}]
\label{eqn:ropt-lowerbound-by-xep}
\end{align}

\remark The first inequality above is tight since $\Xep$ is a high probability event (for the chosen deployment strategy). The second inequality is tight since $P(\X)\geq P(\X_{\epsilon,\delta},\X)\geq 1-\delta$, and hence $P(\X)$ is close to 1. The third inequality is tight when the number of potential relay locations is just enough to meet the requirement $P(\Xep)\geq 1-\delta$, i.e., $P(\Xep)\approxeq 1-\delta$.

\gap
We define
\begin{description}
\item $\overline{D}_s(\omega)$: The maximum Euclidean distance from the BS, of a source location in $\omega$ 
\end{description}

Then, for the conditional expectation term on the right hand side of Eqn.~\eqref{eqn:ropt-lowerbound-by-xep}, we can write
\begin{align}
E[R_{Opt}|\Xep]&=E[R_{Opt}\ind_{\{\overline{D}_s\leq (h_{\max}-1)r\}}|\Xep]\nonumber\\
&\:+\:E[R_{Opt}\ind_{\{\overline{D}_s > (h_{\max}-1)r\}}|\Xep]\nonumber\\
&\geq E[R_{Opt}\ind_{\{\overline{D}_s > (h_{\max}-1)r\}}|\Xep]\nonumber\\
&= P[\overline{D}_s > (h_{\max}-1)r|\Xep]\times \nonumber\\
&\:E[R_{Opt}|\overline{D}_s > (h_{\max}-1)r,\Xep]\nonumber\\
&= P[\overline{D}_s > (h_{\max}-1)r]\times \nonumber\\
&\:E[R_{Opt}|\overline{D}_s > (h_{\max}-1)r,\Xep]
\label{eqn:ropt-lowerbound}
\end{align}
where the last equality follows since the event $\Xep$ depends only on the positions of the relays, while the event $\{\omega: \overline{D}_s(\omega) > (h_{\max}-1)r\}$ depends only on the sources, thus being independent of each other. 

\remark The inequality above may not be loose since the probability that there exists at least one source in the ring with inner and outer radii ($(h_{\max}-1)r,(1-\epsilon)h_{\max}r$) is significantly large compared to that in the inner rings (the last ring having the maximum area among all the rings), and this probability increases with increasing number of sources. 

\gap
Note that $\overline{D}_s > (h_{\max}-1)r$ on an instance in $\Xep$ implies that there exists at least one source in the ring with inner and outer radii ($(h_{\max}-1)r,(1-\epsilon)h_{\max}r$), and, being feasible, it must be $h_{\max}$ hops away from the BS. For ease of writing, let us define

\begin{description}
\item $\X_{h_{\max}}\:=\:\{\omega: \overline{D}_s(\omega) > (h_{\max}-1)r\}\cap \Xep$
\end{description}

Consider an instance $\omega\in\Xh$. In this $\omega$, let us denote by $s(\omega)$, the source which is farthest from the BS among all the sources in the outermost ring (centred at the BS), with inner and outer radii ($(h_{\max}-1)r,(1-\epsilon)h_{\max}r$). 

Observe that, $R_{Opt}(\omega)$ is \emph{lower bounded by the number of relays in the path from the source $s(\omega)$ to the BS in any optimal solution in $\omega$.}

\begin{figure}[t]
\centering
\includegraphics[scale=0.4]{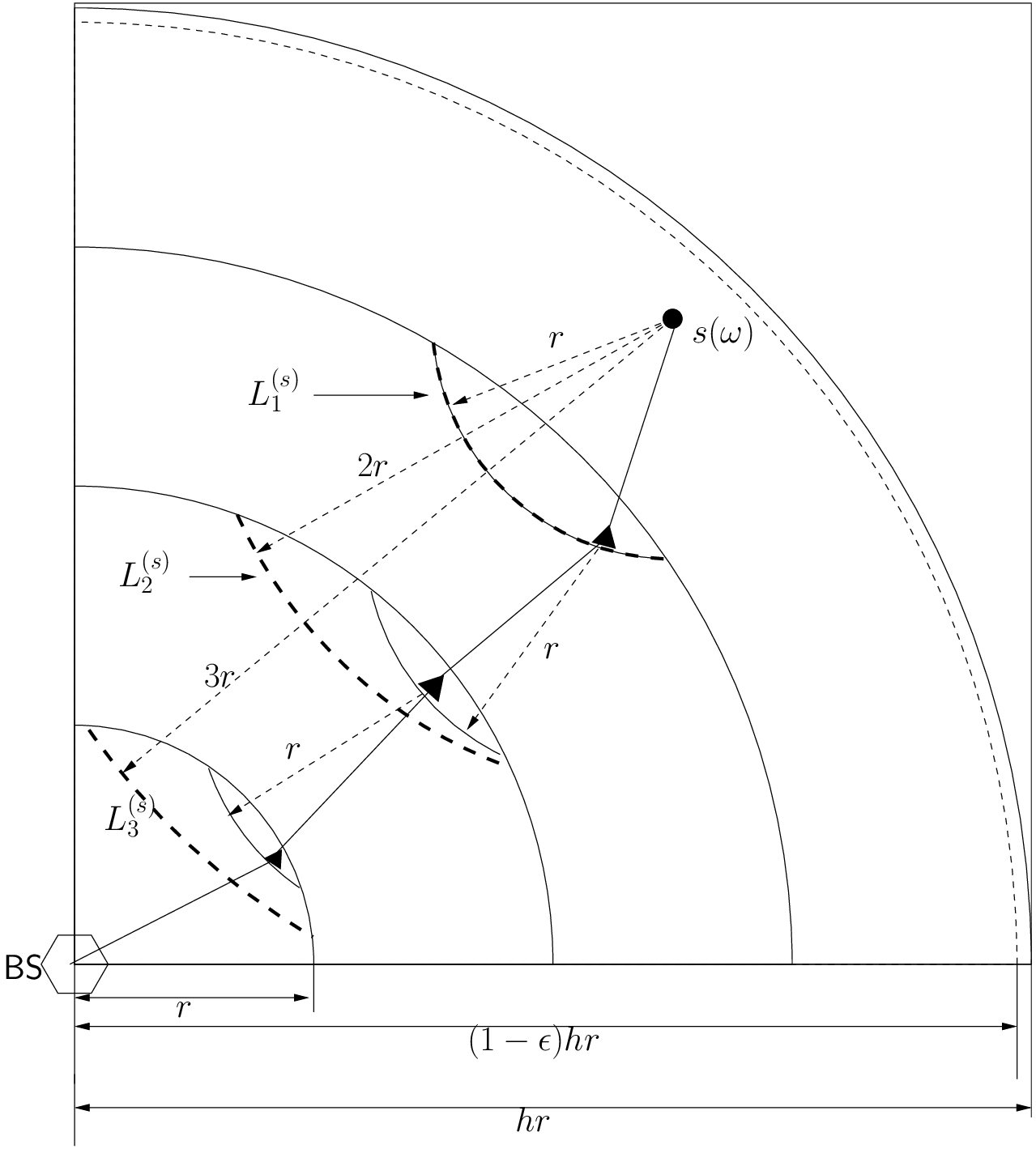}
\\\caption{Illustration of the lenses $L_j^{(s)}, 1\leq j\leq h_{\max}-1$, used in the proof of Lemma~\ref{lem:ropt-lowerbound-final}; $L_j^{(s)}$ contains in it, the $j^{th}$ lens, and hence the $j^{th}$ intermediate node in a feasible path from source $s(\omega)$ to the BS. The solid triangles indicate the intermediate nodes in a feasible path from source $s(\omega)$ to the BS.}
\label{fig:lens_diag}
\end{figure}

For an instance $\omega\in\Xh$, we have the following properties:
\begin{enumerate}
\item for each node in a \emph{feasible} path from the source $s(\omega)$ to the BS, 
\begin{enumerate}
\item \emph{The next hop node must lie in the lens shaped intersection of the communication circle of radius $r$ of that node, and the next ring}. These lenses are disjoint, and there are $(h_{\max}-1)$ of them. This property is illustrated in Figure~\ref{fig:lens_diag}, where we have, $h_{\max}=4$, and we have indicated a feasible path from the source $s(\omega)$ (which is in the outermost ring) to the BS; each link in the path is indicated by a solid straight line, and each intermediate node is indicated by a triangle. Also shown by \emph{narrow solid arc} is the lens shaped intersection of the communication circle of radius $r$ of each node, and the next ring. Note that for each node, the next hop node always lies in this lens shaped intersection. Moreover, these lenses are disjoint.
\item The next hop node is a relay if there does not exist a source node (out of at most $(m-1)$ remaining source nodes) within the next ``lens''.
\end{enumerate}
\item Further, it follows from the \emph{triangle inequality} that the maximum possible Euclidean distance from $s(\omega)$, of the $j^{th}$ node (counting from the source side, excluding the source) in a feasible path from $s(\omega)$ to the BS, is $jr$, \emph{irrespective of the position of the intermediate nodes}. Hence, we denote by $L_j^{(s)}(\omega)$, the lens shaped intersection of the circle of radius $jr$ centred at $s(\omega)$, and the circle of radius $(h_{\max}-j)r$ centred at the BS. Clearly, $L_j^{(s)}(\omega)$ encompasses in it, all possible lenses (depending on the positions of the intermediate nodes) that might contain the $j^{th}$ node in a feasible path of $s(\omega)$. Also note that the lenses $L_j^{(s)},\:1\leq j\leq h_{\max}-1$, are disjoint. For example, see Figure~\ref{fig:lens_diag}, where $h_{max}=4$, and we have indicated by thick dashed arcs, the lenses $L_j^{(s)}, 1\leq j\leq 3$, for the source $s(\omega)$. As can be seen from the figure, $L_j^{(s)}, 1\leq j\leq 3$, contains in it, the $j^{th}$ lens, and hence the $j^{th}$ intermediate node in the feasible path from source $s(\omega)$ to the BS. Also, we see from Figure~\ref{fig:lens_diag} that the lenses $L_j^{(s)}, 1\leq j\leq 3$ are disjoint. 
\end{enumerate}

Now, for any $\omega\in\{\omega:\overline{D}_s(\omega) > (h_{\max}-1)r\}$, we define, $\forall i,\:1\leq i\leq h_{\max}-1$,
\begin{equation*}
Y_i(\omega) = \left\{
\begin{array}{rl}
1, & \text{if } ,\exists\:\text{no source in  $L_i^{(s)}(\omega)$}\\
0, & \text{otherwise}
\end{array} \right.
\end{equation*}

Set $Y_i(\omega)=\infty$ if $\omega\in\{\omega:\overline{D}_s(\omega) > (h_{\max}-1)r\}^c$. $Y_i(\omega)$ is uniquely determined by $\omega$, and does not depend on any particular optimal solution. 

Thus, it follows from the definition of $Y_i(\omega)$ and the properties 1 and 2 above that, in an optimal solution, the number of relays in the path from the source $s(\omega)$ to the BS is at least $\sum_{i=1}^{h_{\max}-1}Y_i(\omega)$ (since whenever $Y_i(\omega)=1,\:1\leq i\leq h_{\max}-1$, the $i^{th}$ hop node in the path from source $s(\omega)$ to the BS must be a relay). Hence, 

\begin{align}
R_{Opt}(\omega)&\geq\sum_{i=1}^{h_{\max}-1}Y_i(\omega)\quad\forall \omega\in\Xh
\end{align}

Let us further define
\begin{equation*}
M^o(\omega) = \left\{
\begin{array}{rl}
\text{number of sources in the outermost ring}, & \text{if } \omega\in\{\omega:\overline{D}_s(\omega) > (h_{\max}-1)r\}\\
0, & \text{otherwise}
\end{array} \right.
\end{equation*}

Thus, we have

\begin{align}
E[R_{Opt}|\Xh] &\geq E\left[\sum_{i=1}^{h_{\max}-1}Y_i|\Xh\right]\nonumber\\
&= \sum_{j=1}^m P[M^o = j|\Xh]\:E\left[\sum_{i=1}^{h_{\max}-1}Y_i|M^o=j,\Xh\right]\nonumber\\
&= \sum_{j=1}^m P[M^o = j|\Xh]\:\left(\sum_{i=1}^{h_{\max}-1}E\left[Y_i|M^o=j,\Xh\right]\right)\nonumber\\
&= \sum_{j=1}^m P[M^o = j|\Xh]\:\left(\sum_{i=1}^{h_{\max}-1}P[Y_i=1|M^o=j,\Xh]\right)\nonumber\\
&\geq \sum_{j=1}^m P[M^o = j|\Xh]\:\left(\sum_{i=1}^{h_{\max}-1}q_i\right),\quad q_i\coloneqq P[Y_i=1|M^o=1,\Xh]\label{eqn:qidef_ineq}\\
&= P[M^o\geq 1|\Xh]\sum_{i=1}^{h_{\max}-1}q_i \nonumber\\
&= \sum_{i=1}^{h_{\max}-1}q_i,\quad\text{since }P[M^o\geq 1|\Xh]=1\label{eqn:ropt-lowerbound-inter}
\end{align}
where, the inequality \ref{eqn:qidef_ineq} follows by taking $(m-1)$ sources (which is the maximum possible number, given that there exists at least one source in the outermost ring) to be free to enter the lenses $L_i^{(s)},\:1\leq i\leq h_{\max}-1$.

To obtain a lower bound on $q_i$, we proceed as follows.

\begin{align}
q_i &= P[Y_i\:=\:1|M^o=1,\Xh]\nonumber\\
&= P[Y_i\:=\:1|M^o=1,\overline{D}_s > (h_{\max}-1)r]\quad \text{$\X_{\epsilon,\delta}$ indept. of source node locations}\label{eqn:p_yi_lowerbound}
\end{align}
When there exists a single source in the outermost ring, $Y_i\:=\:1$ if \emph{there does not exist a source node (out of at most $(m-1)$ remaining source nodes) within the lens $L_i^{(s)}$}.

\gap
\noindent
\textbf{Claim: } The area of lens $L_i^{(s)}$ is upper bounded by $\frac{\pi}{3}n_i^2r^2$, where $n_i=\min(i,h_{\max}-i)$.

\begin{proof} 

\begin{figure}[t]
\centering
\includegraphics[scale = 0.4]{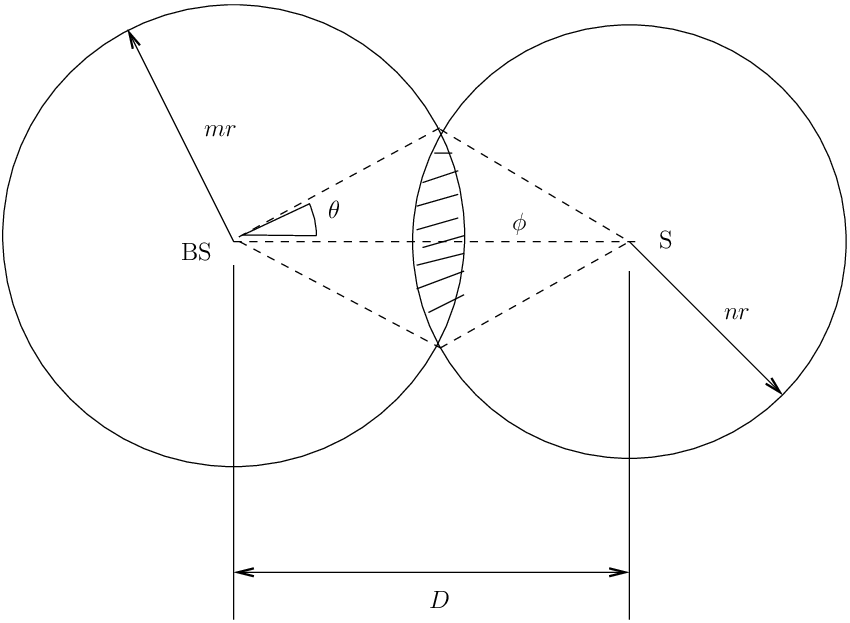}
\caption{The lens shaped intersection between two circles}
\label{fig:shaded_lens}
\end{figure}

Consider the situation shown in Figure~\ref{fig:shaded_lens}. We are interested in the area of the shaded lens shaped region of intersection between two circles of radii $mr$ and $nr$ respectively, $m,n\in\mathbb{N}$, $r\in \Re^+$. Assume without loss of generality, $m \geq n$. Also, assume that the distance $D$ between the centres of the circles satisfies $(m+n)r > D \geq \max(mr,nr)$, so that the circles have a non-zero area of intersection, and neither centre is within the other circle. Let the angles $\theta$ and $\phi$ be as shown in the figure. Let $A_s$ denote the area of the shaded region. Then clearly,

\begin{equation}
\label{eqn:area_upper}
A_s \leq n^2r^2\phi
\end{equation}  

Let $D=ar$, where $a\in \Re^+$. Note that $(m+n) > a \geq m$ (since, $(m+n)r > D \geq \max(mr,nr)\:=\:mr$). 

Now,

\begin{align}
\cos\theta &= \frac{D^2+m^2r^2-n^2r^2}{2Dmr}\nonumber\\
&= \frac{a^2+m^2-n^2}{2am}\nonumber\\
&= \frac{a}{2m}\:+\:\frac{m^2-n^2}{2am}\nonumber
\end{align}

and,

\begin{align}
\cos\phi &= \frac{D^2+n^2r^2-m^2r^2}{2Dnr}\nonumber\\
&= \frac{a^2+n^2-m^2}{2an}\nonumber\\
&= \frac{a}{2n}\:+\:\frac{n^2-m^2}{2an}\nonumber
\end{align}

Thus, 

\begin{align}
\cos\theta - \cos\phi &= \left(\frac{1}{m}-\frac{1}{n}\right)\left[\frac{a}{2}\:+\:\frac{m^2-n^2}{2a}\right]\nonumber\\
&\leq 0\quad\text{since } m\geq n \label{eqn:theta-phi-relation}
\end{align}

Observe that $\theta,\phi\in (0,\frac{\pi}{2})$ (since $n\leq m \leq a < m+n$).

Hence, it follows from Eqn.~\eqref{eqn:theta-phi-relation} that 

\begin{equation}
\label{eqn:theta-phi-rel}
\theta \geq \phi
\end{equation}

Also note that 

\begin{align}
\cos\theta &\geq \frac{a}{2m}\geq \frac{1}{2}\quad\text{since }a\geq m\nonumber\\
\Rightarrow \theta \leq \frac{\pi}{3}\label{eqn:theta-upper}
\end{align}

Finally, combining \ref{eqn:area_upper}, \ref{eqn:theta-phi-rel}, and \ref{eqn:theta-upper}, we have

\begin{align}
A_s &\leq n^2r^2\phi\leq n^2r^2\theta\leq \frac{\pi}{3}n^2r^2\nonumber\\
&= \frac{\pi}{3}[\min(m,n)]^2r^2,\quad\text{since } m \geq n
\end{align}

Hence the claim follows, since the lens $L_i^{(s)}$ is the region of intersection of a circle of radius $ir$ centred at the source $s$, and a circle of radius $(h_{\max}-i)r$ centred at the BS. 
\end{proof} 

Hence, from Equation~\eqref{eqn:p_yi_lowerbound}, 

\begin{align}
\label{eqn:q_expression}
q_i &= P[Y_i\:=\:1|M^o=1,\overline{D}_s > (h_{\max}-1)r]\geq \left(1-\frac{\frac{n_i^2}{3}}{(1-\epsilon)^2(h_{\max})^2}\right)^{m-1},\quad n_i = \min(i,h_{\max}-i)
\end{align}

Finally,
\begin{align}
E[R_{Opt}|\X]&\geq (1-\delta)P[\overline{D}_s > (h_{\max}-1)r]E[R_{Opt}|\Xh], \quad\text{from \eqref{eqn:ropt-lowerbound}}\nonumber\\
&\geq (1-\delta)P[\overline{D}_s > (h_{\max}-1)r]\sum_{i=1}^{h_{\max}-1}q_i,\quad\text{from \eqref{eqn:ropt-lowerbound-inter}}\nonumber\\
&\geq \left[1-\left(\frac{h_{\max}-1}{(1-\epsilon)h_{\max}}\right)^{2m}\right](1-\delta)\sum_{i=1}^{h_{\max}-1}\left(1-\frac{\frac{n_i^2}{3}}{(1-\epsilon)^2h_{\max}^2}\right)^{m-1},\quad\text{from \eqref{eqn:q_expression}}\label{eqn:ropt_final}
\end{align}

where, $n_i = \min(i,h_{\max}-i)$.
\end{proof}

It follows from Lemma 1 and Lemma 2 that:
\begin{theorem}
\label{thm:avg-approx-oneconnect}
The average case approximation ratio of the SPTiRP algorithm over all feasible instances in the stochastic setting described earlier is upper bounded as

\begin{eqnarray}
\alpha &\leq \frac{\overline{N}}{\underline{R}_{Opt}}\label{eqn:avg-approx-oneconnect}
\end{eqnarray}

where, $\overline{N}$ is given by the R.H.S of \eqref{eqn:spt-nodecount-expectn-feasible}, and $\underline{R}_{Opt}$ is given by the R.H.S of \eqref{eqn:ropt-lowerbound-final}.
\end{theorem}

\section{Node Cut based ILP Formulation for RST-MR-HC Problem}
\label{sec:ilp-oneconnect}
We shall formulate the RST-MR-HC problem as an ILP, using certain node cut inequalities (the approach is similar to the one presented in \cite{nigam}). Such a formulation will be useful when the number of potential locations is prohibitively large so that a complete enumeration of all possible solutions to obtain the optimal solution (for comparison against the solution provided by the SPTiRP algorithm) is impractical; in such cases, we can solve the LP relaxation of the ILP to obtain a lower bound on the optimal solution for comparison with the SPTiRP outcome. 

\noindent
We start with a couple of definitions.

\begin{definition}
Given a source and a sink in a graph, a \emph{\textbf{node cut} for that source-sink pair} is defined as a set of nodes whose deletion disconnects the source from the sink \cite{nigam}. 
\end{definition}

\begin{definition}
A \emph{\textbf{minimal node cut} for a source-sink pair} is a node cut which does not contain any other node cut as its subset \cite{nigam}. 
\end{definition}

Consider the graph $G = (Q\cup R, E)$ (notations same as earlier). We define, $\forall k\in Q\backslash\{0\}$, $\forall j\in V\backslash\{k,0\}$,
\begin{equation*}
y_{j,k}=\left\{
\begin{array}{rl}
1 & \text{if node $j$ is selected to connect source $k$ to the sink}\\
0 & \text{otherwise}
\end{array}\right.
\end{equation*}

Let $\mathcal{P}_k, k\in Q\backslash\{0\}$, denote the set of paths from source $k$ to the sink in the graph $G$. A path $p_k\in \mathcal{P}_k$ from source $k$ to sink is said to be \emph{selected} if $y_{j,k}=1\quad\forall j\in p_k$. A source $k$ is said to be connected to the sink if at least one of the paths in $\mathcal{P}_k$ is selected.
 
\begin{theorem}
The following condition is both \emph{necessary and sufficient} for connectivity of all the sources to the sink:

\begin{equation}
\sum_{j\in \gamma}y_{j,k}\geq 1\quad \forall \gamma\in\Gamma^k;\forall k\in Q\backslash\{0\}
\label{eqn:node-cut-ineq}
\end{equation}
where, $\Gamma^k$ is the set of minimal node cuts for a source node $k$.
\end{theorem}

\begin{proof}
We shall only prove the sufficiency. The proof of necessity is as given in \cite{nigam}, where they have stated that the above inequality is a \emph{valid} inequality for the relay node placement problem.

We shall prove by contradiction. Suppose, for an assignment of the variables $y_{j,k}, \forall k\in Q\backslash \{0\}, \forall j\in V\backslash \{k,0\}$, the inequality \eqref{eqn:node-cut-ineq} holds, but at least one source, say source $i$, is not connected to the sink. 

Therefore, for the given assignment of the variables $y_{j,i}, \forall j\in V\backslash\{i,0\}$, no path in the set $\mathcal{P}_i$ got selected. Therefore, for each path $p_i\in \mathcal{P}_i$, there exists at least one node $j\in p_i$ such that $y_{j,i}=0$. Thus, the set of all such nodes from all the paths in $\mathcal{P}_i$ form a node cut for the source $i$ and sink. This node cut will contain a minimal node cut for source $i$ and the sink, say, $\gamma^i_{violated}$ for which $\sum_{j\in \gamma^i_{violated}}y_{j,i}=0$. Thus, inequality \eqref{eqn:node-cut-ineq} is violated for the minimal node cut $\gamma^i_{violated}$, which is a contradiction of our earlier proposition. Hence, if inequality \eqref{eqn:node-cut-ineq} holds for an assignment of the variables $y_{j,k}, \forall k\in Q\backslash \{0\}, \forall j\in V\backslash \{k,0\}$, then all the sources must be connected to the sink for that assignment of variables. 
\end{proof}

We now formulate the ILP as follows:

\begin{align}
\min \quad\sum_{j\in R}y_j\label{obj:ilp}\\
\text{Subject to:}\sum_{j\in \gamma}y_{j,k}&\geq 1\quad \forall \gamma\in\Gamma^k;\forall k\in Q\backslash\{0\}\label{constr:conn}\\
y_j &\geq y_{j,k}\quad \forall j\in R;\forall k\in Q\backslash\{0\}\label{constr:nodeselect}\\
\sum_{j\in V\backslash\{k,0\}}y_{j,k}&\leq h_{\max}-1\quad\forall k\in Q\backslash\{0\}\label{constr:hop}\\
y_{j,k}&\in \{0,1\}\quad \forall k\in Q\backslash\{0\};\forall j\in V\backslash\{k,0\}\label{constr:int1}\\
y_j &\in \{0,1\}\quad \forall j\in R\label{constr:int2} 
\end{align}

Constraint \eqref{constr:conn} in the above formulation ensures connectivity from each source to the sink; constraint \eqref{constr:nodeselect} simply says that a relay node gets selected if it is selected for the path of at least one source; constraint \eqref{constr:hop} ensures that a selected path from a source to the sink has no more than $h_{\max}$ hops; constraints \eqref{constr:int1} and \eqref{constr:int2} are the integer constraints on the node selection variables. The objective function \eqref{obj:ilp} simply minimizes the total number of relay nodes selected.

We shall now show that the optimum value of the objective function for the ILP is indeed the same as the optimum solution (i.e., the minimum number of relays) to the original RST-MR-HC problem.

To do that, we introduce the following notations:

\begin{description}
\item $\F = \{\underline{y}=\{\{y_{j,k}\}_{j\in V\backslash\{k,0\},k\in Q\backslash\{0\}},\{y_j\}_{j\in R}\}: \underline{y} \text{ satisfies constraints \eqref{constr:conn}-\eqref{constr:int2}}\}$: set of all feasible solutions to the ILP

\item $\P_k^{'} = \{p_k: \text{$p_k$ consists of $\leq h_{\max}$ hops from source $k$ to sink}\}\subseteq \P_k$: set of all hop count feasible paths from source $k$ to sink

\item $\U_0 = \{\underline{g}\triangleq \{p_k\}_{k=1}^{|Q|-1}: p_k\in \P_k^{'}\}$: all possible combinations of hop count feasible paths from the sources to the sink
\end{description}

Define a set $\F_0$ in a \emph{one-to-one correspondence} to the set $\U_0$ as follows:

For each $\underline{g}=\{p_k\}_{k=1}^{|Q|-1}\in \U_0$, define $\underline{x}(\underline{g})=\{\{x_{j,k}\}_{j\in V\backslash\{k,0\},k\in Q\backslash\{0\}},\{x_j\}_{j\in R}\}\in\F_0$ such that

\begin{equation*}
x_{j,k}=\left\{
\begin{array}{rl}
1 & \text{if $j\in p_k$}\\
0 & \text{otherwise}
\end{array}\right.
\end{equation*}
\begin{equation*}
x_{j}=\left\{
\begin{array}{rl}
1 & \text{if $x_{j,k}=1$ for some $k\in Q\backslash\{0\}$}\\
0 & \text{otherwise}
\end{array}\right.
\end{equation*}

\begin{lemma}
$\F_0\subseteq \F$
\end{lemma}
\begin{proof}
Verify that any $\underline{x}\in \F_0$ satisfies constraints \eqref{constr:conn}-\eqref{constr:int2}.
\end{proof}

\begin{corollary}
\label{cor:ilpcor1}
$\min_{\underline{y}\in\F}\sum_{j\in R}y_j\leq \min_{\underline{x}\in\F_0}\sum_{j\in R}x_j$
\end{corollary}

Observe that in Corollary~\ref{cor:ilpcor1}, the \emph{L.H.S is the optimum objective function value of the ILP}, whereas the \emph{R.H.S is the optimum solution (i.e., the minimum number of relays) for the RST-MR-HC problem}. Thus, we have proved that the optimum solution to the ILP is a lower bound to the optimum solution to RST-MR-HC problem.

\begin{lemma}
For each $\underline{y}\in\F$, $\exists \:\underline{x}\in\F_0$ such that 
\begin{enumerate}
\item $x_{j,k}\leq y_{j,k}\:\forall j,\forall k$, and hence
\item $\sum_{j\in R}x_j\leq \sum_{j\in R}y_j$
\end{enumerate}
\end{lemma}

\begin{proof}
Given $\underline{y}\in \F$, we can construct paths $p_k\in\P_k^{'},\:k\in Q\backslash\{0\}$ such that $\underline{g}=\{p_k\}_{k=1}^{|Q|-1}\in \U_0$. In doing this, we require constraints \eqref{constr:conn} and \eqref{constr:hop} in the definition of $\F$. Now obtain $\underline{x}\in \F_0$ for this $\underline{g}=\{p_k\}_{k=1}^{|Q|-1}\in \U_0$. Observe that $x_{j,k}\leq y_{j,k}\:\forall j,\forall k$.

Also, since the variables are binary, this implies that $\max_{k\in Q\backslash\{0\}}x_{j,k}\leq \max_{k\in Q\backslash\{0\}}y_{j,k}\:\forall j\in R$, i.e., $x_j\leq y_j\:\forall j\in R$. For otherwise, suppose $\max_{k\in Q\backslash\{0\}}x_{j,k} > \max_{k\in Q\backslash\{0\}}y_{j,k}$ for some $j\in R$. Then that would imply, $\max_{k\in Q\backslash\{0\}}x_{j,k}=1$ and $\max_{k\in Q\backslash\{0\}}y_{j,k}=0$, i.e., for that $j\in R$, $\exists \:k\in Q\backslash\{0\}$ such that $x_{j,k}=1$ and $y_{j,k}=0$. But this contradicts the fact that $x_{j,k}\leq y_{j,k}\:\forall j,\forall k$. Hence the conclusion.

Therefore, it follows that $\sum_{j\in R}x_j\leq \sum_{j\in R}y_j$ 
\end{proof} 

\begin{corollary}
\label{cor:ilpcor2}
\begin{equation*}
\min_{\underline{y}\in\F}\sum_{j\in R}y_j\geq \min_{\underline{x}\in\F_0}\sum_{j\in R}x_j
\end{equation*}
\end{corollary}

\begin{proof}
Suppose $\underline{y}_{opt}=\arg\min_{\underline{y}\in\F}\sum_{j\in R}y_j$. Then, by the above lemma, $\exists \underline{x}^{'}\in\F_0$ such that $\sum_{j\in R}x^{'}_j\leq \sum_{j\in R}y_{opt,j}$. But clearly, $\min_{\underline{x}\in\F_0}\sum_{j\in R}x_j\leq \sum_{j\in R}x^{'}_j$. Hence the proof.
\end{proof}

\begin{theorem}
\label{thm:ilp-main}
\begin{equation*}
\min_{\underline{y}\in\F}\sum_{j\in R}y_j\:=\:\min_{\underline{x}\in\F_0}\sum_{j\in R}x_j
\end{equation*}
\end{theorem}

\begin{proof}
The proof follows by combining Corollaries~\ref{cor:ilpcor1} and \ref{cor:ilpcor2}.
\end{proof}

Theorem~\ref{thm:ilp-main} states that the optimum value of the objective function for the ILP is indeed the same as the optimum solution (i.e., the minimum number of relays) to the original RST-MR-HC problem.

To solve the LP relaxation of this ILP to obtain a lower bound on the optimal solution, we use the algorithm presented in \cite{nigam} (with the Master problem being the ILP represented by Equations~\eqref{obj:ilp}-\eqref{constr:int2}), which uses as a sub-program (to find the node cut constraints iteratively), an algorithm presented by Garg et al.\ \cite{garg} in the context of node weighted multiway cuts.

\section{SPTiRP: Numerical Results}
\label{sec:results}
We performed four sets of experiments to test the SPTiRP algorithm. In all these experiments, the relays and the sources are placed randomly. The first two sets of experiments were performed with a large number of relays, in a setting that conforms to the conditions mentioned in Theorem~\ref{thm:average-case-analysis-setting}, and hence a feasible solution is guaranteed with a high probability. However, due to the large number of relays only a lower bound to the optimal value can be obtained. The third set of experiments were performed with a small number of relays, so that feasibility cannot be assured, but the optimal value can be obtained in every feasible instance. Finally, the fourth set of experiments were performed with a different random graph model (compared to the first three), namely, the Erdos-Renyi random graph model to test the performance of our algorithm on non-geometric input graphs. 

In experiment sets 1 and 2, we need a large number of potential relay locations to ensure the high probability of feasibility. As we had mentioned in Section~\ref{sec:ilp-oneconnect}, for such large problem instances, an exhaustive enumeration of all possible solutions to obtain the optimal solution is impractical. Hence, for these problem instances, we solved the LP relaxations of the corresponding ILPs to obtain \emph{lower bounds} on the optimum relay count. 

In experiment sets 3 and 4, however, the number of potential relay locations, and hence, the problem size was moderate; so we obtained the \emph{exact} optimum relay count for each instance by an exhaustive enumeration technique, starting with the solution provided by the SPTiRP algorithm. The details are provided below.

\subsection{Experiment Set 1}
We generated 100 random networks as follows: we chose $r_{\max}=$ 60 meters, and $h_{\max}=$ 4 for this set of experiments. We also chose $\epsilon = \delta = 0.1$ (see Theorem~\ref{thm:average-case-analysis-setting}). For the chosen parameter values and for an area of $216m \times 216m$, the required number of potential relay locations was found to be $n(\epsilon,\delta,h_{\max},r_{\max})\geq 1908$. Hence, 1908 potential relay locations were selected uniformly randomly over a $216m \times 216m$ area. This ensures that any point within a distance $(1-\epsilon)h_{\max}r_{\max}$ from the BS is at most $h_{\max}$ hops away from the BS with a high probability ($\geq (1-\delta)=0.9$). 10 source nodes were deployed uniformly randomly over the quarter circle of radius $(1-\epsilon)h_{\max}r_{\max}=216m$; hence we have a feasible solution with a high probability ($\geq 0.9$). 

The SPTiRP algorithm was run on the 100 scenarios thus generated; none of the 100 scenarios tested turned out to be infeasible. For each scenario, a lower bound on the optimum relay count was obtained by solving the LP relaxation of the corresponding ILP formulation as described in Section~\ref{sec:ilp-oneconnect}.

The results are summarized in Table~\ref{tbl:efficiency1}.

\begin{table}[ht]
  \centering
\caption{Test Set 1: Performance of the SPTiRP algorithm Compared to Lower Bound on Optimal Solution}
\label{tbl:efficiency1}
\footnotesize
  \begin{tabular}{|c|c|c|c|c|c|}\hline
    Potential & Scenarios & Optimal Design & Off by one & Max off \\
   Relay && matched with & from & from \\
   count && lower bound & lower bound & lower bound\\
    \hline
   1908 & 100 & 23 & 21 & 10\\
   
  \hline
\end{tabular}
\normalsize
\end{table} 

\noindent

\textbf{Observations}
\begin{enumerate}
\item In 44\% of the tested scenarios, the algorithm ends up giving optimal or near-optimal (exceeding optimum just by one relay) solutions. However, note that the comparison was only against a lower bound on the optimal solution, which can potentially be loose depending on the problem scenario, and we suspect the actual performance of the algorithm to be much better (indeed, as we shall see in Experiment Set 3 by comparing against the actual optimal solution, the algorithm performed close to optimal in most of the tested scenarios).   
\item In the remaining cases, where it is off by more than one relay, the maximum difference from the lower bound was found to be 10 relays.
\item We computed the empirical worst case approximation factor from the experiments as follows: for each scenario, we computed the approximation factor given by the SPTiRP algorithm w.r.t the lower bound obtained from the LP relaxation as \emph{approximation factor} $= \frac{Relay_{Algo}}{Relay_{lowerbound}}$. The maximum of these over all the tested scenarios (in the current set of experiments) was taken to be the (empirical) worst case approximation factor. 
\item We also computed the theoretical bound on the average approximation ratio for the given setting and parameter values using Equation~\ref{eqn:avg-approx-oneconnect}, and compared it against the empirical average case approximation ratio obtained from the experiments as 
\begin{equation}
\label{eqn:empirical-avg-approx}
\text{Empirical average case approx. ratio}\:=\:\frac{\text{Average relaycount of SPTiRP over 100 scenarios}}{\text{Average lower bound from LP relaxation over 100 scenarios}}
\end{equation}
The results are summarized in Table~\ref{tbl:approx-ratio1}. 
\end{enumerate}

\begin{table}[ht]
  \centering
\caption{Test Set 1: Approximation ratio for the SPTiRP algorithm}
\label{tbl:approx-ratio1}
\footnotesize
  \begin{tabular}{|c|c|c|c|c|c|}\hline
    Potential & Scenarios & \multicolumn{2}{c|}{Worst case} & \multicolumn{2}{c|}{Average case} \\
   Relay && \multicolumn{2}{c|}{approximation ratio} & \multicolumn{2}{c|}{approximation ratio}\\
   count && Theoretical & Experimental & Theoretical bound (Eqn.~\eqref{eqn:avg-approx-oneconnect})& Experimental(Eqn.~\eqref{eqn:empirical-avg-approx})\\
    \hline
   1908 & 100 & 30 & 5 & 14 & 1.66\\ 
  \hline
\end{tabular}
\normalsize
\end{table} 

In Table~\ref{tbl:exectime_one_1}, we have compared the execution time of the SPTiRP algorithm against the time required to compute a lower bound on the optimal solution by solving the LP relaxation. Both the algorithms were run in MATLAB 7.11 on the Sankhya cluster of the ECE Department, IISc, using a single compute node (linux based) with 16 GB main memory, and a single processor with 4 cores, i.e., 4 CPUs. As can be seen from the table, while the SPTiRP algorithm computes a very good (often optimal) solution in at most a few seconds, computing even the lower bound on the optimal solution (i.e., solving the LP relaxation instead of the actual ILP) can be actually quite time consuming, running into several hours (upto about 12 hours in the worst case). 

\begin{table*}[ht]
  \centering
\caption{Test Set 1: Computation time of the SPTiRP algorithm compared to Optimal solution (lower bound) computation}
\label{tbl:exectime_one_1}
\footnotesize
  \begin{tabular}{|c|c|c|c|c|c|}\hline
Potential & Scenarios & Mean execution time & Mean Execution time & Max execution time & Max execution time \\
   Relay && of SPTiRP & of obtaining  & of SPTiRP & of obtaining \\
          &&           & a lower bound on optimal solution &          & a lower bound on optimal solution\\
   Count &&  in sec & in sec & in sec & in sec\\
    \hline
   1908 & 100 & 6.6621 & 7002 & 18.4438 & 41716\\
  \hline
\end{tabular}
\normalsize
\end{table*}

\subsection{Experiment Set 2}
\label{subsec:sptirp-expt2}
The setting for this set of experiments is very similar to that in Experiment Set 1, except that now the sources were also deployed over the same square area as the potential relay locations, instead of a quarter circle. 

We generated 100 random networks as follows: we chose $r_{\max}=$ 60 meters, and $h_{\max}=$ 4 for this set of experiments. We also chose $\epsilon = \delta = 0.1$. For the chosen parameter values and for an area of $150m \times 150m$, the required number of potential relay locations was found to be $n(\epsilon,\delta,h_{\max},r_{\max})\geq 920$. Hence, 920 potential relay locations were selected uniformly randomly over a $150m \times 150m$ area. This ensures that any point within a distance $(1-\epsilon)h_{\max}r_{\max}$ from the BS is at most $h_{\max}$ hops away from the BS with a high probability ($\geq (1-\delta)=0.9$). 10 source nodes were deployed uniformly randomly over the $150m \times 150m$ area. Note that all the sources are within a radius $(1-\epsilon)h_{\max}r_{\max}\:=\:216$ meters from the BS, since the diagonal of the deployment area is less than 216 meters; hence we have a feasible solution with a high probability ($\geq 0.9$). 

The SPTiRP algorithm was run on the 100 scenarios thus generated; none of the 100 scenarios tested turned out to be infeasible. For each scenario, a lower bound on the optimum relay count was obtained by solving the LP relaxation of the corresponding ILP formulation as described in Section~\ref{sec:ilp-oneconnect}.

The results are summarized in Table~\ref{tbl:efficiency2}.

\begin{table}[ht]
  \centering
\caption{Test Set 2: Efficiency of the SPTiRP algorithm in obtaining the optimal design}
\label{tbl:efficiency2}
\footnotesize
  \begin{tabular}{|c|c|c|c|c|c|}\hline
    Potential & Scenarios & Optimal Design & Off by one & Max off \\
   Relay && w.r.t & from & from \\
   count && lower bound & lower bound & lower bound\\
    \hline
   920 & 100 & 82 & 15 & 2\\
   
  \hline
\end{tabular}
\normalsize
\end{table} 

\noindent

\textbf{Observations}
\begin{enumerate}
\item In over 97\% of the tested scenarios, the algorithm ends up giving optimal or near-optimal (exceeding optimum just by one relay) solutions.
\item In the remaining cases, where it is off by more than one relay, the maximum difference was found to be 2 relays.
\item We computed the empirical worst case approximation ratio in the same manner as was done in Experiment Set 1.
\item We also computed the theoretical bound on the average approximation ratio for the given setting and parameter values using Equation~\ref{eqn:avg-approx-oneconnect}, and compared it against the empirical average case approximation ratio obtained from the experiments as was done in Experiment Set 1. 

The results are summarized in Table~\ref{tbl:approx-ratio2}. 
\end{enumerate}

\begin{table}[ht]
  \centering
\caption{Test Set 2: Approximation ratio for the SPTiRP algorithm}
\label{tbl:approx-ratio2}
\footnotesize
  \begin{tabular}{|c|c|c|c|c|c|}\hline
    Potential & Scenarios & \multicolumn{2}{c|}{Worst case} & \multicolumn{2}{c|}{Average case} \\
   Relay && \multicolumn{2}{c|}{approximation ratio} & \multicolumn{2}{c|}{approximation ratio}\\
   count && Theoretical & Experimental & Theoretical bound (Eqn.~\eqref{eqn:avg-approx-oneconnect}) & Experimental (Eqn.~\eqref{eqn:empirical-avg-approx})\\
    \hline
   920 & 100 & 30 & 2 & 14 & 1.13\\ 
  \hline
\end{tabular}
\normalsize
\end{table} 

In Table~\ref{tbl:exectime_one_2}, we have compared the execution time of the SPTiRP algorithm against the time required to compute a lower bound on the optimal solution by solving the LP relaxation. Both the algorithms were run in MATLAB 7.11 on the Sankhya cluster of the ECE Department, IISc, using a single compute node (linux based) with 16 GB main memory, and a single processor with 4 cores, i.e., 4 CPUs. As can be seen from the table, while the SPTiRP algorithm computes a very good (often optimal) solution in at most a few seconds, computing even the lower bound on the optimal solution (i.e., solving the LP relaxation instead of the actual ILP) was quite time consuming, running well beyond an hour. 

\begin{table*}[ht]
  \centering
\caption{Test Set 2: Computation time of the SPTiRP algorithm compared to Optimal solution (lower bound) computation}
\label{tbl:exectime_one_2}
\footnotesize
  \begin{tabular}{|c|c|c|c|c|c|}\hline
    Potential & Scenarios & Mean execution time & Mean execution time & Max execution time & Max execution time \\
   Relay && of SPTiRP & of obtaining & of SPTiRP & of obtaining \\
          &&           & a lower bound on optimal solution &   & a lower bound on optimal solution\\
   Count &&  in sec & in sec & in sec & in sec\\
    \hline
   920 & 100 & 2.4222 & 2489.2 & 5.5684 & 5902.4\\
  \hline
\end{tabular}
\normalsize
\end{table*}

\subsection{Experiment Set 3}
\label{subsec:sptirp-expt3}
In this set of experiments we deployed a smaller number of relays randomly. Due to the small number of relays, the probabilistic analysis of feasibility is not useful. We generated 1000 random networks as follows: A $150m \times 150m$ area is partitioned into square cells of side 10$m$. Consider the lattice created by the corner points of the cells. 10 source nodes are placed at random over these lattice points. Then the potential relay locations are obtained by selecting $n$ points uniformly randomly over the $150m \times 150m$; $n$ was varied from 100 to 140 in steps of 10, and for each value of $n$, we generated 200 random network scenarios (thus yielding 1000 test cases). We chose $r_{\max}=$ 60 meters, and $h_{\max}=$ 6 for the experiments. 

Given the outcome of the SPTiRP algorithm, an optimal solution can be obtained as follows: Suppose the SPTiRP uses $n$ relays. Then perform an exhaustive search over all possible combinations of $(n-1)$ and fewer relays to check if the performance constraints can still be met. 

In none of the 1000 scenarios tested, the hop constraint turned out to be infeasible.The results are summarized in Table~\ref{tbl:efficiency}.

\begin{table}[ht]
  \centering
\caption{Test Set 3: Efficiency of the SPTiRP algorithm in obtaining the optimal design}
\label{tbl:efficiency}
\footnotesize
  \begin{tabular}{|c|c|c|c|c|c|}\hline
    Potential & Scenarios & Optimal Design & Off by one & Max off \\
   Relay &&&& from \\
   count &&&& optimal\\
    \hline
   100 & 200 & 154 & 42 & 3\\
   110 & 200 & 154 & 40 & 2\\
   120 & 200 & 158 & 39 & 2\\
   130 & 200 & 155 & 36 & 2\\
   140 & 200 & 161 & 38 & 2\\
   \hline
   Total & 1000 & 782 & 195 & 3\\
  \hline
\end{tabular}
\normalsize
\end{table} 

\begin{figure}[ht]
\begin{center}
\includegraphics[scale=0.4]{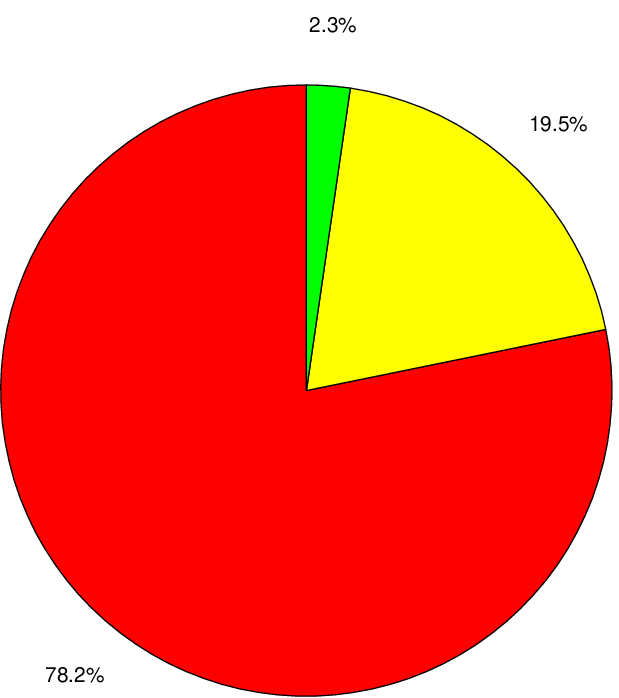}
\end{center}
\caption{ Efficiency of the Algorithm as suggested by Test Results; Red: SPTiRP gives optimal solution; Yellow: SPTiRP is off from optimum by one relay; Green: SPTiRP off from optimum by 2 or more relays}
\label{fig:efficiency}
\end{figure}

The efficiency of the algorithm can be easily visualized from the pie chart in Figure~\ref{fig:efficiency}.

\noindent

\textbf{Observations}
\begin{enumerate}
\item As in the case of test set 1, even for test set 2, in over 97\% of the tested scenarios, the algorithm ends up giving optimal or near-optimal (exceeding optimum just by one relay) solutions.
\item In the remaining cases, where it is off by more than one relay, the maximum difference was found to be 3 relays.
\end{enumerate}

In Table~\ref{tbl:exectime_one}, we have compared the execution time of the SPTiRP algorithm against the time required to compute an optimal solution, given the outcome of the SPTiRP algorithm. Both the SPTiRP algorithm, and the postprocessing on its outcome were run in MATLAB 7.0.1 on a Windows Vista (basic) based PC (Dell Inspiron 1525) having Intel Core 2 Duo T5800 CPU with processor speed of 2 GHz, and 3 GB RAM. Again, while the SPTiRP algorithm computed a very good (often optimal) solution in at most a second or two (averaging less than a second), computing the optimal solution even after being provided with a very good upper bound on the required number of relays by SPTiRP, turned out to be quite time consuming, running into several minutes. 

\begin{table*}[ht]
  \centering
\caption{Test Set 3: Computation time of the SPTiRP algorithm compared to Optimal solution computation}
\label{tbl:exectime_one}
\footnotesize
  \begin{tabular}{|c|c|c|c|c|c|}\hline
    Potential & Scenarios & Mean execution time & Mean execution time & Max execution time & Max execution time \\
   Relay && of SPTiRP & of directly obtaining & of SPTiRP & of directly obtaining \\
          &&          & an optimal solution  &          & an optimal solution\\
   Count &&  in sec & in sec & in sec & in sec\\
    \hline
   100 & 200 & 0.58812 & 661.485 & 1.638 & 1828.7\\
   110 & 200 & 0.70544 & 240.85 & 2.081 & 722.29\\
   120 & 200 & 0.81154 & 423.89 & 1.591 & 944.74\\
   130 & 200 & 0.99343 & 951.495 & 2.606 & 2674.9\\
   140 & 200 & 1.1438 & 140.7 & 2.808 & 355.46\\
   \hline
   Overall & 1000 & 0.84847 & 483.684 & 2.808 & 2674.9\\
  \hline
\end{tabular}
\normalsize
\end{table*}

Also, we note from Table~\ref{tbl:exectime_one} that, as the node density increases, the computation time of the SPTiRP algorithm also increases.

\subsection{Experiment Set 4}
\label{subsec:sptirp-expt4}
This set of experiments were performed using the Erdos-Renyi random graph model to generate the input graphs. We generated 500 random networks as follows: A $150m \times 150m$ area is partitioned into square cells of side 10$m$. Consider the lattice created by the corner points of the cells. 10 source nodes are placed at random over these lattice points. Then the potential relay locations are obtained by selecting $n$ points uniformly randomly over the $150m \times 150m$; $n$ was varied from 100 to 140 in steps of 10, and for each value of $n$, we generated 100 random network scenarios (thus yielding 500 test cases). For each instance, the edges in the input graph were selected iid with probability 0.5, i.e., for each possible (unordered) node pair $(i,j)$, the edge $(i,j)$ was chosen to be a feasible edge with probability 0.5. We chose $h_{\max}=$ 6 for the experiments. Note that apart from the method used for creating the feasible edges, the setting is similar to that in Experiment Set 3. 

Given the outcome of the SPTiRP algorithm, an optimal solution can be obtained in the same manner as in Experiment Set 3. 

In none of the 500 scenarios tested, the hop constraint turned out to be infeasible. Our observations are summarized below.

\gap
\noindent
\textbf{Observations:}
\begin{enumerate}
\item In each of the 500 test cases, the SPTiRP algorithm returned an optimal solution.
\item In 491 cases, no relay was required. In the remaining cases, only one relay was required. 
\end{enumerate}

In Table~\ref{tbl:exectime_one}, we have presented the execution time of the SPTiRP algorithm. Note that since the outcome of the algorithm in all the test cases were zero or one relay, computing the optimum solution given the outcome of the algorithm was trivial. The SPTiRP algorithm was run in MATLAB R2011b on a Windows Vista (basic) based PC (Dell Inspiron 1525) having Intel Core 2 Duo T5800 CPU with processor speed of 2 GHz, and 3 GB RAM. 
\begin{table*}[ht]
  \centering
\caption{Test Set 4: Computation time of the SPTiRP algorithm}
\label{tbl:exectime_one}
\footnotesize
  \begin{tabular}{|c|c|c|c|c|c|}\hline
    Potential & Scenarios & Mean execution time & Max execution time\\
   Relay && of SPTiRP & of SPTiRP \\
   Count &&  in sec & in sec\\
    \hline
   100 & 100 & 0.1199 & 0.4412\\
   110 & 100 & 0.1144 & 0.4403\\
   120 & 100 & 0.1279 & 0.3767\\
   130 & 100 & 0.1204 & 0.2681\\
   140 & 100 & 0.1408 & 0.3675\\
   \hline
   Overall & 500 & 0.1247 & 0.4412\\
  \hline
\end{tabular}
\normalsize
\end{table*}

\section{SPTiRP: Simulation Results}
\label{sec:simulation}
To test the QoS under positive traffic arrival rates, of the network topologies obtained using SPTiRP algorithm, we performed extensive simulations using Qualnet v4.5 \cite{qualnet}. For these simulations, we assumed the PHY and MAC layers to be as specified in the IEEE~802.15.4 standard\cite{IEEE}. 

We generated 20 network topologies as follows: in each case, 10 source nodes, and 120 potential relay locations were randomly selected in a $150m \times 150m$ area in exactly the same way as described in Section~\ref{subsec:sptirp-expt3}. As before, the BS was assumed to be at the corner $(0,0)$. We chose the \textbf{maximum communication range, $r_{\max}=30$ meters}, which, for a \textbf{transmit power of 0 dBm}, and a \textbf{PHY layer packet size of 131 bytes}, corresponds to a \textbf{PER of $\leq 1\%$}, assuming the path loss model given in the standard \cite{IEEE,IEEE152}, a \textbf{fade margin of 20 dB}, and \textbf{receiver sensitivity of $-98.8$ dBm}. The \emph{hop constraint was chosen as $h_{\max}=9$}, which, for a PER of 1\%, corresponds to an \emph{end-to-end delivery probability of 91.35\% (under the lone packet model)}, and an \emph{end-to-end mean delay of 56.16 msec, also under the lone packet model}, assuming the CSMA/CA backoff parameters given in the standard, and a PHY layer packet size of 131 bytes (see \cite{abhijitmeth} for details of how this mean end-to-end delay can be computed). Having chosen $r_{\max}$, we had a graph on the sources, and the potential relay locations. We used the SPTiRP algorithm on this network graph with the above mentioned hop constraint, to obtain a tree topology connecting the sources to the BS using a small number of relays, and satisfying the hop constraint.  

Qualnet simulation was performed on each of the 20 network topologies thus generated, for \textbf{six different traffic arrival rates}, namely, $\lambda =$0.1, 0.2, 0.3, 0.4, 0.5, and 2 packets/sec from each source. The \emph{arrival process was assumed to be Poisson}. The simulation procedure is described below:

\begin{enumerate}
\item We used the following \emph{interference model} in Qualnet: any two nodes that are within \emph{Carrier Sense} (CS) range of each other can hear each other's transmission. If two nodes are within the CS range of a receiver node, then their transmissions interfere with each other at the receiver node. The CS range, $r_{cs}$, was set equal to $r_{\max}$ for the simulations (see above).
\item We used the \emph{collision model} in Qualnet to account for packet losses due to interference. If two or more packet transmissions interfere with one another at a receiver node, then all of those packets are lost. 
\item For each topology, and each arrival rate, \emph{the simulation was repeated for 25 iterations, with each iteration being run for 1500 seconds of simulated time.}.
\item For each topology, and each arrival rate, we recorded the end-to-end delivery probability (we shall use the shorthand $p_{del}$ for this from now on, with slight abuse of notation) from each source to the sink, averaged over the 25 iterations, and the mean end-to-end packet delay from each source to the sink, also averaged over 25 iterations. 

The results are summarized in Table~\ref{tbl:simulation_summary}. To keep the table concise, we have adopted the following strategy: for each arrival rate and each topology, we have computed $p_{del}$ \emph{averaged over the 10 sources}, and reported only the minimum average $p_{del}$ for each rate, the minimum being taken over the 20 scenarios. This constitutes column 3 of Table~\ref{tbl:simulation_summary}. A similar strategy has been adopted for reporting the end-to-end delay (column 6 of the table). We have also reported the minimum $p_{del}$ and the maximum delay \emph{encountered over all sources and all the 20 scenarios} for each rate (columns 2 and 5 respectively), and the maximum $p_{del}$ and the minimum delay \emph{encountered over all sources and scenarios} for each rate (columns 4 and 7 respectively). 
\end{enumerate} 

\begin{table*}[ht]
  \begin{center}
\caption{Summary of Qualnet simulation results for 20 network topologies}
\label{tbl:simulation_summary}
\footnotesize
  \begin{tabular}{|c|c|c|c|c|c|c|}\hline
    Arrival & Minimum & Minimum & Maximum & Maximum & Maximum & Minimum\\
   rate & $p_{del}$ & average $p_{del}$ & $p_{del}$ & delay & average delay & delay\\
   in pkts/sec &         &         &         & in sec & in sec & in sec\\
    \hline
   0.1 & 0.874 & 0.905 & 0.982 & 0.0509 & 0.0401 & 0.0113\\
   0.2 & 0.860 & 0.893 & 0.980 & 0.0510 & 0.0401 & 0.0113\\
   0.3 & 0.846 & 0.879 & 0.981 & 0.0511 & 0.0402 & 0.0113\\
   0.4 & 0.826 & 0.865 & 0.980 & 0.0513 & 0.0403 & 0.0114\\
   0.5 & 0.802 & 0.848 & 0.978 & 0.0514 & 0.0404 & 0.0114\\
   2.0 & 0.557 & 0.667 & 0.967 & 0.0535 & 0.0417 & 0.0117\\
   \hline
\end{tabular}

\end{center}
\normalsize
\end{table*}

\gap
\noindent
\textbf{Observations:}
\begin{enumerate}
\item From Table~\ref{tbl:simulation_summary}, we observe that the mean end-to-end delay \emph{never exceeded} the lone-packet target end-to-end delay of 56.16 msec. 
\item For low arrival rates, the minimum $p_{del}$ violated the lone-packet target $p_{del}$ only by a small margin. For each rate, we quantify this margin of violation as follows: 
\begin{equation*}
\text{Percentage violation in $p_{del}$} = \frac{\text{Lone-packet target $p_{del}$} - \text{Minimum $p_{del}$ under current arrival rate}}{\text{Lone-packet target $p_{del}$}}\times 100
\end{equation*}
These results are summarized in Table~\ref{tbl:pdel-vilation}.

\begin{table*}[ht]
  \begin{center}
\caption{Proximity of positive traffic QoS to lone-packet target QoS}
\label{tbl:pdel-vilation}
\footnotesize
  \begin{tabular}{|c|c|}\hline
    Arrival & Maximum percent violation in $p_{del}$\\
   rate & w.r.t lone-packet target\\
   in pkts/sec &  (over the 20 scenarios tested)\\
    \hline
   0.1 &	4.3721\\
   0.2 & 5.8784\\
   0.3 & 7.3715\\
   0.4 &	9.5948\\
   0.5 & 12.1749\\
   2.0 & 39.0019\\
   \hline
\end{tabular}
\end{center}
\normalsize
\end{table*}

\item From Table~\ref{tbl:pdel-vilation}, we note that the tested network topologies, although designed for the lone-packet model, \emph{can handle light positive traffic arrival rate (upto 0.4 packets/sec) from each source, without exceeding the lone-packet QoS by more than 10\%}. 
\end{enumerate}

\section{RSN$k$-MR-HC: Complexity of Obtaining a Feasible Solution}
\label{sec:formulation_kconnect}

We have already seen that the RSN$k$-MR-HC problem is NP-Hard. In this Section, we shall show that even the problem of obtaining a \emph{feasible} solution (as opposed to an optimal solution) to the RSN$k$-MR-HC problem is NP-Complete. Hence, we cannot hope for a polynomial time \emph{approximation algorithm} for the general RSN$k$-MR-HC problem, and have to resort to developing good heuristics instead. 

\begin{lemma}
\label{l2}
Given a graph $G= (V, E)$, specified vertices $s$ and $t$, positive integers $k \geq 2$, and $H \leq |V|$. 
The problem to determine if $G$ contains $k$ or more mutually vertex disjoint paths from $s$ to $t$, none involving more than $H$ edges is NP-Complete for all fixed $H\geq 5$.
\end{lemma}
\begin{proof}
See [Itai, Perl, and Shiloach, 1977]. The proof there is via transformation from 3 Satisfiability.
\end{proof}

\begin{corollary}
\label{cor1}
Given an edge weighted graph $G= (V, E)$, with $V= Q\cup R$, the problem of finding a subgraph with $k\geq 2$ vertex disjoint paths from each source node to sink such that each path has hop count $\leq h_{\max}$ (RSN$k$-HC) is NP-Complete. 
\end{corollary}
\begin{proof}
From Lemma~\ref{l2}, it follows that this problem is NP-Complete for any $h_{max}\geq 5$. Hence, the general problem is NP-Complete.
\end{proof}

\begin{corollary}
\label{thm1}
Unless $P=NP$, there does not exist any polynomial time complexity algorithm for providing a feasible solution to the RSN$k$-MR-HC problem.
\end{corollary}

\begin{corollary}
\label{thm2}
Unless $P=NP$, no polynomial time complexity algorithm can provide finite approximation guarantee for the general RSN$k$-MR-HC problem.
\end{corollary}

\begin{proof}
Suppose, in a certain instance of the RSN$k$-MR-HC problem, the source nodes alone are sufficient to meet the design requirements, i.e., there exist node disjoint hop constrained paths from each source to the sink, involving only other source nodes, and no additional relay nodes. But, from Corollary~\ref{cor1}, it follows that no polynomial time algorithm is guaranteed to predict the existence of such a solution even when there exists one. 

Therefore (unlike the SPTiRP algorithm which uses zero relays whenever the optimal solution uses zero relays), in this case, one might end up using a non zero number of relay nodes despite the fact that the optimal solution uses zero relays. Hence, for the general RSN$k$-MR-HC problem, a polynomial time algorithm cannot provide finite approximation guarantee.  
\end{proof}

\section{E-SPTiRP: A Polynomial Time Heuristic for RSN$k$-MR-HC}
\label{sec:algorithms}

We propose Extended SPTiRP, a polynomial time heuristic for the RSN$k$-MR-HC problem, which builds on the SPTiRP algorithm for one connectivity, described in Section~\ref{sec:sptalgo}. Before we discuss the algorithm, we describe below two limitations that are common to any \emph{polynomial time} algorithm for the RSN$k$-MR-HC problem.

Recall from the proof of Corollary~\ref{thm2} that no polynomial time algorithm for the RSN$k$-MR-HC problem is guaranteed to predict the existence of a solution involving only the source nodes whenever such a solution exists.

Also, from Corollary~\ref{thm1}, it follows that unless $P=NP$, no polynomial time algorithm for the RSN$k$-MR-HC problem is guaranteed to find a feasible solution whenever there exists one. Therefore, if a polynomial time heuristic for the RSN$k$-MR-HC problem fails to find a feasible solution, we shall say that the problem is \emph{possibly infeasible}. 

However, note that it is possible to determine, in polynomial time, if the corresponding RST-MR-HC problem (i.e., the problem of obtaining a one-connected, hop constrained subgraph) is feasible (simply by computing the shortest path tree, and checking if the paths from the sources to the sink therein meet the hop count bound). Hence, \emph{if we cannot find even one path with desired hop count bound from some of the sources to the sink, then the problem is actually infeasible}.

\subsection{Algorithm E-SPTiRP}

Since the algorithm consists of many steps, we organize the presentation of the algorithm as follows: we shall first present the key steps of the algorithm in Section~\ref{subsec:ideamain}. Then, each key step will be explained in more detail along with pseudo code in the subsequent subsections.

\subsubsection{Main Idea/Key Steps in the Algorithm}
\label{subsec:ideamain}

Given $G= (V, E)$, where $V=Q\cup R$, connectivity requirement $k$, and hop constraint $h_{\max}$

\begin{enumerate}
\item \textbf{Phase 1: Checking for $k$-connectivity on $Q$ alone}

Check for $k$ connectivity with hop constraint on $Q$ alone
\begin{itemize}
\item If the answer to this step is positive, done
\item Else go to the next step
\end{itemize}
See Section~\ref{subsubsec:phase1}, Steps 1-6 of the pseudo code for details of this phase.
\item \textbf{Phase 2: Obtaining node-disjoint paths, with a small relay count, from each source to the sink}
We come to this phase if Phase 1 fails to find a network on $Q$ alone satisfying the design objectives. Our objective in this phase is to obtain $k$ node disjoint, hop constrained paths from each source to the sink, using as few additional relays as possible. To that end, we proceed as follows.
\begin{enumerate}
\item \textbf{Obtaining a one-connected, hop constrained network: }Run the SPTiRP algorithm on the entire graph of sources and potential relay locations, to obtain a one connected hop constrained network with a small number of relays. If this step completes successfully, we get a hop constrained path from each source to the sink.
\begin{itemize}
\item If the SPTiRP algorithm returns failure, we can declare the problem to be \emph{infeasible}, and stop, as we could not even obtain a one connected network satisfying the hop constraint.
\end{itemize}
See Section~\ref{subsubsec:phase2}, Steps 1-2 of the pseudo code for more details.
\item \textbf{Obtaining alternate node disjoint, hop constrained routes: }Next, \emph{for each source}, we aim to obtain an alternate node disjoint hop count feasible path to the sink.
\begin{itemize} 
\item If we fail to find an alternate hop count feasible node disjoint route for some of the sources, we declare the problem to be \emph{possibly infeasible}, and stop.
\end{itemize} 
Details of this step are provided in Section~\ref{subsubsec:phase2}, Steps 3-11 of the pseudo code.
\item \textbf{Relay pruning: }If we can find an alternate hop count feasible node disjoint route from a source to the sink, we start with that feasible solution, and aim to prune relays, while retaining hop count feasibility, in order to obtain a better solution in terms of relay count. 

\begin{itemize}
\item we pursue a relay pruning strategy wherein, we aim to prioritize the reuse of relays used by the solution so far, and minimize the use of relays that are
unused as yet. See Section~\ref{subsubsec:phase2}, Steps 12-16 of the pseudo code for detailed procedure.
\end{itemize}

\item This procedure is repeated until all the sources have $k$ node disjoint hop constrained paths to the sink, or the problem has been declared (possibly) infeasible.

\end{enumerate}

\end{enumerate}

Next, we explain each of the above steps in more details.

\subsubsection{\textbf{Phase 1: Checking for $k$-connectivity on $Q$ alone}}
\label{subsubsec:phase1}
In this phase, we shall check if the design objectives ($k$ connectivity with hop constraint) can be met using only the source nodes, and no additional relays. In other words, we aim at finding $k$ node disjoint hop constrained paths from each source to the sink, \emph{using only other source nodes}. 

\gap
\noindent
\fbox{
\begin{minipage}[t]{3in}
\textbf{Input:} $G_Q=(Q,E_Q)$, $h_{\max}$, $k$\\
\comment $E_Q$ is the set of all edges of length $\leq r_{\max}$ on $Q$\\
\textbf{Output:} $T$ (the desired network)\\
\textbf{flags:} boolean $F_{inf}$ (if $F_{inf}= 1$, no feasible solution found)\\
\textbf{Initialize:} $T=\varnothing$, $F_{inf} = 0$

\gap
\noindent
\textbf{Outer loop: } \emph{for} each source $S_i$, $1 \leq i \leq |Q|-1$ (\comment the following steps will be repeated for each source)

\gap
\step {1}: $\overline{Q}=\{S_i,0\}$; $\overline{R}=Q \backslash \overline{Q}$
\end{minipage}}

\gap
\noindent
\remark For each source, we treat all the other sources as relays (the set $\overline{R}$), and try to obtain $k$ node disjoint hop constrained paths from the source to the sink, using the nodes in $\overline{R}$. 

\gap
\noindent
\fbox{
\begin{minipage}[t]{3in}
\step {2}: $l=1$ ($l$ is the loop variable for the inner loop, described next)

\noindent
\textbf{Inner loop:} \emph{while} $l\leq k$ (\comment the following steps (Steps 3 to 6) will be repeated until we have $k$ node disjoint paths from source $i$ to sink) 

\gap
\step {3}: $path_{\max}(S_i,0) \leftarrow SPTiRP(\overline{Q}\cup \overline{R}, E_Q)$
\end{minipage}}

\gap
\noindent
\remark Treating the remaining sources as relays, we run the SPTiRP algorithm to obtain the $l^{th}$ node disjoint path ($path_{\max}(S_i,0)$) from source $i$ to sink; the reason for using the SPTiRP algorithm is to use as few nodes as possible from the set $\overline{R}$, so that there are enough nodes left to construct the $(l+1)^{th}$ node disjoint path in the next iteration of the inner loop.

\gap
\noindent
\fbox{
\begin{minipage}[t]{3in}
\step {4}: \emph{if} hopcount($path_{\max}(S_i,0)$) $> h_{\max}$

\hspace{9mm} $\{F_{inf}\leftarrow 1$;

\hspace{10mm}\emph{exit Phase 1}\}

\gap 
\hspace{9mm}\emph{else} go to next Step
\end{minipage}}

\gap
\noindent
\remark Note that if the hop constraint cannot be met in the first iteration (of the inner loop) itself (i.e., for $l=1$), it implies that that the shortest path from source $i$ to sink using only the other source nodes does not satisfy the hop constraint. Then, we can conclude for sure that the design objectives cannot be met using only the source nodes, and we can proceed to Phase 2 of the algorithm. 

However, if the hop constraint is met in the first iteration, and cannot be satisfied in some subsequent iteration (i.e., for some $l>1$), we cannot conclude for sure that $Q$ alone was not sufficient to meet the design requirements (recall Corollary~\ref{cor1}, and our discussion at the beginning of Section~\ref{sec:algorithms}). All we can say at this point is that Phase 1 of our algorithm failed to find a feasible solution on $Q$ alone, and therefore, we shall proceed to the next phase of the algorithm, assuming that the problem on $Q$ alone is \emph{possibly} infeasible.

\gap
\noindent
\fbox{
\begin{minipage}[t]{3in}
\step {5}: $T\leftarrow T\cup path_{\max}(S_i,0)$

\comment We augment the network with the current feasible path. 

\gap
\step {6}: $R_{used}\leftarrow \overline{R}\cap path_{\max}(S_i,0)$

\hspace{10mm}$ \overline{R}\leftarrow \overline{R}\backslash R_{used}$

\hspace{10mm}$l\leftarrow l+1$

\comment We identify the nodes in $\overline{R}$ used in the current path from source $i$ to sink, and remove them from $\overline{R}$ before proceeding to the next iteration of the inner loop. 
\end{minipage}}

\gap
\noindent
\remark The above step is necessary since in each iteration, we need to identify \emph{node disjoint} paths from the source to the sink.   

At the end of Phase 1, we either have a network $T$, consisting only of the source nodes, and meeting the design requirements (in which case, we are done), or we find that the problem on $Q$ alone is \emph{possibly} infeasible ($F_{inf} = 1$), in which case, we proceed to Phase 2 of the algorithm.

\subsubsection{\textbf{Phase 2: Obtaining node-disjoint paths from each source to the sink with a small relay count}}
\label{subsubsec:phase2}
We come to this phase if phase 1 fails to find a network on $Q$ alone satisfying the design objectives. Our objective in this phase is to obtain $k$ node disjoint hop constrained paths from each source to the sink, using as few additional relays as possible. To that end, we proceed as explained in Section~\ref{subsec:ideamain}.

We present below, the detailed pseudo code for this phase, along with necessary remarks, and explanations. 

\gap 
\noindent
\fbox{
\begin{minipage}[t]{3in}
\textbf{Input:} $G=(V, E)$, $h_{\max}$, $k$\\
\comment $V=Q\cup R$, and $E$ is the set of all edges of length $\leq r_{\max}$ on $Q\cup R$.\\
\textbf{Output:} $T$ (the desired network)\\
\textbf{flag:} boolean $F_{inf}$ (if $F_{inf}=1$, problem is (possibly) infeasible)\\
\textbf{Initialize:} $T=\varnothing$, $F_{inf}=0$

\gap
\noindent
\step {1}: $(T,F_{inf})\leftarrow SPTiRP(G)$
\end{minipage}}

\gap
\noindent
\remark We run the SPTiRP algorithm on $G$ to obtain a one connected hop constrained network with as few relays as possible.

\gap
\noindent
\fbox{
\begin{minipage}[t]{3in}
\step {2}: \emph{if} $F_{inf}=1$

\hspace{10mm} exit Phase 2

\hspace{7mm}\emph{else} go to the next step
\end{minipage}}

\gap
\noindent
\remark If the SPTiRP algorithm fails to meet the hop constraint for some of the sources, we declare the problem to be infeasible, and stop. Otherwise, we proceed \emph{to find alternate node disjoint hop constrained paths} from each of the sources to the sink, as below.

\gap
\noindent
\fbox{
\begin{minipage}[t]{3in}
\step {3}: $Q\leftarrow sort(Q)$ \emph{in decreasing order of Euclidean distance from the sink}
\end{minipage}}

\gap
\noindent
\remark We arrange the sources in decreasing order of their Euclidean distances from the sink; we shall start the alternate route determination procedure with the farthest source and proceed in that order. The logic behind this approach is as follows:
\begin{itemize}
\item Farther sources are likely to consume more relays.
\item So meet their need first.
\item As we consider sources closer to the sink, coax these sources to share the already used relays.
\end{itemize} 

\gap
\noindent
\fbox{
\begin{minipage}[t]{3in}
\step {4}: $n = 2$ (\comment $n$ is the loop variable for the outer loop to be defined next; $n$ keeps track of the number of node disjoint paths discovered, including current iteration)

\gap
\textbf{Outer loop:} \emph{while} $n\leq k$ (\comment the following steps will be repeated until all the sources have $k$ node disjoint hop constrained paths to the sink)

\gap
\noindent
\step {5}: $L^{(r)}\leftarrow R \cap T$

\hspace{5mm}$\overline{R}\leftarrow R\backslash L^{(r)}$
\end{minipage}}

\gap
\noindent
\remark In the $n^{th}$ iteration of the outer loop, we shall try to obtain the $n^{th}$ node disjoint, hop constrained path from each source to the sink, using as few relays as possible. 

We define a \emph{Locked set} $L^{(r)}$ as the set of relays used so far in the course of the network design algorithm (and hence, are part of the final desired network). For example, at the start of the $n^{(th)}$ iteration ($n=2, \ldots, k$), $L^{(r)}$ consists of the relays used in the first $(n-1)$ node disjoint paths from each of the sources to the sink. 

We also define a \emph{free} relay set, $\overline{R}$, as the set of relays \emph{not used} so far in the network design. 

Therefore, in our attempt to minimize the number of additional relays used, we shall try to reuse relays from the Locked set $L^{(r)}$ whenever possible, and try to minimize the use of relays from the free relay set. With this in mind,we proceed to obtain the alternate node disjoint paths from each source to the sink as below.

\gap
\noindent
\fbox{
\begin{minipage}[t]{3in}
\step {6}: 

\gap
\textbf{Inner loop:} \emph{for} each source $S_i\in Q$, $1\leq i \leq |Q|-1$ (\comment The following steps will be repeated for each source, starting with the source farthest from the sink)

\gap
\noindent
\step {7}: $V^{used}_i\leftarrow V \cap \{\cup_{l=1}^{n-1} path_{\max}(S_i,0)\}$
\end{minipage}}

\gap
\noindent
\remark We designate by $path_{\max}(S_i,0)$, the $l^{th}$ node disjoint hop constrained path from source $S_i$ to sink ($1\leq l\leq k$). In Step 7, we identify the set of nodes (designated by $V^{used}_i$) used by the first $n-1$ node disjoint paths from source $S_i$ to sink. Since we want to find another \emph{node disjoint} path from source $S_i$ to sink in the current ($n^{th}$) iteration (of the outer loop), we need to remove the set of vertices, $V^{used}_i$, except $S_i$ and 0, from consideration for the $n^{th}$ path before proceeding further in the current iteration. We do that in the next step.

Note that at the end of Steps 7 to 11, we shall either have a feasible solution for the $n^{th}$ node disjoint path from source $S_i$ to sink, or we shall end up with \emph{possible} infeasibility. 

\gap
\noindent
\fbox{
\begin{minipage}[t]{3in}
\step {8}: $V_i^{(r)} \leftarrow \{V\backslash V_i^{used}\} \cup \{S_i,0\}$
\end{minipage}}

\gap
\remark $V_i^{(r)}$ is the set of vertices not used in the first $n-1$ paths from source $S_i$ to the sink, and therefore, eligible to be part of the $n^{th}$ node disjoint path from source $S_i$ to sink.

\gap
\noindent
\fbox{
\begin{minipage}[t]{3in}
\step {9}: $G_i^n \leftarrow$ \emph{restriction of} $G$ to $V_i^{(r)}$   
\end{minipage}}

\gap
\remark In order to obtain the $n^{th}$ node disjoint path from source $S_i$ to sink, we restrict the graph $G$ to the \emph{eligible} node set $V_i^{(r)}$.

\gap
\noindent
\fbox{
\begin{minipage}[t]{3in}
\step {10}: $path_n^{sh}(S_i,0) \leftarrow SPT(G_i^n)$

\gap
\noindent
\step {11}: \emph{if} hopcount($path_n^{sh}(S_i,0)$) $> h_{\max}$

\hspace{9mm} $\{F_{inf}\leftarrow 1$;

\hspace{10mm}\emph{exit Phase 2}\}

\gap 
\hspace{9mm}\emph{else} go to the next Step
\end{minipage}}

\gap
\noindent
\remark \begin{enumerate}
\item If $path_n^{sh}(S_i,0)$, the shortest path from source $S_i$ to sink in $G_i^n$, does not meet the hop constraint, we declare the problem to be \emph{possibly} infeasible, and stop, as we have failed to obtain a feasible solution for the $n^{th}$ node disjoint path ($1<n \leq k$) from source $S_i$ to sink. If hop constraint is satisfied by $path_n^{sh}(S_i,0)$, we have a feasible solution for the $n^{th}$ node disjoint path from source $S_i$ to sink, and we proceed to the next step to prune relays from this feasible solution in order to achieve a better solution in terms of relay count. 
\end{enumerate}

\gap
\noindent
\fbox{
\begin{minipage}[t]{3in}
\step {12}: $L_i^{(r)}\leftarrow L^{(r)}\cap V_i^{(r)}$

\hspace{5mm} $Q_i^{(r)}\leftarrow Q\cap V_i^{(r)}$
\end{minipage}}

\gap
\noindent
\remark We designate by $L_i^{(r)}$, the members of the \emph{locked} relay set that are eligible to be part of the $n^{th}$ node disjoint path from source $S_i$ to sink. Similarly, $Q_i^{(r)}$ denotes the set of source nodes that are part of the eligible node set $V_i^{(r)}$. Note that, apart from the sets $L_i^{(r)}$ and $Q_i^{(r)}$, the only other component of the eligible node set $V_i^{(r)}$ is the \emph{free} relay set, $\overline{R}$, i.e., $V_i^{(r)}=L_i^{(r)}\cup Q_i^{(r)}\cup \overline{R}$. 

Now, \emph{in our attempt to minimize the number of relays used, we shall try to prune one at a time, the \emph{free} relays (i.e., relays in $\overline{R}$) that are part of the initial feasible path, $path_n^{sh}(S_i,0)$, and try to obtain an alternate hop constrained path, reusing the relays in the \emph{eligible locked} set $L_i^{(r)}$ as much as possible.} We explain this procedure in the next steps. 

\gap
\noindent
\fbox{
\begin{minipage}[t]{3in}
\step {13}: $\overline{R_i}\leftarrow \overline{R}\cap path_n^{sh}(S_i,0)$
\end{minipage}}

\gap
\noindent
\remark We identify as $\overline{R_i}$, the set of \emph{free} relays that are part of the initial feasible solution, $path_n^{sh}(S_i,0)$, for the $n^{th}$ node disjoint path from source $S_i$ to sink. To reduce the relay count, we shall try to prune the relays in $\overline{R_i}$ one at a time, while maintaining the hop constraint.

\gap
\noindent
\fbox{
\begin{minipage}[t]{3in}
\step {14}: $V_i\prime^{(r)}\leftarrow Q_i^{(r)}\cup L_i^{(r)}\cup \overline{R_i}$

$G_i\prime^{(r)}\leftarrow $ \emph{restriction} of $G_i^n$ to $V_i\prime^{(r)}$ (\comment After pruning a relay from the initial feasible path, we shall search for a hop constrained path over this graph, and not over $G_i^n$).
\end{minipage}}

\gap
\noindent
\remark After pruning a ``free" relay from the feasible path, $path_n^{sh}(S_i,0)$, we shall search for a better path (in terms of relay count) using only the remaining \emph{free} relays in $path_n^{sh}(S_i,0)$, and the \emph{locked} relays and sources in $V_i^{n}$, irrespective of whether they are part of $path_n^{sh}(S_i,0)$. In other words, we shall restrict our search space for the hop constrained path to the graph $G_i\prime^{(r)}$, thereby disallowing the use of \emph{free} relays that are not part of $path_n^{sh}(S_i,0)$, while still allowing the use of (eligible) \emph{locked} relays and sources even if they are not part of the initial feasible path. The idea behind this selection of search space is as follows:

\begin{itemize}
\item Restricting the search space to only the \emph{free} relays in $path_n^{sh}(S_i,0)$ ensures reduction in relay count, if we can find a hop constrained path after relay pruning.
\item Since the relays in $L_i^{(r)}$, and the sources $Q_i^{(r)}$ are already part of the final desired network, their inclusion in the current search space does not contradict our objective of reducing relay count, but helps by improving the chance of finding a hop constrained path after pruning a \emph{free} relay. 
\item This selection of search space, thus, enforces the reuse of relays in \emph{Locked} set, while trying to avoid the use of \emph{free} relays, thereby improving the relay count.  
\end{itemize}

\gap
\noindent
\fbox{
\begin{minipage}[t]{3in}
\step {15}: \emph{for} each node $j\in \overline{R_i}$ (\comment the following sub steps will be repeated until no more relay pruning from the set $\overline{R_i}$ is possible without violating the hop constraint)        

\gap
\noindent
\step {15a}: $TempPath(S_i,0)\leftarrow SPT(G_i\prime^{(r)}\backslash j)$
\end{minipage}}

\gap
\remark With slight abuse of notation, we designate by $G_i\prime^{(r)}\backslash j$, the restriction of the graph $G_i\prime^{(r)}$ to the node set $V_i\prime^{(r)}\backslash j$. After pruning a relay $j$, we obtain the shortest path ($TempPath(S_i,0)$) from source $S_i$ to sink, using the remaining vertices in $V_i\prime^{(r)}$. 

\gap
\noindent
\fbox{
\begin{minipage}[t]{3in}
\step {15b}: \emph{if} hopcount($TempPath(S_i,0)$) $> h_{\max}$

\hspace{9mm} \emph{continue}; (\comment Go back to Step 15, and try pruning the next relay in $\overline{R_i}$)

\gap 
\emph{else} go to next Step
\end{minipage}}

\gap
\noindent
\remark If the shortest path from $S_i$ to sink in $G_i\prime^{(r)}$ after pruning relay $j$ does not satisfy the hop constraint, we replace back the relay $j$, and try pruning the next relay in $\overline{R_i}$. 

\gap
\noindent
\fbox{
\begin{minipage}[t]{3in}
\step {15c}: $path_n(S_i,0)\leftarrow TempPath(S_i,0)$

$V_i\prime^{(r)} \leftarrow V_i\prime^{(r)}\backslash j$

$\overline{R_i} \leftarrow \overline{R_i} \backslash j$

$G_i\prime^{(r)} \leftarrow restriction$ of $G_i^n$ to $V_i\prime^{(r)}$
\end{minipage}}

\gap
\noindent
\remark If relay $j$ can be pruned successfully without violating the hop constraint, we update the $n^{th}$ node disjoint path, $path_n(S_i,0)$, and the sets $V_i\prime^{(r)}$ (the set of vertices to be part of the search space for a feasible path) and $\overline{R_i}$ (the set of \emph{free} relays to be pruned) as above before proceeding to the next iteration of the loop (i.e., before pruning the next relay in $\overline{R_i}$). 

Note that since relay $j$ has been pruned, the updated candidate path will have at least one relay less than the relay set used by the initial feasible path, $path_n^{sh}(S_i,0)$.

\gap
\noindent
\fbox{
\begin{minipage}[t]{3in}
\step {16}: $T\leftarrow T\cup path_n{S_i,0}$

\hspace{5mm} $L^{(r)}\leftarrow L^{(r)}\cup \{R\cap path_n{S_i,0}\}$

\hspace{5mm} $\overline{R}\leftarrow R\backslash L^{(r)}$

\textbf{End of Inner loop}
\end{minipage}}

\gap
\noindent
\remark After obtaining the $n^{th}$ node disjoint, hop constrained path from a source $S_i$ to the sink using as few relays as possible, we augment the network $T$ with the path from source $S_i$ to sink, namely, $path_n(S_i,0)$. Also, before proceeding to the next source (i.e., the next iteration of the inner loop), we update the \emph{Locked relay set} with the new relays used in $path_n(S_i,0)$; the \emph{free relay set} $\overline{R}$ is also updated accordingly.   

\gap
\noindent
\fbox{
\begin{minipage}[t]{3in}
\step {17}: $n\leftarrow n+1$

\textbf{End of Outer loop}
\end{minipage}}

\gap
\noindent
\remark When we have obtained a node disjoint, hop constrained path from all the sources to the sink in the current iteration (of the outer loop), we proceed to the next iteration of the outer loop.

\noindent
\fbox{
\begin{minipage}[t]{3in}
\textbf{End of pseudo code for E-SPTiRP}
\end{minipage}}

\gap
\noindent
\textbf{Remarks:}

Two ideas play a key role in improving the performance of the E-SPTiRP algorithm, namely,

\begin{itemize}
\item prioritizing the reuse of relays already used by the solution thus far (see Section~\ref{subsubsec:phase2}, Step 5 of the pseudo code).
\item ordering the sources in order of their distances from the sink, and searching for alternate path starting with the farthest source (see Section~\ref{subsubsec:phase2}, Step 3 of the pseudo code).
\end{itemize}

To demonstrate that these two ideas do improve performance, we devised another empty algorithm which mimics most of the steps of the E-SPTiRP algorithm, except that, while finding alternate node disjoint routes for the sources, 
\begin{itemize}
\item We did not put any emphasis on the reuse of relays 
\begin{itemize}
\item that are already part of the one connected network that is obtained in Step 1 of Phase 2
\item relays that have been used to construct alternate routes for the previous sources
\end{itemize}

\item We did not impose any ordering on the sources according to their distances from the BS.  
\end{itemize}
The details of the dummy algorithm are presented in Section~\ref{subsec:dummyalgo}. As we shall see in Section~\ref{sec:results_kconnect}, the performance of this dummy algorithm was significantly worse compared to that of E-SPTiRP in terms of relay count.

\subsection{Dummy Algorithm}
\label{subsec:dummyalgo}
\subsubsection{\textbf{Phase 1: Checking for $k$-connectivity on $Q$ alone}}

\begin{itemize}
\item We aim at finding $k$ node disjoint, hop constrained paths from each source to the sink, \emph{using only other source nodes}.
\item The procedure, and the conclusions are the same as in E-SPTiRP.
\item If hop constraint is met for all sources, we are done. No relays required for $k$-connectivity. Else, proceed to the next phase.  
\end{itemize}

\subsubsection{\textbf{Phase 2: Obtaining node-disjoint paths from each source to the sink with a small relay count}}

We come to this phase if Phase 1 fails to find a network on $Q$ alone satisfying the design objectives. Our objective in this phase is to obtain $k$ node disjoint hop constrained paths from each source to the sink, using as few additional relays as possible. 

\gap 
\noindent
\fbox{
\begin{minipage}[t]{3in}
\textbf{Input:} $G=(V, E)$, $h_{\max}$, $k$\\
\comment $V=Q\cup R$, and $E$ is the set of all edges of length $\leq r_{\max}$ on $Q\cup R$.\\
\textbf{Output:} $T$ (the desired network)\\
\textbf{flag:} boolean $F_{inf}$ (if $F_{inf}=1$, problem is (possibly) infeasible)\\
\textbf{Initialize:} $T=\varnothing$, $F_{inf}=0$

\gap
\noindent
\step {1}: $(T,F_{inf})\leftarrow SPTiRP(G)$
\end{minipage}} 

\gap
\noindent
\remark We run the SPTiRP algorithm on $G$ to obtain a one connected hop constrained network with as few relays as possible.

\gap
\noindent
\fbox{
\begin{minipage}[t]{3in}
\step {2}: \emph{if} $F_{inf}=1$

\hspace{10mm} exit Phase 2

\hspace{7mm}\emph{else} go to next step
\end{minipage}}

\gap
\noindent
\remark If the SPTiRP algorithm fails to meet the hop constraint for some of the sources, we declare the problem to be infeasible, and stop.

\gap
\noindent
\fbox{
\begin{minipage}[t]{3in}
\step {3}: $n=2$ (\comment $n$ is the loop variable for the outer loop described next)

\noindent
\textbf{Outer loop: }\emph{while} $n \leq k$ (\comment the following steps (Steps 4 to 12) will be repeated until all the sources have $k$ node disjoint hop constrained paths to the sink)

\gap
\noindent
\textbf{Inner loop: } \emph{for} each source $S_i \in Q$, $1 \leq i \leq |Q|-1$ (\comment Steps 4 to 11 will be repeated for each source)

\gap
\noindent
\step {4}: $V^{used}_i\leftarrow V \cap \{\cup_{l=1}^{n-1} path_{\max}(S_i,0)\}$
\end{minipage}}

\gap
\noindent
\remark We designate by $path_{\max}(S_i,0)$, the $l^{th}$ node disjoint hop constrained path from source $S_i$ to sink ($1\leq l\leq k$). In Step 4, we identify the set of nodes (designated by $V^{used}_i$) used by the first $n-1$ node disjoint paths from source $S_i$ to sink. Since we want to find another \emph{node disjoint} path from source $S_i$ to sink in the current ($n^{th}$) iteration (of the outer loop), we need to remove the set of vertices, $V^{used}_i$, except $S_i$ and 0, from consideration for the $n^{th}$ path before proceeding further in the current iteration. We do that in the next step.

\gap
\noindent
\fbox{
\begin{minipage}[t]{3in}
\step {5}: $V_i^{(r)} \leftarrow \{V\backslash V_i^{used}\} \cup \{S_i,0\}$
\end{minipage}}

\gap
\noindent
\remark $V_i^{(r)}$ is the set of vertices not used in the first $n-1$ paths from source $S_i$ to the sink, and therefore, eligible to be part of the $n^{th}$ node disjoint path from source $S_i$ to sink.

\gap
\noindent
\fbox{
\begin{minipage}[t]{3in}
\step {6}: $G_i^n \leftarrow$ \emph{restriction of} $G$ to $V_i^{(r)}$   
\end{minipage}}

\gap
\noindent
\remark In order to obtain the $n^{th}$ node disjoint path from source $S_i$ to sink, we restrict the graph $G$ to the \emph{eligible} node set $V_i^{(r)}$.

\gap
\noindent
\fbox{
\begin{minipage}[t]{3in}
\step {7}: $path_n^{sh}(S_i,0) \leftarrow SPT(G_i^n)$

\gap
\noindent
\step {8}: \emph{if} hopcount($path_n^{sh}(S_i,0)$) $> h_{\max}$

\hspace{9mm} $\{F_{inf}\leftarrow 1$;

\hspace{10mm}\emph{exit Phase 2}\}

\gap 
\hspace{9mm}\emph{else} go to next Step
\end{minipage}}

\gap
\noindent
\remark \begin{enumerate}
\item We obtain the shortest path (designated by $path_n^{sh}(S_i,0)$)from source $S_i$ to sink in $G_i^n$; if this path meets the hop constraint, we have a feasible solution for the $n^{th}$ node disjoint path from source $S_i$ to sink, and we can proceed to the next step to prune relays from the feasible solution in order to achieve a better solution in terms of relay count. If, however, the path obtained in Step 7 does not meet the hop constraint, we declare the problem to be \emph{possibly} infeasible, and stop, as we have failed to obtain a feasible solution for the $n^{th}$ node disjoint path ($1<n \leq k$) from source $S_i$ to sink.
\item Successful completion of Steps 7 and 8 thus guarantee the existence of a hop count feasible, node disjoint path from source $S_i$ to sink. Now, since our objective is to meet the design requirements using as few relays as possible, in the next step, we shall try to obtain a better path (in terms of relay count) from source $S_i$ to sink, by pruning relays from the feasible path obtained earlier. \emph{The dummy algorithm will differ from E-SPTiRP algorithm in this relay pruning procedure.}
\item Note that we can simply run the SPTiRP algorithm on the graph $G_i^n$ to obtain the $n^{th}$ node disjoint hop constrained path from source $S_i$ to sink, update the network $T$ with that path, and move on to the next source (next iteration of the inner loop). But since the SPTiRP algorithm is designed to `optimize' the \emph{total} number of relays used by all the source nodes in $G_i^n$, the path so obtained from source $S_i$ to sink may not be the best in terms of relay count for source $S_i$. Hence, we shall explore other methods of reducing relay count in the path from source $S_i$ to sink, as indicated in Step 9.
\end{enumerate}

\gap
\noindent
\fbox{
\begin{minipage}[t]{3in}
\step {9}: $path_n^{(1)}(S_i,0)\leftarrow SPTiRP(G_i^n)$ (\comment Candidate path 1)

\gap
$path_n^{(2)}(S_i,0)\leftarrow RoutineA(G_i^n, path_n^{sh}(S_i,0), R)$ (\comment Candidate path 2)
\end{minipage}}

\gap
\noindent
\remark We obtain two candidate routes for the $n^{th}$ node disjoint hop constrained path from source $S_i$ to sink in $G_i^n$. The first candidate path is obtained simply by running the SPTiRP algorithm on the graph $G_i^n$. The second candidate route is obtained by pruning relays from the feasible path, $path_n^{sh}(S_i,0)$, obtained in Steps 7 and 8 earlier. The routine for this relay pruning procedure, namely \emph{RoutineA}, is described next. Once we have the candidate routes, we shall choose the best among them in terms of relay count.

\gap
\noindent
\fbox{
\begin{minipage}[t]{3in}
\textbf{Pseudo code for RoutineA}

\gap
\noindent
\textbf{Input:} $G_i^n$, $path_n^{sh}(S_i,0)$, $R$

\noindent
\textbf{Output:} $path_n^{(2)}(S_i,0)$ (\comment the second candidate path)

\noindent
\textbf{Initialize:} $path_n^{(2)}(S_i,0)\leftarrow path_n^{sh}(S_i,0)$

\gap
\noindent
\step {a}: $Q_i^A \leftarrow V_i^{(r)}\cap Q$

\hspace{1mm} $R_i^A \leftarrow R \cap path_n^{sh}(S_i,0)$

\hspace{1mm} $V_i^A \leftarrow Q_i^A \cup R_i^A$
\end{minipage}}

\gap
\noindent
\fbox{
\begin{minipage}[t]{3in}
\hspace{1mm} $G_i^A \leftarrow restriction$ of $G_i^n$ to $V_i^A$ (\comment This is the graph over which we shall search for a hop constrained path after pruning a relay)
\end{minipage}}

\gap
\noindent
\remark $R_i^A$ is the set of relays used in the feasible path, $path_n^{sh}(S_i,0)$, from source $S_i$ to sink. We shall try to prune the relays in the set $R_i^A$ one by one to achieve a better path in terms of relay count. 

$Q_i^A$ is the set of sources (may or may not be used in the initial feasible path, $path_n^{sh}(S_i,0)$) belonging to the \emph{eligible} vertex set, $V_i^n$, defined earlier in Step 5 of Dummy Algorithm, phase 2. We designate by $V_i^A$, the vertex set consisting of the sources in $Q_i^A$, and the relays in $R_i^A$. 

We shall allow our search space to be the vertex set $V_i^A$ (i.e., the graph $G_i^A$), i.e., we shall allow \emph{all the \emph{eligible} source nodes in $V_i^n$, irrespective of whether they were used or not in the initial feasible path}, and allow only the relays used in the initial feasible path, $path_n^{sh}(S_i,0)$.

\begin{itemize}
\item Since we are allowing only the relays used in the initial feasible path, this method will still ensure a reduction in relay count, if a hop constrained path is found after pruning a relay.
\item In the worst case, the routine may end up with $path_n^{sh}(S_i,0)$ as outcome.
\end{itemize}

\gap
\noindent
\fbox{
\begin{minipage}[t]{3in}
\step {b}: \textbf{Loop:} \emph{for} each node $j \in R_i^A$ (\comment the following steps will be repeated until no more relay pruning is possible without violating the hop constraint)

\gap
\noindent
\step {c}: $TempPath(S_i,0)\leftarrow SPT(G_i^A\backslash j)$
\end{minipage}} 

\gap
\noindent
\remark With slight abuse of notation, we designate by $G_i^A\backslash j$, the restriction of the graph $G_i^A$ to the node set $V_i^A\backslash j$. After pruning a relay $j$, we obtain the shortest path ($TempPath(S_i,0)$) from source $S_i$ to sink, using the remaining vertices in $V_i^A$. 

\gap
\noindent
\fbox{
\begin{minipage}[t]{3in}
\step {d}: \emph{if} hopcount($TempPath(S_i,0)$) $> h_{\max}$

\hspace{9mm} \emph{continue}; (\comment Go back to step 2 and try pruning the next relay in $R_i^A$)

\gap 
\emph{else} go to next Step
\end{minipage}}

\gap
\remark If the shortest path from $S_i$ to sink in $G_i^A$ after pruning relay $j$ does not satisfy the hop constraint, we replace back the relay $j$, and try pruning the next relay in $R_i^A$. 

\gap
\noindent
\fbox{
\begin{minipage}[t]{3in}
\step {e}: $path_n^{(2)}(S_i,0)\leftarrow TempPath(S_i,0)$

$R_i^A \leftarrow R \cap path_n^{(2)}(S_i,0)$

$V_i^A \leftarrow Q_i^A \cap R_i^A$

$G_i^A \leftarrow restriction$ of $G_i^n$ to $V_i^A$
\end{minipage}}

\gap
\noindent
\remark If relay $j$ can be pruned successfully without violating the hop constraint, we update the candidate path, $path_n^{(2)}(S_i,0)$, and the sets $V_i^A$ and $R_i^A$ (the set of relays used in the candidate path) as above before proceeding to the next iteration of the loop (i.e., before pruning the next relay). 

Note that since relay $j$ has been pruned, the updated relay set $R_i^A$ (and hence the updated candidate path) will have at least one relay less than the relay set used by the initial feasible path, $path_n^{sh}(S_i,0)$.
 
\gap
\fbox{
\begin{minipage}[t]{3in}
\emph{\textbf{end of Pseudo code for RoutineA}}
\end{minipage}}

\gap
\noindent
\fbox{
\begin{minipage}[t]{3in}
\step {10} (of Dummy Algorithm, phase 2): 

$path_n(S_i,0) \leftarrow \arg \min \{relaycount(path_n^{(1)}(S_i,0)), \\ relaycount(path_n^{(2)}(S_i,0))\}$
\end{minipage}}

\gap
\noindent
\remark Once we have the candidate routes for the $n^{th}$ node disjoint hop constrained path from source $S_i$ to sink, we choose the best among them in terms of relay count. 

\gap
\noindent
\fbox{
\begin{minipage}[t]{3in}
\step {11}: $T \leftarrow T \cup path_n(S_i,0)$

\textbf{Inner loop end}
\end{minipage}}

\gap
\noindent
\remark Update the network $T$ with the $n^{th}$ node disjoint, hop constrained path from source $S_i$ to sink. 

\gap
\noindent
\fbox{
\begin{minipage}[t]{3in}
\step {12}: $n \leftarrow n+1$

\textbf{Outer loop end}

\gap
\textbf{end of Pseudo code for Dummy Algorithm}
\end{minipage}}

\subsubsection{Time complexity of the Dummy algorithm}

Phase 1 of Dummy algorithm involves repeating the SPTiRP algorithm on the set $Q$ alone at most $k$ times for each of the sources. Hence, the time complexity of this phase is upper bounded by $k(|Q|-1)g_{sptirp}(|Q|)$. 

Phase 2 starts by running the SPTiRP algorithm on the entire graph. The time complexity involved therein is $g_{sptirp}(|Q|+|R|)$. Now the alternate route finding procedure involves finding SPT on $G_i$ (see the algorithm for definition of $G_i$), followed by two different methods of finding candidate routes. 

Now, the time complexity of finding SPT on $G_i$ is $\leq g_{spt}(|Q|+|R|)$, and this step is repeated at most $k-1$ times for each source. 

RoutineA (for computing the second candidate path) involves pruning relays from a path, one at a time, and running SPT on the remaining searchspace (i.e., the relays on the path, and all the eligible source nodes. See Dummy algorithm for detailed explanation) to check hop constraint feasibility. The time complexity of RoutineA is, therefore, upper bounded by $|R|g_{spt}(|Q|+|R|-1) \leq |R|g_{spt}(|Q|+|R|)$. This method is repeated at most $k-1$ times for each source.  

Computation of the first candidate path involves running the SPTiRP algorithm on $G_i$, and its worst case complexity is upper bounded by $g_{sptirp}(|Q|+|R|)$. This method is also repeated at most $k-1$ times for each source.

Hence, the time complexity of Dummy algorithm is upper bounded by
$k(|Q|-1)g_{sptirp}(|Q|)+g_{sptirp}(|Q|+|R|)+(|Q|-1)(k-1)\{g_{spt}(|Q|+|R|)+|R|g_{spt}(|Q|+|R|)+g_{sptirp}(|Q|+|R|)\}$.

Recall that $g_{sptirp}(\cdot)$ and $g_{spt}(\cdot)$ are both polynomial time, and hence the above expression is polynomial time in $|Q|$, $|R|$, and $k$.

\subsection{Analysis of E-SPTiRP}

\subsubsection{Time complexity}
We show below that the time complexity of the algorithm is upper bounded by polynomials in $|Q|$, $|R|$, and $k$. Hence the algorithm is polynomial time.

\begin{lemma}
\label{lemma3}
The time complexity of the E-SPTiRP algorithm is upper bounded by $k(|Q|-1)g_{sptirp}(|Q|)+g_{sptirp}(|Q|+|R|)+(k-1)(|Q|-1)(|R|+1)g_{spt}(|Q|+|R|)$, where $g_{sptirp}(\cdot)$ is the time complexity of the SPTiRP algorithm, and $g_{spt}(\cdot)$ is the time complexity of finding the shortest path tree.
\end{lemma}

\begin{proof}

The first and second terms in the above expression can be derived by similar arguments as given for Dummy algorithm above. 

Now, the alternate route finding procedure starts by finding an SPT on a graph $G_i$, a restriction of the graph $G$. The complexity of this is upper bounded by $g_{spt}(|Q|+|R|)$. This is repeated for each source at most $k-1$ times.

The next step in alternate route finding consists of pruning from a path ($path_2(i,0)$), relays chosen from a certain selected set ($R_i$), one at a time, and finding the SPT on the resulting restricted graph (see E-SPTiRP for detailed explanation) to check if hop constraint is satisfied by the resulting path. The worst case complexity of this step is upper bounded by $|R|g_{spt}(|Q|+|R|)$. This step is also repeated at most $k-1$ times for each source.

Hence, the worst case complexity of E-SPTiRP is upper bounded by $k(|Q|-1)g_{sptirp}(|Q|)+g_{sptirp}(|Q|+|R|)+(k-1)(|Q|-1)(|R|+1)g_{spt}(|Q|+|R|)$.
\end{proof}

From Lemma~\ref{lemma3}, it follows that E-SPTiRP algorithm is polynomial time.

\subsubsection{Worst case approximation guarantee}
As already stated in Corollary~\ref{thm2} (Section~\ref{sec:formulation_kconnect}), no polynomial time algorithm can provide finite approximation guarantee for the general class of RSN$k$-MR-HC problems. However, for a subclass of RSN$k$-MR-HC problems where the optimal solution for one connectivity with hop constraint (RST-MR-HC) uses at least one relay, we can derive a polynomial factor worst case approximation guarantee for polynomial time complexity algorithms. This is the content of Lemma~\ref{lemma2}, Corollary~\ref{corollary:lem2}, and Theorem~\ref{thm:approxfactor}.

\begin{lemma}
\label{lemma2}
For any fixed $k\in \{1,2,\ldots\}$, if the optimal solution for the RSN$k$-MR-HC problem on the graph $G=(V,E)$, uses $n>0$ relays, then the optimal solution for the RSN$(k+1)$-MR-HC problem on the same graph with the same hop constraint as the RSN$k$-MR-HC problem, uses at least $n+1$ relays. 
\end{lemma}

\begin{proof}
Consider a problem instance where the optimal solution for the RSN$k$-MR-HC problem uses $n>0$ relays.  

Suppose we claim that the optimal solution for the RSN$(k+1)$-MR-HC problem on that problem instance also uses $n$ relays (it evidently cannot use fewer than $n$ relays as that would contradict the hypothesis that the RSN$k$-MR-HC problem uses at least $n$ relays).

Therefore, for that problem instance, there exists a relay set $R_{opt}=\{R_1,R_2,\ldots ,R_n\}$ such that each source has $k+1$ node disjoint paths to the sink involving some of those relays. 

Note that $R_{opt}$ is an optimal solution for the RSN$k$-MR-HC problem on this problem instance. This is because, from each source to the sink we can take any $k$ of the $k+1$ paths provided by this solution to RSN$(k+1)$-MR-HC; this will be an optimal solution to RSN$k$-MR-HC. In this optimal solution to RSN$k$-MR-HC, let $S_i$ be the set of sources that use the relay $R_i$, $1\leq i\leq n$. Note that if this set is empty for some $i$, $1\leq i\leq n$, then there is nothing left to prove, since we could obtain an optimal solution to the RSN$k$-MR-HC problem using only the relays in $R_{opt}\backslash R_i$ (which, in turn, contradicts the assumption that the optimal solution to the RSN$k$-MR-HC problem uses $n$ relays). We, therefore, assume that $S_i$ is nonempty for all $i$, $1\leq i\leq n$, and derive a contradiction. 

By our claim, $R_{opt}$ is also the optimal solution for the RSN$(k+1)$-MR-HC problem. Since the $(k+1)$ paths from each source to the sink in this solution must be \emph{node disjoint}, therefore, each source in the set $S_i$ uses the relay $R_i$ in only one of its $(k+1)$ paths to the sink. It follows, therefore, that each source in $S_i$ has $k$ hop constrained, node disjoint paths to the sink \emph{using only the relays in $R_{opt}\backslash R_i$.} This, in turn, implies that the set $R_{opt}\backslash R_i$ is sufficient to obtain $k$ node-disjoint, hop constrained paths from each of the sources to the sink, i.e., $R_{opt}\backslash R_i$ is, in fact, an optimal solution for the RSN$k$-MR-HC problem. This contradicts our earlier assumption that the optimal solution for RSN$k$-MR-HC problem uses $n$ relays. 

Therefore, our claim that the optimal solution for the RSN$(k+1)$-MR-HC problem uses $n$ relays is wrong. Hence, the optimal solution for RSN$(k+1)$-MR-HC problem must use at least $n+1$ relays.
\end{proof}

\begin{corollary}
\label{corollary:lem2}
If the optimal solution for the RST-MR-HC problem on the graph $G=(V,E)$, uses $m>0$ relays, then the optimal solution for the RSN$k$-MR-HC problem on the same graph with the same hop constraint as the RST-MR-HC problem, uses at least $m+k-1$ relays. 
\end{corollary}

\begin{proof}
If the optimal solution for the RST-MR-HC problem on the graph $G=(V,E)$, uses $m>0$ relays, then by Lemma~\ref{lemma2}, the optimal solution for the RSN$2$-MR-HC problem must use at least $m+1$ relays. Therefore, using Lemma~\ref{lemma2} once more, the optimal solution for the RSN$3$-MR-HC problem must use at least $m+2$ relays. Thus, by repeated use of Lemma~\ref{lemma2}, the optimal solution for the RSN$k$-MR-HC problem must use at least $m+k-1$ relays. 
\end{proof}

\begin{theorem}
\label{thm:approxfactor}
For the set of problem instances where the optimal solution for the RST-MR-HC problem uses at least one relay, the worst case approximation guarantee given by any polynomial time complexity algorithm for the RSN$k$-MR-HC problem, whenever the algorithm terminates with a feasible solution, is $\min\{m(h_{\max}-1),|R|/k\}$, where $m$ is the number of sources, and $h_{\max}$ is the hop count bound.   
\end{theorem}

\begin{proof}
Since the optimal solution for the RST-MR-HC problem uses at least one relay, the worst case scenario is that the optimal relay count for RST-MR-HC problem is just 1, and hence, from Corollary~\ref{corollary:lem2}, the optimal solution for RSN$k$-MR-HC problem uses at least $k$ relays, whereas a polynomial time algorithm for the same problem may end up using all the $|R|$ relays, or at most $mk(h_{\max}-1)$ relays, whichever is the smaller, whenever the algorithm obtained a feasible solution for the problem. Hence, the worst case approximation guarantee is $\min\{m(h_{\max}-1),|R|/k\}$.
\end{proof}

\subsubsection{Average Case Approximation Guarantee}
\label{subsubsec:avg-case-analysis-kconnect}
We provide a bound on the average case approximation ratio for the RSNk-MR-HC  problem for a stochastic setting very similar to that in Section~\ref{subsubsec:avg-approx-oneconnect}. We consider a square area $A(\subset \Re^2_+)$ of side $a$. The BS is located at (0,0). We deploy $n$ potential locations independently and identically distributed (i.i.d) uniformly randomly over the area $A$; then deploy $m$ sources i.i.d uniformly randomly over the area $A_{\epsilon}$ (recall that for a given $\epsilon\in(0,1)$, $A_{\epsilon}(\subset A)$ denotes the quarter circle of radius $(1-\epsilon)h_{\max}r$ centred at the BS, where $h_{\max}$ is the hop constraint, and $r$ is the maximum allowed communication range). The probability space of this random experiment is $(\Omega^{(n)}_{m,\epsilon},\mathcal{B}^{(n)}_{m,\epsilon},P^{(n)}_{m,\epsilon})$
where, $\Omega^{(n)}_{m,\epsilon}$, $\mathcal{B}^{(n)}_{m,\epsilon}$, $P^{(n)}_{m,\epsilon}$ are as defined earlier.

We consider the random geometric graph $\mathcal{G}^r(\omega)$ induced by considering all links of length $\leq r$ on an instance $\omega \in \Omega^{(n)}_{m,\epsilon}$. We introduce the following notations, in addition to the notations introduced earlier in Section~\ref{subsubsec:avg-approx-oneconnect}:

\begin{description}
\item $\X_k\:=\:\{\omega:\:\exists\:\text{at least $k$ node disjoint paths with hop count $\leq h_{\max}$ from each source to the BS in $\G(\omega)$}\}$: set of all feasible instances for the RSN$k$-MR-HC problem. Note that $\X_k\subset \X$.
\item $\X_{algo}\subset \X_k$: set of all feasible instances where E-SPTiRP algorithm obtains a feasible solution
\item $N_{E-SPTiRP}(\omega)$: Number of relays in the outcome of the E-SPTiRP algorithm on $\mathcal{G}^r(\omega)$ ($\infty$ if $\omega\in\X_{algo}^c$)
\item $R_{Opt,k}(\omega)$: Number of relays in the optimal solution to the RSN$k$-MR-HC problem on $\mathcal{G}^r(\omega)$ ($\infty$ if $\omega\in\X_{k}^c$)
\end{description}

Recall from Corollary~\ref{thm1} that no polynomial time algorithm (and in particular, the E-SPTiRP algorithm) is guaranteed to obtain a feasible solution to the RSN$k$-MR-HC problem whenever such a solution exists. Also recall from Corollary~\ref{thm2} in Section~\ref{sec:formulation_kconnect} that when an optimal solution to the RST-MR-HC problem on an instance uses zero relays, we cannot obtain any finite approximation guarantee for the RSN$k$-MR-HC problem on that instance. Therefore, we consider the set of feasible instances of the RSN$k$-MR-HC problem where the E-SPTiRP algorithm returns a feasible solution, and the optimal solution to the corresponding RST-MR-HC problem uses non-zero number of relays, i.e., $R_{Opt} > 0$.

The \emph{average case approximation ratio} of the E-SPTiRP algorithm over all such feasible instances is defined as
\begin{equation}
\text{Average case approximation ratio, }\alpha_k \define \frac{E[N_{E-SPTiRP}|\X_{algo},R_{Opt}>0]}{E[R_{Opt,k}|\X_{algo},R_{Opt}>0]}
\end{equation}

In the derivation to follow, we will need $\mathcal{X}_k$ to be a high probability event, i.e., with probability greater than $1 - \delta$ for a given $\delta > 0$. The following result, similar to Theorem~\ref{thm:average-case-analysis-setting}, ensures that this holds for the construction provided earlier, provided the number of potential locations is large enough.

\begin{theorem}
\label{thm:average-case-analysis-setting-kconnect}
For any given $\epsilon,\delta \in (0,1)$, $k>1$, $h_{\max}>0$ and $r>0$, there exists $n_0(\epsilon,\delta,k,h_{\max},r)\in \mathbb{N}$ such that, for any $n\geq n_0$, $P^{(n)}_{m,\epsilon}(\X_k)\:\geq\:1-\delta$ in the random experiment $(\Omega^{(n)}_{m,\epsilon},\mathcal{B}^{(n)}_{m,\epsilon},P^{(n)}_{m,\epsilon})$.  
\end{theorem} 

\begin{proof}
The proof is very similar to the proof of Theorem~\ref{thm:average-case-analysis-setting}. 

We make the same construction as shown in Figure~\ref{fig:cutarc}, and define the following events.
\begin{description}
\item $C_{i,j}^l$ = \{$\omega$: $\exists$ at least $k$ nodes out of the $n$ potential loacations in the $i^{th}$ strip of $B_j^l$\}
\item $\X_{\epsilon,\delta,k}\:=\:\{\omega:\:\omega\:\in\cap_{j=1}^{J(r)} \cap_{i=1}^{h_{\max}-1} C_{i,j}^l\}$: Event that there exists at least $k$ nodes out of the $n$ potential locations in each of the first $h_{\max}-1$ strips (see Figure~\ref{fig:h hop}) for all the blades $B_j^l$
\end{description}

Note that for an instance $\omega\in\X_{\epsilon,\delta,k}$, \emph{all nodes (and in particular, all sources) at a distance $<(p-q)hr$ from $b_l$, $1\leq h\leq h_{\max}$, are reachable via at least $k$ node disjoint paths, each with at most $h$ hops}. Since $1>p>q>0$, we can choose $p-q$ to be equal to $1 - \epsilon$, for the given $\epsilon > 0$. It follows that
\begin{equation}
\label{eq:intersection-k}
\X_{\epsilon,\delta,k}\subseteq \X_k
\end{equation}

and hence, $P^{(n)}_{m,\epsilon}(\X_k)\geq P^{(n)}_{m,\epsilon}(\X_{\epsilon,\delta,k})$.

Thus, to ensure $P^{(n)}_{m,\epsilon}(\X_k)\geq 1-\delta$, it is sufficient to ensure that $P^{(n)}_{m,\epsilon}(\X_{\epsilon,\delta,k})\geq 1-\delta$, which we aim to do next. 

As done earlier, we upper bound $J(r)$ as $J(r) \leq \left \lceil \frac{\pi h_{\max}}{2\sqrt{1-p^2}}\right \rceil$.

To simplify notations, we write $P(\cdot)$ to indicate $P^{(n)}_{m,\epsilon}(\cdot)$.

Now we compute,
\begin{eqnarray}
  \label{eq:stripnode-k}
  \lefteqn{P(\X_{\epsilon,\delta,k})} &=& 1-P\left(\cup_{j=1}^{J(r)}\cup_{i=1}^{h_{\max}-1}
{C_{i,j}^{l}}^c\right)\nonumber \\
&\geq& 1-\sum_{j=1}^{J(r)}\sum_{i=1}^{h_{\max}-1}
P\left({C_{i,j}^{l}}^c\right) \nonumber \\
  &\geq& 1 - \left \lceil \frac{\pi h_{\max}}{2\sqrt{1-p^2}}\right \rceil (h_{\max}-1)
\sum_{i=0}^{k-1}{n\choose i} \left(1-\frac{u(r)t(r)}{A}\right)^{n-i}\left(\frac{u(r)t(r)}{A}\right)^i \nonumber \\
  &\geq& 1- \left \lceil \frac{\pi h_{\max}}{2\sqrt{1-p^2}}\right \rceil (h_{\max}-1) \sum_{i=0}^{k-1}{n\choose i} 
e^{-\frac{(n-i)u(r)t(r)}{A}}\left(\frac{u(r)t(r)}{A}\right)^i \nonumber \\
  &=& 1- \left \lceil \frac{\pi h_{\max}}{2\sqrt{1-p^2}}\right \rceil (h_{\max}-1) \sum_{i=0}^{k-1}{n\choose i}
e^{-(n-i)\frac{q\sqrt{1-p^2}r^2}{A}}\left(\frac{q\sqrt{1-p^2}r^2}{A}\right)^i\nonumber\\
&\rightarrow 1 \text{ as } n\rightarrow \infty 
\end{eqnarray}

The first inequality comes from the union bound, the second
inequality, from the upper bound on $J(r)$. The third inequality uses
the result $1-x \leq e^{-x}$. 

Hence the theorem follows.

Note that $p$ and $q$ can be obtained as earlier in terms of $\epsilon$ by maximizing $q\sqrt{1-p^2}$ under the constraint $p-q\:=\:1-\epsilon$.
\end{proof}

\gap
\noindent
\remark For fixed $h_{\max}$ and $r$, $n_0(\epsilon,\delta,k)$ increases with decreasing $\epsilon$ and $\delta$, and increasing $k$.

\gap
\noindent
\textbf{The experiment:} In the light of Theorem~\ref{thm:average-case-analysis-setting-kconnect}, we employ the following node deployment strategy to ensure, w.h.p, feasibility of the RSN$k$-MR-HC problem in the area $A$. 
Choose arbitrary small values of $\epsilon,\delta \in (0,1)$. Given the hop count bound $h_{\max}$ and the maximum communication range $r$, obtain $n_0(\epsilon,\delta,k,h_{\max},r)$ as defined in Theorem~\ref{thm:average-case-analysis-setting-kconnect}. Deploy $n\geq n_0$ potential locations i.i.d uniformly randomly over the area of interest, $A$. $m$ sources are deployed i.i.d uniformly randomly within a radius $(1-\epsilon)h_{\max}r$ from the BS, i.e., over the area $A_{\epsilon}$. By virtue of Theorem~\ref{thm:average-case-analysis-setting-kconnect}, this ensures that any source deployed within a distance $(1-\epsilon)h_{\max}r$ has at least $k$ node disjoint paths to the BS, that are no more than $h_{\max}$ hops w.h.p, thus ensuring feasibility of the RSN$k$-MR-HC problem w.h.p. We run the E-SPTiRP algorithm on the induced random geometric graph with hop count as cost to check if the algorithm returns a feasible solution. We also run the SPTiRP algorithm on the same graph to check if the optimal solution to the RST-MR-HC problem on that instance uses non-zero number of relays. \emph{In this stochastic setting, we derive an upper bound on the average case approximation ratio, $\alpha_k$, of the E-SPTiRP algorithm} as follows.

\begin{lemma}
\label{lemma:kconnect-nodecount-feasible}
$E[N_{E_{SPTiRP}}|\X_{algo},R_{Opt}>0]\leq mk(h_{\max}-1)$
\end{lemma}

\begin{proof}
Observe that for any feasible solution given by the E-SPTiRP algorithm, the number of relay nodes on each of the $k$ paths for a source cannot exceed $h_{\max}-1$ (since each path has at most $h_{\max}$ hops). The lemma follows immediately.
\end{proof}

\begin{lemma}
\label{lemma:kconnect-optimal-feasible}
$E[R_{Opt,k}|\X_{algo},R_{Opt}>0]\geq \underline{R}_{Opt}+k-1$

where, $\underline{R}_{Opt}$ denotes the R.H.S of Equation~\ref{eqn:ropt-lowerbound-final}.
\end{lemma}
\begin{proof}
Given $\X_{algo}$ and $R_{Opt}>0$, it follows from Corollary~\ref{corollary:lem2} that 

\begin{equation*}
R_{Opt,k}(\omega)\geq R_{Opt}(\omega)+k-1\:\forall \omega\in\X_{algo}\cap \{\omega:R_{Opt}(\omega)>0\}
\end{equation*}

Hence,

\begin{equation}
\label{eqn:ropt-k-lower-bound-init}
E[R_{Opt,k}|\X_{algo},R_{Opt}>0]\geq E[R_{Opt}|\X_{algo},R_{Opt}>0]+k-1
\end{equation}

Before proceeding further, we take a small detour. Recall that the E-SPTiRP algorithm declares \emph{possible infeasibility} if after finding a one-connected solution using the SPTiRP algorithm, the alternate path from some source to the BS (a constrained shortest path using only the nodes not used in the previous paths from that source) violates the hop constraint. 
  
Note that the event $\Xepk$ ensures that there exist at least $k$ node disjoint shortest paths (\emph{all of the same minimum hop count}) from each source to the BS, which is a \emph{sufficient condition to ensure that the E-SPTiRP algorithm can find a feasible solution} (since the constrained shortest path avoiding nodes used in upto $(k-1)$ previous paths will still be a shortest path, and hence will meet the hop constraint). Thus, $\Xepk\subset \X_{algo}$.

With this in mind, let us get back to the main proof.

\begin{align}
E[R_{Opt}|\X_{algo},R_{Opt}>0]&\geq E[R_{Opt}|\X_{algo}]\nonumber\\
&\geq E[R_{Opt}\ind_{\Xepk}|\X_{algo}]\nonumber\\
&= P[\Xepk|\X_{algo}]\:E[R_{Opt}|\Xepk,\X_{algo}]\nonumber\\
&\geq P[\Xepk,\X_{algo}]\:E[R_{Opt}|\Xepk,\X_{algo}]\nonumber\\
&= P[\Xepk]\:E[R_{Opt}|\Xepk]\quad\text{since, }\Xepk\subset \X_{algo}\nonumber\\
&\geq (1-\delta)E[R_{Opt}|\Xepk]\label{eqn:ropt-k-lower-bound-inter}
\end{align}

Along similar lines of derivation as in Equations~\eqref{eqn:ropt-lowerbound}-\eqref{eqn:ropt_final} (with $\Xepk$ instead of $\Xep$), we obtain the following bound:

\begin{equation}
\label{eqn:ropt-eps-lower-bound}
E[R_{Opt}|\Xepk]\geq \left[1-\left(\frac{h_{\max}-1}{(1-\epsilon)h_{\max}}\right)^{2m}\right]\:\sum_{i=1}^{h_{\max}-1}q_i
\end{equation}

where, $q_i$ is lower bounded as in Equation~\eqref{eqn:q_expression}.

Finally, combining Equations~\eqref{eqn:ropt-k-lower-bound-init}, \eqref{eqn:ropt-k-lower-bound-inter}, and \eqref{eqn:ropt-eps-lower-bound}, we have the desired lemma.
\end{proof}

Combining Lemma~\ref{lemma:kconnect-nodecount-feasible} and Lemma~\ref{lemma:kconnect-optimal-feasible}, we obtain the following upper bound on the average case approximation ratio of E-SPTiRP algorithm.

\begin{theorem}
The average case approximation ratio, $\alpha_k$, of E-SPTiRP algorithm is upper bounded as

\begin{equation}
\alpha_k\leq \frac{mk(h_{\max}-1)}{\underline{R}_{Opt}+k-1}
\end{equation}

where, $\underline{R}_{Opt}$ is given by the R.H.S of \eqref{eqn:ropt-lowerbound-final}. 
\end{theorem}

\remark For a 2-connectivity problem with 10 sources, and a hop constraint $h_{\max}=4$, the upper bound on the average case approximation ratio of the E-SPTiRP algorithm turns out to be 22, provided the source locations, and the potential relay locations are distributed according to the experiment described earlier.

\section{Numerical Results for the $k$-Connectivity Algorithms}
\label{sec:results_kconnect}
To evaluate the performance of the E-SPTiRP algorithm, we ran both the algorithms (E-SPTiRP as well as the Dummy algorithm presented in the Appendix) to solve the RSN$k$-MR-HC problem with $k=2$ on the same random network scenarios (test set 3) that were generated to test the SPTiRP algorithm (see Section~\ref{subsec:sptirp-expt3}). Due to the small number of relays, the probabilistic analysis of feasibility was not useful; however, in none of the 1000 scenarios tested, the hop constraint turned out to be infeasible. The results are summarized in Table~\ref{tbl:relay_kconnect}.

\begin{table*}[ht]
  \centering
\caption{Performance comparison of the $k$ connectivity algorithms for $k=2$}
\label{tbl:relay_kconnect}
\footnotesize
  \begin{tabular}{|c|c|c|c|c|c|c|c|c|c|c|c|c|}\hline
    Potential & Scenarios & \multicolumn{6}{c|}{Relay Count} & E-SPTiRP  & E-SPTiRP  & E-SPTiRP  & \multicolumn{2}{c|}{Mean execution}\\
     relay & & \multicolumn{3}{c|}{Dummy} & \multicolumn{3}{c|}{E-SPTiRP} & better than & same as & worse than & \multicolumn{2}{c|}{time}\\
count & & &&&&&& Dummy & Dummy & Dummy & \multicolumn{2}{c|}{in sec}\\
 & & Average & Max & Min & Average & Max & Min &  &  &  & Dummy & E-SPTiRP\\
 \hline
   100 & 200 & 5.295 & 13 & 0 & 4.13 & 9 & 0 & 134 & 54 & 12 & 11.163 &	3.0429\\
   110 & 200 & 4.88 & 10 & 0 & 3.895 & 9 & 0 & 121 & 68 & 11 & 13.973 &	3.8558\\
   120 & 200 & 5.45 & 12 & 1 & 4.18 & 8 & 1 & 129 & 52 & 19 & 16.252  &	4.3314\\
   130 & 200 & 5.15 & 11 & 0 & 4 & 8 & 0 & 135 & 54 & 11 & 18.97 & 5.2316\\
   140 & 200 & 5.27 & 14 & 0 & 3.945 & 9 & 0 & 133 & 53 & 14 & 23.748 & 6.3596\\
   \hline
   Total & 1000 & 5.209 & 14 & 0 & 4.03 & 9 & 0 & 652 & 281 & 67 & 16.821 & 4.564\\
  \hline
\end{tabular}
\normalsize
\end{table*}

From Table~\ref{tbl:relay_kconnect}, we can make the following observations:
\begin{enumerate}
\item In all 5 sets of experiments (with different node densities), \emph{the average relay count required by E-SPTiRP to achieve 2 connectivity is less than that required by Dummy algorithm}.
\item \emph{In over 65\% of the tested scenarios, E-SPTiRP performed better than Dummy algorithm in terms of relay count.} In another 28.1\% of cases, they performed equally well.
\item In all 5 sets of experiments, the maximum relaycount required by Dummy algorithm is more than that required by E-SPTiRP (although the maximums for the two algorithms may have been on different random scenarios).
\item In terms of mean execution time, E-SPTiRP performed much better than Dummy algorithm in all 5 sets of experiments. This is probably because of the fact that in the alternate path finding procedure (Steps 5-8 of Phase 2, Dummy algorithm, Section~\ref{subsec:dummyalgo}), Dummy algorithm finds two candidate routes, and chooses the best among them, whereas, E-SPTiRP (Section~\ref{subsubsec:phase2}, Steps 3-16) determines the alternate node disjoint path in one attempt. Thus, the alternate route determination procedure for Dummy algorithm is possibly more time consuming than that of E-SPTiRP.

\item For both the algorithms, \emph{the average execution time increases with increasing node density}.

\end{enumerate}

For each of the five sets of experiments, we also noted the minimum (maximum) relay count required by either algorithm over scenarios where the other algorithm uses a maximum (minimum) number of relays. The comparative study is summarized in Table~\ref{tbl:compare_kconnect}.

 \begin{table*}[ht]
  \centering
\caption{Comparison of Maximum and Minimum Relaycount of the $k$ connectivity algorithms for $k=2$}
\label{tbl:compare_kconnect}
\footnotesize
  \begin{tabular}{|c|c||c|c||c|c||c|c||c|c||}\hline
    Potential & Scenarios & Max & Min & Max & Min & Min & Max & Min & Max \\
     relay & & relay count & relay count & relay count & relay count & relay count & relay count & relay count & relay count\\
count & & of  & of  & of  & of  & of  & of  & of  & of \\
 & & Dummy ($n_1$) & E-SPTiRP & E-SPTiRP ($n_2$) & Dummy & Dummy ($m_1$) & E-SPTiRP & E-SPTiRP ($m_2$) & Dummy\\
 & & & when & & when & & when & & when\\
 & & & Dummy & & E-SPTiRP & & Dummy & & E-SPTiRP\\
 & & & uses $n_1$ & & uses $n_2$ & & uses $m_1$ & & uses $m_2$\\
 \hline
   100 & 200 & 13 & 8 &	9 & 13 &	0 & 0 & 0 &	0\\
   110 & 200 & 10 & 4 &	9 & 10 &	0 & 0 & 0 &	0\\
   120 & 200 & 12	& 6 &	8 &  9 &	1 & 2 & 1 &	6\\
   130 & 200 & 11 & 7 &	8 &  9 &	0 & 0 & 0 &	0\\
   140 & 200 & 14	& 7 &	9 & 12 &	0 & 0 & 0 &	0\\
   \hline
\end{tabular}
\normalsize
\end{table*}

From Table~\ref{tbl:compare_kconnect}, we observe that
\begin{enumerate}
\item For all 5 sets of experiments, in scenarios where Dummy algorithm performs at its worst in terms of relay count, the minimum relay count of E-SPTiRP is always much better than the relaycount of Dummy algorithm. Also observe that the maximum relay count used by E-SPTiRP in all sets of experiments is better than that of Dummy algorithm.
\item In scenarios where E-SPTiRP uses a maximum number of relays, the minimum relay count used by Dummy algorithm is still higher than the relay count of E-SPTiRP. 
\item In scenarios where Dummy algorithm uses zero relays, E-SPTiRP also uses zero relays (which is expected, since Phase 1 is same for both algorithms (see Dummy algorithm and E-SPTiRP in Section~\ref{sec:algorithms})).  
\item In scenarios where Dummy algorithm uses the minimum non-zero number of relays (1 relay), the maximum relay count used by E-SPTiRP was just 1 more than the relay count used by Dummy algorithm. 
\item In scenarios where E-SPTiRP uses the minimum non zero number of relays (1 relay), the maximum relay count used by Dummy algorithm was as high as 6.

\end{enumerate}

Thus, from our observations in Table~\ref{tbl:compare_kconnect}, we can conclude that the worst case performance of E-SPTiRP in terms of relay count is better than that of Dummy algorithm.

To compare the performance of the proposed algorithm against the worst case performance bound given in Theorem~\ref{thm:approxfactor}, we did the following:

\begin{itemize}
\item For each of the five sets of experiments, we identified the scenarios where the optimal solution for the RST-MR-HC problem is non zero.
\item For each of the scenarios thus identified, we can compute the lower bound on the optimal number of relays required for 2-connectivity, using Lemma~\ref{lemma2}, as follows. If the optimal solution for the RST-MR-HC problem uses $n$ relays, the optimal number of relays required for 2-connectivity is lower bounded by $n+1$.

\item For each scenario, we obtained the approximation factor given by the proposed algorithm w.r.t the lower bound computed above as \emph{approximation factor} $= \frac{Relay_{Algo}}{Relay_{lowerbound}}$.
\item For each of the five sets of experiments, we obtained the worst and the best approximation factors (as computed above) achieved by both E-SPTiRP and the Dummy algorithm, and also the worst case performance bound obtained from Theorem~\ref{thm:approxfactor}.
\end{itemize}

The results are summarized in Table~\ref{tbl:worstcase_kconnect}.

\begin{table*}[ht]
  \centering
\caption{Performance Comparison of the $k$ connectivity algorithms against theoretical performance bound for $k=2$}
\label{tbl:worstcase_kconnect}
\footnotesize
  \begin{tabular}{|c|c|c|c|c|c|c|}\hline
    Potential & Scenarios (out of 200 in Table~\ref{tbl:relay_kconnect}) & Worst Case Theoretical & Worst Approx. & Best Approx. & Worst Approx. & Best Approx.\\
     relay & where RST-MR-HC has & performance & factor of & factor of & factor of & factor of\\
count & non zero optimal & bound & Dummy & Dummy & E-SPTiRP & E-SPTiRP\\
  ($|R|$) & solution & ($\min\{10(h_{\max}-1),|R|/2\}$) & & & & \\
 \hline
   100 & 156 & 50 & 5.5	& 1	& 3.5	& 1\\
   110 & 146 & 50 & 5	& 1	& 3.5 & 1\\
   120 & 158 & 50 & 5.5 & 1   & 3   & 1\\
   130 & 150 & 50 & 4.5 & 1   & 3.5 & 1\\
   140 & 147 & 50 & 5.5	& 1.5 & 3.5 & 1\\
   \hline
 Total & 757 & NA & 5.5 & 1   & 3.5 & 1\\
  \hline 
\end{tabular}
\normalsize
\end{table*}

From Table~\ref{tbl:worstcase_kconnect}, we observe that 
\begin{enumerate}
\item For each of the five sets of experiments, \emph{the worst case approximation factor (as defined earlier, for scenarios where optimal solution of RST-MR-HC problem is non zero) achieved by both the algorithms is much better than the theoretical performance bound predicted in Theorem~\ref{thm:approxfactor}}. 

Also note that these approximation factors were computed based on a lower bound on the optimal solution for 2 connectivity; hence the actual performance of the algorithms is even better than this.
\item In all five sets of experiments, E-SPTiRP outperformed Dummy algorithm significantly in terms of the worst case approximation factor.
\item The best approximation factor achieved by both the algorithms was 1, i.e., the lower bound was actually achieved by the algorithms in some of the test cases.

\end{enumerate}

In the relatively small number of test scenarios where Theorem~\ref{thm:approxfactor} does not apply (i.e., optimal solution for RST-MR-HC problem is zero), and hence there is no bounded factor approximation guarantee for the algorithms, we obtained the maximum and minimum number of relays used by the two algorithms over those scenarios. The results are presented in Table~\ref{tbl:zeropt_kconnect}.

\begin{table*}[ht]
  \centering
\caption{Performance of the $k$ connectivity algorithms for $k=2$ in scenarios where there is no bounded approximation guarantee}
\label{tbl:zeropt_kconnect}
\footnotesize
  \begin{tabular}{|c|c|c|c|c|c|c|}\hline
    Potential & Scenarios (out of 200 in Table~\ref{tbl:relay_kconnect}) & Theoretical & Max relaycount & Min relaycount & Max relaycount & Min relaycount\\
     relay & where RST-MR-HC has & performance &  of & of & of & of\\
count & zero optimal & bound & Dummy & Dummy & E-SPTiRP & E-SPTiRP\\
      & solution     &  & & & & \\
 \hline
   100 & 44 & NA & 7	& 0 & 5	& 0\\
   110 & 54 & NA & 7	& 0 & 4 & 0\\
   120 & 42 & NA & 7 & 0 & 5 & 0\\
   130 & 50 & NA & 7 & 0 & 5 & 0\\
   140 & 53 & NA & 7	& 0 & 5 & 0\\
   \hline
 Total & 243 & NA & 7 & 0 & 5 & 0\\
  \hline 
\end{tabular}
\normalsize
\end{table*}

From Table~\ref{tbl:zeropt_kconnect}, we see that even in scenarios where there is no bounded factor approximation guarantee for the algorithms, the performance of the algorithms is reasonably good, with the maximum relay count being 7 relays for Dummy algorithm, and 5 relays for E-SPTiRP. The minimum relay count for both the algorithms is zero (which is clearly optimal) in those scenarios.

\section{Conclusion}
\label{sec:conclude}
In this paper, we have studied the problem of determining an optimal relay node placement strategy such that certain performance objective(s) (in this case, hop constraint, which, under a lone-packet model, ensures data delivery to the BS within a certain maximum delay) is (are) met. We studied both one connected hop constrained network design, and $k$-connected (survivable) hop constrained network design. We showed that the problems are NP-Hard, and proposed polynomial time approximation algorithms for the problems. The algorithm for one connected hop constrained network design problem, as can be concluded from numerical experiments presented in Section~\ref{sec:results}, gives solutions of reasonably good quality, using extremely reasonable computation time. 

From the numerical results presented in Section~\ref{sec:results_kconnect}, we can conclude that the algorithm proposed for the $k$-connected network design problem behaves significantly better than the worst case performance bound predicted in Theorem~\ref{thm:approxfactor} (in Section~\ref{sec:algorithms}) for the subclass of problems to which the Theorem~\ref{thm:approxfactor} applies. Even for problems where the algorithm does not have any bounded approximation guarantee, we found from our experiments that the algorithm behaves reasonably well in terms of relay count.

One might ask why the local search algorithms presented in this paper work so well in the tested random scenarios. The answer to this question is not immediately obvious, but, for the RST-MR-HC and RSN$k$-MR-HC problems, the graphs we ran our tests on were all geometric graphs; hence, a formal analysis of the properties of the underlying random geometric graph might provide some useful insights into the performance of these local search algorithms. We wish to address this issue in our future work.

Further, we are working on extending the design to traffic models more complex than the lone packet traffic model considered here. This requires the analysis of packet delays in a mesh network with more complex traffic flows and the nodes accessing the medium using CSMA/CA as defined in IEEE~802.15.4 \cite{Singh,rachitpaper}.
\section*{Acknowledgement}
This work was supported by the Automation Systems Technology (ASTeC) Center, a program of the Department of Information Technology at CDAC, Trivandrum, Kerala, India. 
\bibliography{dit_astec}
\end{document}